\documentclass[a4paper]{article}
\usepackage{easyReview}
\usepackage{amsmath,amsthm}
\usepackage{amsfonts}
\usepackage{amssymb}
\usepackage{bm}
\usepackage{mathtools}
\usepackage{mathrsfs}
\usepackage{appendix}

\usepackage{url}

\usepackage[version=4]{mhchem}
\usepackage{stmaryrd}

\usepackage{subcaption} %
\usepackage{graphicx}
\usepackage{pgfplots}
\usepackage[all]{nowidow}
\usepackage[utf8]{inputenc}
\usepackage{tikz}
\usepackage{multicol}
\usepackage{algpseudocode,algorithm,algorithmicx}
\usepackage{scalerel}
\usepackage{natbib}
\usepackage{cleveref}
\usepackage{lipsum}

\usepackage[normalem]{ulem}

\usepackage[english]{babel}

\usepackage{mathdots}
\usepackage{yhmath}
\usepackage{cancel}
\usepackage{color}
\usepackage{array}
\usepackage{multirow}
\usepackage{gensymb}
\usepackage{tabularx}
\usepackage{booktabs}
\usetikzlibrary{fadings}
\usetikzlibrary{patterns}
\usetikzlibrary{shadows.blur}

\usepackage{enumerate}
\usepackage{algorithm}
\usepackage{algpseudocode}
\usepackage{paralist}

\DeclareMathOperator*{\argmax}{arg\,max}

\DeclareMathOperator{\KL}{d_{\rm KL}}

\newcommand\wt{\widetilde}
\newcommand\wh{\widehat}

\def\m{\mathcal}
\def\dd{{\rm d}}
\def\mb{\mathbb}

\def\mf{\mathfrak}

\def\wt{\widetilde}
\def\wh{\widehat}

\newcommand\iid{\overset{\text{iid}}{\sim}}

\newcommand{\bx}{{\boldsymbol x}}

\newcommand{\by}{{\boldsymbol y}}
\newcommand{\mN}{\mathcal{N}}

\newcommand{\mfT}{\mathfrak{T}}
\newcommand{\F}{\mathcal{T}}
\newcommand{\E}{\mathrm{E}}
\newcommand{\R}{\mathbb{R}}
\newcommand{\spt}{\mathrm{spt}}

\newcommand{\abs}[1]{\left\vert#1\right\vert}
\newcommand{\norm}[1]{\|#1\|}
\newcommand{\dprod}[1]{\langle#1\rangle}

\theoremstyle{thmstyleone}%

\newtheorem{theorem}{Theorem}%
\newtheorem{lemma}{Lemma}
\theoremstyle{thmstyletwo}%
\newtheorem{remark}{Remark}%
\theoremstyle{thmstylethree}%

\usepackage[doublespacing,nodisplayskipstretch]{setspace}
\usepackage[fontsize=12pt]{fontsize}

\usepackage[nodisplayskipstretch]{setspace}

\newcommand\blfootnote[1]{%
  \begingroup
  \renewcommand\thefootnote{}\footnote{#1}%
  \addtocounter{footnote}{-1}%
  \endgroup
}

\begin{document}

	\def\spacingset#1{\renewcommand{\baselinestretch}%
		{#1}\small\normalsize} \spacingset{1}

	\title{\bf {An Optimal Transport-Based Generative Model for Bayesian Posterior Sampling}}

		\author{Ke Li,\hspace{.2cm}Wei Han,\hspace{.2cm}Yuexi Wang\hspace{.2cm}and Yun Yang\footnote{Corresponding author.}
			\blfootnote{Ke Li and Wei Han are Ph.D. students at Department of Statistics, University of Illinois at Urbana-Champaign, USA. Yuexi Wang is Assistant Professor at Department of Statistics, University of Illinois at Urbana-Champaign, USA. Yun Yang is Associate Professor at Department of Mathematics, University of Maryland at College Park, USA.   Emails: {\tt kelife94@gmail.com}, {\tt weih2@illinois.edu}, {\tt yxwang99@illinois.edu} and {\tt yy84@umd.edu}. Wang's research is partially supported by NSF grant DMS-2515542. Yang's research is partially supported by NSF grant DMS-2210717. }\\
		}
        \date{\vspace{-1.6cm}}
		\maketitle

	\bigskip
	\begin{abstract}
We investigate the problem of sampling from posterior distributions with intractable normalizing constants in Bayesian inference. Our solution is a new generative modeling approach based on optimal transport (OT) that learns a deterministic map from a reference distribution to the target posterior through constrained optimization. The method uses structural constraints from OT theory to ensure uniqueness of the solution and allows efficient generation of many independent, high-quality posterior samples. The framework supports both continuous and mixed discrete-continuous parameter spaces, with specific adaptations for latent variable models and near-Gaussian posteriors. Beyond computational benefits, it also enables new inferential tools based on OT-derived multivariate ranks and quantiles for Bayesian exploratory analysis and visualization. We demonstrate the effectiveness of our approach through multiple simulation studies and a real-world data analysis.
	\end{abstract}

	\noindent%

	\spacingset{1.65} %

\vspace{-1em}
\section{Introduction}
\label{sec:intro}

\vspace{-0.5em}
A central problem in modern Bayesian statistics is how to conduct valid statistical inference on model parameters in a computationally efficient manner. In a standard Bayesian model, the posterior distribution of the model parameter $\theta$
\vspace{-0.5em}
\begin{align}\label{eq:posterior}
\pi_n(\theta) :\,= p(\theta \,|\, Z^{(n)}) = \frac{\pi(\theta)\, p(Z^{(n)}\,|\, \theta)}{\int_{\Theta} \pi(\theta)\, p(Z^{(n)}\,| \,\theta) \,\dd \theta},\qquad\forall \theta\in\Theta\\[-3em]\notag
\end{align}
is obtained by applying the Bayes rule, where $\pi(\theta)$ denotes the prior distribution (density) over parameter space $\Theta\subset \mb R^p$, $Z^{(n)}=(Z_1,\,Z_2,\ldots,Z_n)$ is a sample of $n$ observations, and $p(Z^{(n)}\,|\, \theta)$ represents the likelihood function specified by the statistical model.
The main hurdle in implementing Bayesian inference is the multi-dimensional integral as the normalization constant in equation~\eqref{eq:posterior}, i.e., the marginal data density  $p(Z^n) = \int_{\Theta} \pi(\theta)\, p(Z^{(n)}\,| \,\theta) \,\dd \theta$. This integral is often mathematically intractable unless a conjugate prior is employed, which however significantly restricts the applicability of the Bayesian procedure to complicated models where a conjugate prior is either difficult to identify or not flexible enough to use. 

To circumvent this issue, two approaches are widely used to enable Bayesian inference without explicit knowledge of the normalizing constant. The first approach is Markov Chain Monte Carlo (MCMC) {\citep{metropolis1953equation, hastings1970monte, geman1984stochastic, robert1999monte,gelman2013bayesian}}, which generates approximate posterior samples by constructing a Markov chain whose equilibrium distribution matches the target. The second approach employs variational approximation {\citep{jordan1999introduction,wainwright2008graphical,blei2017variational}}, which reformulates the integration problem as an optimization problem. Yet, both approaches have notable limitations. MCMC algorithms often suffer from slow mixing, especially in high-dimensional or highly structured models, leading to high correlation among successive draws. On the other hand, while variational approximation scales well to large datasets, it does not guarantee exact sampling from the target density, potentially resulting in biased estimation and unreliable uncertainty quantification.

{This paper builds on the framework of transport map-based generative models for posterior sampling~\citep{el2012bayesian,marzouk2016sampling} and proposes a new class of transport maps based on optimal transport. The goal is to learn a map
$T^*$ that pushes forward i.i.d.\ samples from a simple reference
distribution $\mu$ (e.g., a standard Gaussian or the prior $\pi$) to the
target posterior distribution $\pi_n$. Once learned, $T^*$ enables the efficient generation of an arbitrary number of independent posterior draws, thereby avoiding Markov chain simulation and producing independent samples from an approximation of the target distribution, with computational efficiency comparable to variational approaches. While several classes of transport maps have been explored in \cite{el2012bayesian,hoffman2019neutra,duan2021transport,katzfuss2021scalable},
our contribution lies in the specific structural restriction imposed on the
map. In contrast to prior approaches that allow non-unique transport maps
and may suffer from computational instability and inefficiency,} we constrain $T^*$ to be the optimal transport (OT) map from $\mu$ to $\pi_n$ in the Monge formulation
\citep{santambrogio2015optimal,villani2008optimal}, namely
\vspace{-0.5em}
\begin{align}
T^* = \arg\min_{T\in\mathcal{T}}
\mathbb{E}_{X\sim\mu}\big[\|T(X)-X\|_2^2\big],\label{eqn:OPT_map}\\[-3em]\notag
\end{align}
where $\mathcal{T}=\{T:\,\pi_n=T_\#\mu\}$ denotes the class of all maps
that push forward $\mu$ to $\pi_n$.

From a modeling perspective, this OT constraint produces the \emph{most parsimonious} map that transports $\mu$ to $\pi_n$ up to a measure-preserving transformation (see Remark~\ref{rmk:brenier} and \cite{brenier1991polar}).
From a practical perspective, the estimated transport map is \emph{more interpretable}. For example, {when the prior $\pi$ is used as the reference distribution $\mu$,} the collection of transportation maps $\big\{[(1-t) id + t T^\ast]_\# \mu:\,t\in[0,1]\big\}$, where $id$ denotes the identity map, describes the shortest --- or geodesic --- path connecting the prior $\pi$ and posterior $\pi_n$ in the space of all probability distributions. In particular, the $2$-Wasserstein distance $W_2(\pi, \pi_n)$, which quantifies the discrepancy between the prior and posterior and reflects the amount of information contributed by the data, can be efficiently computed as $\big\{ \mathbb{E}_{\theta\sim\pi} \big[\|T^\ast(\theta) - \theta\|_2^2 \big] \big\}^{1/2}$, since $T^\ast$ provides the optimal transport map. Moreover, since $T^\ast$ is rank-preserving (c.f.~\Cref{sec:center_outward}), it enables new posterior summaries and inference through the computation of OT-derived multivariate ranks, as well as the construction of center-outward quantile contours and sign curves as proposed by \cite{hallin2021distribution} (see Section~\ref{sec:center_outward}), providing powerful new tools for Bayesian exploratory analysis and visualization.

Unfortunately, directly computing the optimal transport map $T^\ast$ from definition~\eqref{eqn:OPT_map} is intractable in a Bayesian setting, because the constraint $\pi_n = T_\# \mu$ is difficult to enforce analytically or computationally (see \Cref{app:constraint} for a detailed discussion). As an alternative,  {we adopt the following variational formulation,
also considered in \citep{marzouk2016sampling}:
\vspace{-1em}
\begin{align}\label{eqn:KL_Form}
\min_{T \in \F} \KL(T_\# \mu \,\| \,\pi_n),\\[-3em]\notag
\end{align}
where $\F$ is a pre-specified class of transport maps satisfying
certain structural conditions. In this work, we introduce several choices
of $\F$ derived from optimal transport theory, which ensure that
$T^\ast$ is uniquely determined. Here, $\KL(p \,\|\,q)$ denotes the Kullback–Leibler (KL) divergence between two generic distributions $p$ and $q$. Our new proposal of the class $\F$ offers several key advantages.}

First, the class $\F$ is defined independently of the target
distribution $\pi_n$ and is derived purely from a primal--dual analysis of
the OT problem \eqref{eqn:OPT_map}. This structural independence provides
both theoretical clarity and practical convenience. When $\theta$ contains
only continuous variables (Section~\ref{sec:continuous}), $\F$ can
be taken as the set of gradients of convex functions, in which case the
inverse map $(T^\ast)^{-1}$ can be computed efficiently via convex
optimization, a property useful for constructing center-outward ranks (see
\Cref{sec:center_outward}). When $\theta$ is mixed, containing both
continuous and discrete components as in Bayesian latent variable models,
$\F$ admits a more intricate but still analytically tractable
characterization (cf.~\Cref{thm:discrete_case} in
Section~\ref{sec:mixed_case}). Second, although an equivalent formulation of \eqref{eqn:KL_Form} avoids
explicit evaluation of the density of $T_\# \mu$ and the normalization
constant $p(Z^n)$, it relies on the invariance of the KL divergence under
invertible transformations (cf.~Lemma~\ref{prop:equiv}). Our structural
conditions on $\F$ ensure that every $T\in\F$ is
invertible, so the equivalent formulation remains valid throughout the
optimization.

Building on this characterization, we adapt the OT map $T^\ast$ to the
Bayesian inference setting in several ways. We give explicit forms of
$T^\ast$ for mixed discrete--continuous parameters, extending the classical
semi-discrete OT framework. For regular parametric models with
approximately Gaussian posteriors, restricting $\F$ to affine maps
preserves statistical accuracy while greatly reducing computation. For
Bayesian latent variable models, we further incorporate mean-field ideas
from variational inference \citep{blei2017variational,zhang2024bayesian} to
impose additional structure on $\F$, improving scalability without
sacrificing modeling flexibility or accuracy.

The paper is organized as follows. \Cref{sec:background} provides a general overview of generative model learning and some optimal transport theory that forms the foundation of our approach. In \Cref{sec:method}, we introduce our methodology for posteriors over continuous and mixed parameters. \Cref{sec:applications} presents three applications of our method: the first two involve specialized approximation classes, while the third highlights the inferential capabilities enabled by our framework. We evaluate the performance of our method against existing transport map approaches using simulated examples in \Cref{sec:simulation}, and conduct a real-data analysis of the \textit{yeast} dataset in \Cref{sec:real_data}. A concluding discussion is provided in \Cref{sec:discussion}.

\noindent
{\bf Notation.} We write $\spt(P)$ to denote the support of a probability distribution $P$. The inner product $\langle \cdot, \cdot \rangle$ is defined by $\langle a, b \rangle = \sum_{i=1}^d a_i b_i$ for vectors $a, b \in \mathbb{R}^d$, {and the Euclidean norm $\norm{\cdot}_2$ is defined by $\norm{a}_2 = \sqrt{\langle a, a \rangle}$}. The abbreviation i.i.d.~stands for ``independently and identically distributed".
We denote the local H\"{o}lder space of smoothness level $k + \beta$ over a bounded open set $A \subset \mathbb{R}^p$ by $C_{\text{loc}}^{k,\beta}(A)$, where $\beta \in (0,1)$ and $k$ is a nonnegative integer. That is, $u \in C_{\text{loc}}^{k,\beta}(A)$ if, for every compact subset $V \subset A$, we have $u \in C^{k,\beta}(V)$.
The effective domain of a function $u$ (often referred to simply as its domain) is defined as $\mathrm{Dom}(u) := \big\{x \in \mathbb{R}^p \mid |u(x)| < \infty\big\}$. The closure of a set $\Omega$ is denoted by $\overline{\Omega}$. We use $T_\# \mu$ to denote the pushforward measure of $\mu$ under $T$, defined by $T_\# \mu(B) = \mu(T^{-1}(B))$ for any measurable set $B$. A set $S$ is called a $(p-1)$-dimensional surfaces of class $C^2$ if for every point $x \in S$, there exists an opening neighborhood $U \subset \R^p $ of $x$ and a function $g \in C^2(U)$ such that $S \cap U = \{y \in U : g(y) = 0\}$ and $\nabla g(y) \neq 0$ for all $y \in S \cap U$.

\vspace{-0.5cm}
\section{Generative Models and Optimal Transport Theory}\label{sec:background}
In this section, we begin by introducing the concept of generative models and discuss how they can be adapted for posterior approximation. We then provide a brief overview of optimal transport theory and present the formal definition of the optimal transport map.

\vspace{-0.2cm}
\subsection{Generative model approaches for Bayesian sampling}\label{sec:generative_model}
{In the machine learning literature, generative models, such as Generative Adversarial Network (GAN,~\cite{goodfellow2014generative, li2015generative,10.1214/19-AOS1858}), Wasserstein GAN (WGAN,~\cite{arjovsky2017wasserstein}) and Wasserstein Auto-Encoder (WAE,~\cite{tolstikhin2019wasserstein, zhao2018infovae}), have received great success in generating synthetic realistic-looking images and texts~\citep{brock2018large,oord2016wavenet} by implicit distribution estimation over complex data space bypassing the need of evaluating probability densities. Generative model based approaches for statistical inference~\citep{mohamed2016learning,huszar2017variational} avoid explicitly specifying the likelihood function while utilizing the expressive power of neural networks to capture underlying probabilistic structures. In our setting of posterior sampling, the generative model operates over the parameter space $\Theta \subset \mathbb{R}^p$, which may also include latent variables.}

{In this paper, we adopt the transport map-based generative model framework of
\citep{el2012bayesian,marzouk2016sampling} for Bayesian posterior sampling, which aims to learn a transport map $T^\ast$ by solving the
variational problem~\eqref{eqn:KL_Form}.
In contrast to standard variational inference \citep{blei2017variational}, directly solving \eqref{eqn:KL_Form} is generally computationally infeasible, since evaluating
the KL divergence requires the explicit density of the pushforward
distribution $T_\#\mu$, which typically depends on the inverse map
$T^{-1}$ when it exists and is tractable to compute.
A key observation is that the KL divergence in \eqref{eqn:KL_Form} is invariant
under common invertible transformations applied to both $\mu$ and $\pi_n$.
Exploiting this property leads to an equivalent objective that avoids
computing $T^{-1}$. This idea also appears earlier in the training of
normalizing flows \citep{rezende2015variational} and was later adopted in the
transport map framework in~\cite{marzouk2016sampling}. We restate the result in the following lemma and provide a proof in \Cref{sec:equiv_pf} for completeness.}

\begin{lemma}\label{prop:equiv}
Assume that each element $T\in\m T$ is differentiable and invertible on the support of $\mu$. Then optimization problem in \eqref{eqn:KL_Form} is equivalent to 
\vspace{-0.5em}
\begin{align}\label{eqn:new_obj}
\max_{T \in \m T} \E_\mu \big[ \log ( \pi_n \circ T)  + \log |\det(J_T)|\big],\\[-3em]\notag
\end{align}
where $J_T\in\mb R^{p\times p}$ denotes the Jacobian matrix associated with map $T$ from $\mb R^p$ to $\mb R^p$ and $\det (A)$ denotes the determinant of a square matrix $A$. %
\end{lemma}

\vspace{-0.3cm}
There is an intuitive interpretation of the objective function in \eqref{eqn:new_obj} of \Cref{prop:equiv}. The first term, $\log(\pi_n \circ T)$, can be viewed as a goodness-of-fit term that encourages alignment with the target distribution $\pi_n$, while the second term serves as a regularization term that penalizes the ``roughness'' of the transport map $T$. More specifically, for each $X \sim \mu$, the quantity $\lvert \det(J_T(X)) \rvert$ measures the local volume distortion induced by the Jacobian map $J_T(X): \mathbb{R}^p \to \mathbb{R}^p$ around $X$. As shown later through the Monge--Amp\`ere equation \eqref{eq:MA} (see \Cref{lem:regularity_OT}), the optimal transport map $T^*$ satisfies $\lvert \det(J_{T^*}) \rvert = \frac{\mu}{\pi_n \circ T^*}$. Substituting this relation into the objective in \eqref{eqn:new_obj} shows that the optimal value coincides with the negative differential entropy of the reference distribution $\mu$, namely $-h(\mu) = \mathbb{E}_{\mu}[\log \mu]$.

\subsection{Review of optimal transport theory}
In this subsection, we briefly review some theoretical background and fundamental results of optimal transport (OT) theory that are necessary before introducing our development. We focus on the OT problem with quadratic cost, commonly known as the Kantorovich Problem (KP), defined as
\vspace{-0.5em}
\begin{align*}
   \mbox{(KP)}\qquad \inf_{\gamma \in \Pi(\nu_1, \nu_2)} \int \|x - y\|^2\, \dd \gamma(x, y),\\[-3em]\notag
\end{align*}
where $\Pi(\nu_1, \nu_2)$ is the set of all couplings between two probability distributions, $\nu_1$ over $\m X\subset\mb R^p$ and $\nu_2$ over $\m Y\subset \mb R^p$. That is, any probability measure $\gamma$ belongs to $\Pi(\nu_1, \nu_2)$ if and only if its marginal distributions are  $\nu_1$ and $\nu_2$, respectively.
The Kantorovich problem always admits a solution \citep[Theorem 1.7]{santambrogio2015optimal}, referred as the \emph{optimal transport plan}. The optimal objective value of KP is defined as the squared 2-Wasserstein distance between $\nu_1$ and $\nu_2$, denoted as $W_2^2(\nu_1,\,\nu_2)$. In particular, when one of the distributions, say $\nu_1$, is absolutely continuous with respect to the Lebesgue measure on $\mb R^p$, the first part of the lemma below \citep[Theorem 1.22]{santambrogio2015optimal} shows that, under some mild conditions, the optimal transport plan is unique and takes the form of $(id,\, T^*)_\# \nu_1$, implying $\nu_2= T^*_\#\nu_1$~\citep[also see, e.g.,][]{brenier1991polar,mccann1995existence}.

\begin{lemma}[OT map existence and regularity]\label{lem:regularity_OT}
Suppose $\nu_1$ and $\nu_2$ both have finite second moments, and $\nu_1$ gives no mass to $(p-1)$ surfaces of class $C^2$, then we have:
\newline\noindent
{\bf 1.} An optimal transport plan solving (KP) exists and is unique. It admits the form of $(id,\, T^*)_\# \nu_1$ with $T^* = \nabla u^*$, where $u^*$ is a convex function. The converse is also true: if there exists a convex function $u$ such that $[\nabla u^*]_\# \nu_1=\nu_2$, then $(id,\, T^*)_\# \nu_1$ with $T^* = \nabla u^*$ is the unique solution of  (KP).
\newline\noindent
{\bf 2.}  If both $\nu_1$ and $\nu_2$ are absolutely continuous with respect to the Lebesgue measure of $\mb R^p$, with density functions denoted by $\nu_1$ and $\nu_2$ as well, then $T^*$ is differentiable $\nu_1$-a.e., and solves the following Monge-Ampere equation:
\vspace{-0.5em}
\begin{equation}\label{eq:MA}
|\det ( J_{T^*})| = \frac{\nu_1}{\nu_2\circ T^*}\,,
\qquad \mu\text{-a.e.}\\[-2.5em]\notag
\end{equation}
\newline\noindent
{\bf 3.} Assume that $\nu_1$ and $\nu_2$ are supported on bounded open sets $S_1$ and $S_2$ in $\mathbb{R}^d$ respectively, with $S_2$ being convex. Then the smoothness of $T^*$ is always one degree higher than the density functions, in the sense that if $\mu \in C_{\text{loc}}^{k, \beta}(S_1)$ and $\pi_n \in C_{\text{loc}}^{k, \beta}(S_2)$, then $T^* \in C_{\text{loc}}^{k + 1, \beta}(S_1)$. This regularity extends to the closure $\bar{S}_1$ and $\bar{S}_2$.
\end{lemma}
{\noindent The function $u^*$ in the lemma is referred to as the optimal potential function. The first half of item~1 follows from Theorem~1.22 of~\citep{santambrogio2015optimal}; the converse in the second half is not stated explicitly in the literature, but follows from keeping track of the duality gap in the proof of Theorem~1.22 (see also the proof of Theorem~\ref{thm:discrete_case} in Appendix~\ref{sec:proof_theorem1}).
Item~2 is taken from Section~1.7.6 of the same reference. Item~3 follows from Caffarelli’s regularity theory~\citep{villani2003topics, caffarelli1992boundary, caffarelli1992regularity}; several versions of this result exist, and here we state a convenient sufficient condition from~\cite[Theorem~3.1--3.3]{de2014monge}.}

Another commonly used formulation of the optimal transport is the Monge problem (MP), which seeks a transport map $T^*$ that minimizes the expected quadratic transportation cost while exactly pushing forward $\nu_1$ to $\nu_2$:
\begin{align*}
   \mbox{(MP)}\qquad \inf_{T} \int \|x - T(x)\|^2 \,\dd \nu_1(x), \quad \mbox{s.t.}\quad T_\#\nu_1=\nu_2.
\end{align*}
Under the regularity conditions stated in \Cref{lem:regularity_OT}, problem (MP) admits a unique solution, referred to as the \emph{optimal transport map}. However, it is important to note that (MP) may have no solution if these regularity conditions are not satisfied; see, for example, Section 1.4 in \cite{santambrogio2015optimal}.
Moreover, if $\nu_2$ assigns no mass to any $(p-1)$-dimensional $C^2$ surface, then the optimal transport map $T^*$ is invertible, and there exists a unique convex function $u^{\dagger}$ such that $[\nabla u^{\dagger}]_\# \nu_2 = \nu_1$. This function $u^{\dagger}$ is related to the potential $u^*$ through the relation $\nabla u^{\dagger} = (\nabla u^*)^{-1}$, $\mu$-almost everywhere. We review this duality and the role of convex conjugates in \Cref{sec:potential_conjugate}.

\begin{remark}[Multivariate extension of monotonicity] In the one-dimensional case, the condition that a transport map $T$ is the gradient of a convex function is equivalent to $T$ being non-decreasing. In the multi-dimensional setting, the characterization of $T$ as the gradient of a convex function generalizes this notion of monotonicity. Specifically, it ensures that a certain multivariate extension of quantiles~\citep{chernozhukov2017monge} is well-defined (with the quantile ordering preserved through the mapping), and the transportation paths $\big\{\gamma_x(t) = (1 - t)x + t\, T(x) :\, t \in [0,1]\big\}$ are non-crossing for different $x \in \mathcal{X}$, a property known in optimal transport theory as the Monge–Mather shortening principle \citep[Chapter 8]{villani2008optimal}.
Leveraging this property, we show in \Cref{sec:center_outward} that when $\mu$ is chosen to be the standard multivariate Gaussian distribution $\m N(0,I_p)$, the center-outward ranks, quantile contours, and sign curves proposed by \cite{hallin2021distribution} can be constructed by mapping quantile spheres of $\m N(0,I_p)$ through the optimal transport map.
\end{remark}

\begin{remark}[OT map is the most parsimonious map]\label{rmk:brenier} Brenier’s polar factorization theorem \citep{brenier1991polar} states that any transport map $T$ that pushforwards $\nu_1$ to $\nu_2$, that is, $\nu_2 = T_\# \nu_1$, can be factorized as $T = T^* \circ R$, where $T^*$ is the optimal transport map described in \Cref{lem:regularity_OT}, and $R$ is a measure-preserving map with respect to $\nu_1$, meaning $R_\# \nu_1 = \nu_1$.
This decomposition highlights that the optimal transport map $T^*$ is the \emph{most parsimonious} way to transport $\nu_1$ to $\nu_2$, corresponding to the special case where $R = \mathrm{id}$ is the identity map. While any transport map can be expressed as the composition of $T^*$ with some measure-preserving map $R$, the converse does not necessarily hold, as $R$ may not be invertible.
Consequently, the optimal transport map can also be seen as the ``minimal” transport map, analogous to the role of a minimal sufficient statistic in statistical inference, capturing only the essential transformation required to map $\nu_1$ to $\nu_2$.
\end{remark}

\Cref{lem:regularity_OT} can be tailored to a special setting in which $\nu_2$ is supported on a union of disjoint sets that are strongly separable, that is, each pair can be separated by a hyperplane. In this case, one can derive a more specialized characterization of the optimal potential function. For example, \cite[Theorem 5.1]{kitagawa2019free} show that when $\mathrm{spt}(\nu_2) = \overline{\Omega}_1 \cup \overline{\Omega}_2$, where $\Omega_k \subset \mathbb{R}^p$ ($k = 1, 2$) are two disjoint convex subsets, the optimal convex potential $u^*$ associated with the OT map $T^*$ takes the form $u^* = \max\{u^*_1, u^*_2\}$. Here, each $u^*_k$ is a convex function whose gradient $\nabla u^*_k$ maps points to $\overline{\Omega}_k$, for $k \in \{1, 2\}$.
Moreover, this result can be further generalized to the case where $\mathrm{spt}(\nu_2)$ consists of finitely many disjoint and strongly separable regions. Specifically, \cite[equation (5.5)]{kitagawa2019free} provide such an extension, which we summarize in the following lemma.

\begin{lemma}[Target distribution with disconnected support from \cite{kitagawa2019free}]\label{lem:regularity_OT_2}
Under the same conditions as in \Cref{lem:regularity_OT}, and additionally assuming that $\mathrm{spt}(\nu_2) = \bigcup_{k \in I} \overline{\Omega}_k$ for a finite index set $I$, where each $\overline{\Omega}_k$ can be separated from every other $\overline{\Omega}_\ell$ ($\ell \neq k$) by a hyperplane, the optimal potential function $u^*$ associated with the OT map $T^\ast$ from $\nu_1$ to $\nu_2$ can be expressed as
\vspace{-0.7em}
\begin{align*}
u^*(x) = \max_{k\in I} u^*_k(x), \quad \mathrm{with } \quad \nabla u^*_k(x) \in \overline{\Omega}_k, \quad \forall x \in \mathrm{Dom}(\nabla u^*).
\end{align*}
\end{lemma}

\noindent
The setting in which $\nu_2$ is supported on a union of disjoint sets is particularly relevant in Bayesian contexts where the posterior distribution may exhibit multimodality with well-separated local modes \citep[e.g.,~label switching in Bayesian mixture models,][]{10.1214/088342305000000016}, where each component potential $u^*_k$ can be interpreted as the local potential function corresponding to a local mode or region. 
Under this setting, for each draw $X$ from the reference distribution $\mu$, the optimal map effectively selects a local region via the index $k(X) = \argmax_{k \in I} u^*_k(X)$, and then maps $X$ to $T^*(X) = \nabla u^*_{k(X)}(X)$ using the corresponding local potential $u^*_{k(X)}$ (c.f.~Lemma~\ref{lem:max_jacobian}). Even when the target posterior $\pi_n$ is not multimodal, one may still choose to represent the global potential $u^*$ as the maximum over several convex functions, as motivated by Lemma~\ref{lem:regularity_OT_2} in Appendix~\ref{sec:computation_J}, to enhance modeling flexibility and approximation capacity. This leads to the construction of the transport map class $\mathcal{T}$ used in Section~\ref{sec:continuous}.

\section{Optimal Transport for Bayesian Sampling}\label{sec:method}
In this section, we study how to learn the optimal transport map $T^*$ for Bayesian sampling, where $T^*$ solves the Monge problem (MP) with $\nu_1=\mu$ and $\nu_2=\pi_n$. Once $T^*$ is learned, samples $\{\theta_i\}_{i=1}^N$ from the posterior $\pi_n$ are generated by drawing i.i.d.\ samples $X_i\sim\mu$ and setting $\theta_i=T^*(X_i)$. A key distinction from existing transport-based methods is that we approximate only the \emph{optimal} transport map $T^*$, rather than an arbitrary map $T\in\mathcal T$ satisfying $T_\#\mu=\pi_n$. This removes the non-uniqueness of general transport maps and leads to a solution that is stable, efficient, and interpretable. As noted in \Cref{rmk:brenier}, any other transport map can be recovered from $T^*$, so no flexibility is lost. To achieve this, we adopt a constrained optimization approach by imposing structural restrictions on the transport class $\F$ when minimizing the KL objective~\eqref{eqn:KL_Form}. We exploit structural properties of $T^*$ given in Lemma~\ref{lem:regularity_OT} when $\theta$ is continuous, and extend this characterization to mixed settings where $\theta$ contains discrete components, as in many Bayesian latent variable models. In such cases, although item~1 of Lemma~\ref{lem:regularity_OT} remains valid, Lemma~\ref{prop:equiv}, which makes problem~\eqref{eqn:KL_Form} computationally tractable, no longer applies because $T^*$ may not be invertible (see Section~\ref{sec:mixed_case}).

\vspace{-0.3cm}

\subsection{Optimal transport map with continuous parameters}\label{sec:continuous}
We choose a reference distribution $\mu$ that is absolutely continuous with respect to the Lebesgue measure on $\R^p$, for example, $\mathcal N(0, I_p)$. Under this choice, Lemma~\ref{lem:regularity_OT} ensures that the optimal transport map $T^*$ is unique and admits the representation $T^*=\nabla u^*$ for a convex function $u^*$. The converse also holds: if a convex function $u^*$ satisfies that $T^*=\nabla u^*$ pushes $\mu$ forward to $\pi_n$, then $T^*$ is the unique optimal transport map.

When $\theta$ is continuous in $\R^p$, the posterior $\pi_n$ has a continuous density. By items~2 and~3 of Lemma~\ref{lem:regularity_OT}, the optimal map $T^*$ is differentiable and invertible, implying that the convex potential $u^*$ is at least twice differentiable. Consequently, $u$ can be characterized through its Hessian $\nabla^2 u = J_T$, which must be positive semidefinite. Moreover, the Monge--Amp\`ere equation in Lemma~\ref{lem:regularity_OT} implies that $J_T$ is non-singular on $\mathrm{supp}(\mu)$, so $\nabla^2 u$ is in fact positive definite, meaning that $u$ is strongly convex and $T$ is globally invertible. These observations motivate the specification of the transport map class in~\eqref{eqn:new_obj} as
\vspace{-0.5em}
\begin{equation}\label{eq:invertible_class}
\mathcal T_{\rm conv}
=
\big\{\nabla u \,\big|\,
u:\mathrm{spt}(\mu)\to\Theta,\;
u\in C^2(\mathrm{spt}(\mu)),\;
\nabla^2 u \in \R^{p\times p}\ \text{is p.d.}
\big\},\\[-0.5em]\notag
\end{equation}
which ensures that $T^*$ is the unique element in $\mathcal T_{\rm conv}$ satisfying $T_\#\mu=\pi_n$.

Since the transport class $\mathcal T_{\rm conv}$ satisfies the assumptions of \Cref{prop:equiv}, Lemma~\ref{lem:regularity_OT} implies that $T^*$ is the unique solution to the computationally tractable formulation~\eqref{eqn:new_obj} within this class. 
In practice, we adopt a Monte Carlo approximation and choose a parametric class $\mathfrak T$ that approximates $\mathcal T_{\rm conv}$, leading to the practical estimator
\begin{equation}\label{eq:That}
\widehat T
=
\arg\min_{T\in \mathfrak T}
\frac{1}{M}\sum_{m=1}^M
\Big[
\log \widetilde{\pi}_n\big(T(X_m)\big)
+
\log \det \big(J_T(X_m)\big)
\Big],
\end{equation}
where $\{X_m\}_{m=1}^M$ are i.i.d.\ samples from $\mu$, and $\widetilde{\pi}_n$ denotes a tractable but unnormalized form of the posterior density $\pi_n$.

\noindent {\bf Construction of approximating class $\mathfrak T$.}
{We adopt an approximating family motivated by the structural properties of optimal transport maps for multimodal distributions described in Lemma~\ref{lem:regularity_OT_2}. In \Cref{sec:review_other_methods} and \Cref{sec:additional_quantile}, we further discuss and compare this choice with other popular transport maps, including those based on input convex neural networks (ICNNs; \cite{amos2017input}) and triangular maps~\cite{el2012bayesian,marzouk2016sampling}.}
Specifically, to construct the approximating class $\mathfrak{T}$, we first introduce a class of convex building blocks, referred to as convex units, motivated by the following lemma.
\vspace{-0.3cm}
\begin{lemma}\label{lem:splineConvexity}
For constants $\alpha, \beta \in \mathbb{R}^{p}$ and $w, v \in \mathbb{R}$, and an increasing and bounded univariate function $\varphi$ defined on $\mathbb{R}$, the following function defined by
\begin{align}\label{eq:activation}
f(x;\,\alpha, \beta, w, v) = \int^{\langle \alpha, x \rangle + w}_{-\infty} \varphi(u)\, \dd u  + \langle \beta, x \rangle + v
\end{align}
is a convex function of $x\in\R^p$.
\end{lemma}
\vspace{-0.3cm}
\noindent The proof is provided in \Cref{sec:splineConvexity_pf}. In our implementation, we adopt the following choices for $\varphi$: 1.~the Tanh function: $\varphi(x) =\frac{e^x -e^{-x}}{e^{x}+e^{-x}}$; 2.~the Softsign function: $\varphi(x) = \frac{x}{1 + |x|}$; 3.~the Square Nonlinearity (SQNL) function: $\varphi(x) = \text{sign} (x) 1_{\abs{x} >2} +\big(x-\text{sign} (x) \frac{x^{2}}{4}\big) 1_{\abs{x} \leq 2}$.
In particular, the parameter $\alpha > 0$ in the convex unit controls the degree of nonlinearity to enhance flexibility: over any compact set $\mathcal{C} \subset \mathbb{R}^p$, a smaller $\alpha$ makes the convex unit closely resemble an affine transformation for $x \in \mathcal{C}$. Adding an affine term $\langle \beta, x \rangle + v$ preserves the convexity of the function $g(x)$, while further increasing its expressiveness. %

With the convex unit, we can now define our family $\m U_{L,M}$ of convex functions $u$ to approximate the optimal convex potential $u^*$ as 
\begin{align}
    \m U_{L,M}=\bigg\{u:&\,\R^p \to \R\,\bigg|\, u(x)= \max_{k\in[L]} u_k(x), \, u_k(x)=\sum_{m=1}^Mf\big(x;\,\alpha^{(k)}_{m}, \beta^{(k)}_{m}, w^{(k)}_{m}, v^{(k)}_{m}\big), \label{eq:max_of_convex}\\
    &\alpha^{(k)}_{m}\in\R^p, \beta^{(k)}_{m}\in\R^p, w^{(k)}_{m}\in\R, v^{(k)}_{m}\in\R, \ \mbox{for all } k\in[L] \mbox{ and } m\in[M]  \bigg\},\notag
\end{align}
where $L$ is the number of local potential functions $u_k$ in the maximum, and $M$ is the number of convex units used to define each $u_k$. The corresponding transport map class $\mfT$ used in \eqref{eq:That} is then defined as $\mfT = \big\{T:\,\R^p\to\R^p\,\big|\,T=\nabla u(x),\, u\in \m U_{L,M}\big\}$.

A few observations are in order. First, since sums and maxima of finitely many convex functions remain convex, every function in $\mathcal{U}_{L,M}$ is convex. Second, any convex function on $\mathbb{R}^p$ can be approximated arbitrarily well by elements of $\mathcal{U}_{L,M}$ for sufficiently large $L$ and $M$, because the union $\bigcup_{L\ge 1,\,M\ge 1}\mathcal{U}_{L,M}$ contains all max-affine functions, which are known to be universal approximators of convex functions \citep{magnani2009convex}. Third, the use of the maximum operator is motivated by Lemma~\ref{lem:regularity_OT_2}: when the target distribution has multiple well-separated clusters, the optimal potential $u^*$ naturally takes the form of a maximum of several local potentials, each associated with a cluster. We discuss a data-driven strategy for selecting $L$ in \Cref{sec:choice_L}. Without the max operation, the number of convex units required to approximate $u^*$ may grow exponentially with the number of clusters. Finally, Lemma~\ref{lem:max_jacobian} in Appendix~\ref{sec:computation_J} provides an efficient way to compute the transport map $T=\nabla u=\nabla \max_{k\in[L]} u_k$ and its Jacobian, both required for implementing gradient-based algorithm for computing~\eqref{eq:That}.

\vspace{-0.3cm}

\subsection{Optimal transport map with mixed parameters}\label{sec:mixed_case}
In a more general setting where $\theta$ consists of both continuous and discrete components, referred to as mixed parameters or variables, we decompose $\theta=(\tau,\,\zeta)$, where $\tau\in[K]:\,=\{1,2,\ldots,K\}$ denotes the discrete component with $K$ levels, and $\zeta\in\mb R^p$ represents the continuous component of dimension $p$. 

{\begin{remark}[Embedding of discrete components]\label{rmk:embedding}
Since $\tau$ is discrete, we introduce embeddings $\{b_k\}_{k=1}^K \subset \mathbb{R}^r$ to encode the categorical information. The choice depends on the nature of $\tau$. For a nominal variable, a common choice is one-hot encoding, where each category $k$ is represented by the standard basis vector $b_k \in \mathbb{R}^K$ with a $1$ in the $k$-th coordinate and $0$ elsewhere ($r=K$). If $\tau$ is ordinal, we choose $\{b_k\}_{k=1}^K$ as an increasing sequence in $\mathbb{R}$ ($r=1$), where $b_{k+1}-b_k$ reflects the gap between levels $k$ and $k+1$. A detailed discussion of embedding choices is given in \Cref{sec:embedding_choice}.
\end{remark}}

With the discrete embedding into $\mathbb{R}^r$, we accordingly decompose $x\sim \mu$ as $x=(x^{(1)},x^{(2)})\in\mathbb{R}^r\times\mathbb{R}^p$ and write the transport map as $T=(T^{(1)},T^{(2)})$, where $T^{(1)}:\mathbb{R}^r\times\mathbb{R}^p\to [K]$ maps to the discrete component and $T^{(2)}:\mathbb{R}^r\times\mathbb{R}^p\to\mathbb{R}^p$ maps to the continuous component. When the reference distribution $\mu$ admits a density (e.g.~$\mathcal{N}(0,I_{r+p})$), Lemma~\ref{lem:regularity_OT} guarantees that the optimal transport (OT) problem is equivalent to the Monge problem with mixed variables, formulated as follows:
\vspace{-0.5em}
\begin{align*}
     \mbox{(MPm)}\quad  &\inf_{T=(T^{(1)},T^{(2)})} \int_{\mb R^r\times \mb R^p} \Big[ \ \ \kappa  \underbrace{\|x^{(1)} - b_{T^{(1)}(x)}\|^2}_{\mbox{discrete part cost}} \ + \underbrace{\|x^{(2)} - T^{(2)}(x)\|^2}_{\mbox{continuous part cost}}\Big]\ \ \dd \mu(x), \\
     & \mbox{s.t.}\quad T_\#\mu=\pi_n \quad \mbox{and}\quad T^{(1)}(x) \in[K],  \quad \mbox{for all } x=(x^{(1)},x^{(2)}) \in\mb R^r\times \mb R^p,\\[-3em]\notag
\end{align*}
where $\kappa>0$ controls the relative weight of the discrete cost; throughout the paper we take $\kappa=1$. Because the posterior $\pi_n$ has singular components induced by discreteness, items~2 and~3 of Lemma~\ref{lem:regularity_OT} no longer apply. In particular, for each fixed $x^{(2)}\in\mathbb{R}^p$, the convex function $u(\cdot,x^{(2)})$ becomes piecewise linear over $\mathbb{R}^r$ (see Theorem~\ref{thm:discrete_case}), so the optimal map $T=\nabla u$ is no longer invertible. As a result, \Cref{prop:equiv} is not applicable. We therefore seek a finer characterization of $T^*$ and its potential $u^*$ that preserves computational tractability. The next lemma provides such a characterization and suggests modeling the discrete and continuous components separately.

\vspace{-0.3cm}

\begin{theorem}\label{thm:discrete_case}
Suppose the reference distribution $\mu$ is absolutely continuous on $\mathbb{R}^r\times\mathbb{R}^p$, and the target distribution $\pi_n$ is defined over $\theta=(\tau,\zeta)$ with $\tau\in[K]$ and $\zeta\in\mathbb{R}^p$. Then the optimal potential function $u^*$ associated with $T^*=(T^{*(1)},T^{*(2)})$ in problem (MPm) takes the following form:
\vspace{-0.5em}
\begin{align*}
    u^*\big(x^{(1)},x^{(2)}\big) = \kappa \max_{k\in [K]} \big\{\langle x^{(1)},\, b_k\rangle + \phi_k(x^{(2)}) \big\},
\end{align*}
where $\kappa$ is the weighting parameter in (MPm), $b_k\in\mathbb{R}^r$ is the embedding vector for category $k$ defined in \Cref{rmk:embedding}, and $\phi_k:\mathbb{R}^p\to\mathbb{R}$ are convex functions. The corresponding OT map $T^*=(T^{*(1)},\,T^{*(2)})$ is
\vspace{-0.5em}
\begin{align}
\label{eq:dis_trans_map}
\left\{\begin{array}{ll}
\tau &= T^{*(1)}\big(x^{(1)},\,x^{(2)}\big) = \argmax_{k\in[K]} \left\{ \langle x^{(1)}, \, b_k\rangle + \phi_k(x^{(2)})\right\}, \\
\zeta &= T^{*(2)}\big(x^{(1)},\,x^{(2)}\big) = \kappa \nabla \phi_\tau(x^{(2)}).
\end{array}\right.
\end{align}
Furthermore, if for each category $k\in[K]$ the conditional posterior $\pi_n(\zeta\mid \tau=k)$ is absolutely continuous with respect to the Lebesgue measure on $\mathbb{R}^p$, then $\phi_k$ is at least twice differentiable and strongly convex, so $\nabla\phi_k$ is differentiable and invertible.
\end{theorem}

\vspace{-0.3cm}

\noindent This result shows that $u^*$ is piecewise linear in $\mathbb{R}^r$ due to the linear term $\langle x^{(1)}, b_k\rangle$. This structure induces probability mass accumulation in the discrete component $\tau$ and makes $T^*$ non-invertible. When $\mu$ is a product measure on $\mathbb{R}^r\times\mathbb{R}^p$, the map $T_k^*:=\nabla\phi_k$ in \Cref{thm:discrete_case} can be interpreted as the optimal transport map that pushes the $\mathbb{R}^p$ marginal of $\mu$ to the conditional posterior $\pi_n(\zeta\mid\tau=k)$. The overall map $T^*$ can therefore be viewed as a location-dependent mixture of the local maps $\{T_k^*\}_{k=1}^K$.
\cite{duan2021transport} also construct random transport maps as mixtures of maps $\{T_i\}_{i=1}^N$, motivated by Bayesian nonparametric approximations of conditional densities \citep{dunson2007density}. In their approach, each point is transported by $T_i$ with a location-dependent probability $\omega_i$. However, since their maps are restricted to location-scale transformations, representing complex targets is costly and the method scales poorly in the presence of discrete variables.
\cite{el2012bayesian,marzouk2016sampling} propose lower-triangular transport maps, which require invertibility of the map. This assumption is violated in the mixed-variable setting considered here.

In the special case where $\mu=\mathcal{N}(0,I_{r+p})$, the first equation in \eqref{eq:dis_trans_map} resembles a multi-class Probit model. Conditional on $x^{(2)}$, the latent score for class $k$ is $\langle x^{(1)}, b_k\rangle+\phi_k(x^{(2)})$, and $\langle x^{(1)}, b_k\rangle+\phi_k(x^{(2)}) \mid x^{(2)} \sim \mathcal{N}\big(\phi_k(x^{(2)}), \|b_k\|^2\big)$, with the label assigned to the class achieving the largest score. The continuous component $\zeta$ is then generated by pushing the marginal reference distribution $\mathcal{N}(0,I_p)$ through the map $T^*_\tau$, yielding samples from $\pi_n(\zeta\mid\tau)$.
This factorized structure motivates the following lemma, analogous to \Cref{prop:equiv}, which enables efficient computation in the mixed-variable setting.
\vspace{-0.2cm}
\begin{lemma}[Optimization objective for mixed parameters]\label{lem:mixture_KL}
Let $T: \mathcal{X} \rightarrow \mathcal{Y}$ be a transport map that pushforwards $\mu$ to $\pi_n$, where $\mathcal{X} = \mathcal{X}_1 \times \mathcal{X}_2 :\,= \mb R^r\times \mb R^p$ and $\mathcal{Y} = \m Y_1 \times \m Y_2 :\,= [K]\times \mb R^p$. Define the intermediate transport map $\wt{T}: \mathcal{Y}_1\times \mathcal{X}_2 \rightarrow \mathcal{Y}_1 \times \mathcal{Y}_2$ induced by $T$ as
\vspace{-0.5em}
\begin{equation*}
    \left(\begin{array}{c}
    \tau \\
    x^{(2)}
    \end{array}\right)\mapsto 
    \left(\begin{array}{c}
    \tau \\
    \zeta = \kappa \nabla \phi_{\tau}(x^{(2)})
    \end{array}\right).\\[-0.5em]\notag
\end{equation*}
Then, the optimization problem in \eqref{eqn:KL_Form}, with $\mathcal{T}$ chosen as the class of transport maps taking the form of~\eqref{eq:dis_trans_map} in Theorem~\ref{thm:discrete_case}, is equivalent to:
\vspace{-0.5em}
\begin{align}
\label{eq:mixture_KL_opt}
    \min_{T\in \mathcal{T}} \mathbb{E}_{(X^{(1)}, X^{(2)})\sim \mu} \left[\log \mathbb{P}(\tau \mid X^{(2)}) - \log \pi_n\circ \wt{T}(\tau, X^{(2)}) - \log|\det(J_{\wt{T}}(\tau, X^{(2)}))|\right],
\end{align}
 with $\tau = \argmax_{k\in[K]} \big\{\dprod{X^{(1)}, \, b_k} + \phi_k(X^{(2)})\big\}$.
\end{lemma}
\noindent To evaluate the objective function in the optimization problem~\eqref{eq:mixture_KL_opt} for a given transport map $T$, one apply the Monte Carlo method by approximating the expectation with an empirical average; additional computational details are provided in \Cref{sec:mixture_KL_computation}.

\vspace{-0.5cm}
\section{Applications in Bayesian Inference}\label{sec:applications}
In this section, we present several applications of our method to Bayesian inference. We first consider more structured choices of the approximating class $\mathfrak T \subset \mathcal T$ to further improve computational efficiency. When the posterior is approximately Gaussian, we use an affine transport class $\mathfrak T_{\mathrm{aff}}$. For Bayesian latent variable models, we adopt a mean-field transport class $\mathfrak T_{\mathrm{MF}}$ that reduces the effective cardinality of the discrete variable $\tau$. We then turn to inference and uncertainty quantification enabled by the interpretation of the OT map, which naturally induces multivariate ranks, center outward quantile contours, and sign curves \citep{hallin2021distribution}, providing new tools for Bayesian exploratory analysis and visualization.

\vspace{-0.3cm}

\subsection{Approximately Gaussian posteriors}
\vspace{-0.1cm}
The Bernstein–von Mises theorem for parametric models \citep{le2000asymptotics,van2000asymptotic} states that, under the frequentist setting where $\{Z_1,\ldots,Z_n\}$ are i.i.d.\ samples from the true model $p(\cdot\mid\theta_0)$, the posterior $\pi_n(\theta)=p(\theta\mid Z^n)$ is asymptotically normal, centered near the maximum likelihood estimator and with covariance given by the inverse Fisher information. 
This result allows us to simplify the approximating transport family while retaining high accuracy, as the following theorem shows that an affine transport map suffices to approximate a posterior that is close to Gaussian without sacrificing statistical precision for inference.

\begin{lemma}[Affine transport maps for nearly Gaussian posteriors]\label{thm:bvm_linear} Suppose the posterior $\pi_n$ satisfies the following Gaussian approximation bound:
\vspace{-0.5em}
\begin{align}\label{eq:bvm_kl}
d_\mathrm{KL} \Big(  \mN\big(\widehat \theta_n, \, n^{-1}I^{-1}_{\theta_0}\big)\,\Big\|\,\pi_n \Big) = \m O_p(n^{-1/2}),\quad \mbox{as } n\to\infty, \\[-3em]\notag
\end{align}
where $\widehat \theta_n$ denotes the MLE, and $I_{\theta_0}$ is the nonsingular Fisher information matrix at the true parameter $\theta_0$. Then, with the choice of affine transport map class $\mathfrak{T}_{\rm aff}$ of:
\vspace{-0.5em}
\begin{align}\label{eq:linear_map}
{\footnotesize\mfT_{\rm aff} = \Big\{ T: \R^p \to \Theta \,\Big|\, T(x)= m +\frac{1}{\sqrt n}S^{T}x,\text{ with } m \in \R^p \mbox{ and } S\in \R^{p\times p} \text{ being p.d.}\Big\},} \\[-3em]\notag
\end{align}
the estimated transport map $\wh T$, as defined in equation~\eqref{eq:That}, satisfies
\vspace{-0.5em}
\begin{align*}
    d_\mathrm{KL}({\wh T}_\#\mu \,\| \, \pi_n) = \m O_p(n^{-1/2}), \quad \mbox{as }n\to\infty.\\[-2.5em]\notag
\end{align*}
\end{lemma}

\noindent {We include its proof in \Cref{sec:bvm_linear_pf}.}
The Gaussian approximation in \eqref{eq:bvm_kl} typically holds for parametric models with a fixed number of parameters under standard regularity conditions; see, for example, \cite{katsevich2023improved,zhang2024bayesian}. Under the same conditions, the MLE satisfies $\sqrt n(\widehat\theta_n-\theta_0)\overset{\mathrm d}{\to}\mathcal N(0,I_{\theta_0}^{-1})$ as $n\to\infty$. Consequently, credible intervals constructed from the approximation $\widehat T_\#\mu$ to the posterior $\pi_n$ achieve their nominal coverage asymptotically, so using the affine class $\mathfrak T_{\mathrm{aff}}$ preserves statistical accuracy.

\vspace{-0.5cm}

\subsection{Bayesian latent variable models}
\label{sec:gmm}
\vspace{-0.1cm}
Many Bayesian latent variable models involve multiple discrete latent variables $\{c_j\}_{j=1}^J$, for which the cardinality of the latent space often grows exponentially in $J$, creating significant computational challenges for sampling. For instance, in Gaussian mixture models and hidden Markov models \citep{rabiner1989tutorial}, the latent variables $c_j$ are observation specific, so $J$ equals the sample size $n$. In spike and slab regression \citep{ishwaran2005spike}, each $c_j$ is a binary indicator for inclusion of the $j$-th covariate, and $J$ equals the number of covariates. In this subsection, we use the Bayesian Gaussian mixture model as a representative example and show how to construct and approximate the optimal transport map for the posterior over the mixed parameter $\theta$, which includes both model parameters and many discrete latent variables. To improve scalability, we incorporate the mean-field approximation from variational inference \citep{blei2017variational,zhang2024bayesian} to define a transport class $\mathfrak{T}_{\mathrm{MF}}$ for approximating $T^*$.

Specifically, we consider the standard hierarchical formulation of a Gaussian mixture model with $K$ isotropic Gaussian components in $d$ dimensions. Let $Z^{(n)}=(Z_1,\ldots,Z_n)^\top$ denote the observations and $\{c_1,\ldots,c_n\}$ the corresponding latent variables. The model is
$Z_i \mid c_i,\{m_k\}_{k=1}^K \sim \mathcal{N}(m_{c_i}, \sigma^2 I_d)$,
$c_i \sim \mathrm{Categorical}\big(\tfrac{1}{K},\ldots,\tfrac{1}{K}\big)$, and
$m_k \sim \mathcal{N}(m_0,\lambda^2 I_d)$,
for $i=1,\ldots,n$ and $k=1,\ldots,K$. Here $m_k$ is the mean of cluster $k$, and $c_i$ indicates the cluster assignment of $Z_i$. The first two relations define the complete data likelihood, while the third specifies the prior for the model parameter $\bm\zeta=(m_1,\ldots,m_K)^\top$. By Bayes rule, the joint posterior $\pi_n$ of $\theta=(\tau,\bm\zeta):=(\{c_i\}_{i=1}^n,\{m_k\}_{k=1}^K)$ can be written as
\vspace{-0.5em}
\begin{align*}
    \pi_n(\theta) = \pi(\tau,\,\bm \zeta \,|\,Z^{(n)}) \propto  \prod_{i=1}^n \big[\,p(Z_i\,|\,c_i,\,\bm\zeta)\cdot p(c_i)\,\big] \cdot\prod_{k=1}^K \pi( m_k).\\[-3em]
\end{align*}

Since the $n$ latent variables $\{c_i\}_{i=1}^n$ take values in the same unordered label set $[K]$, we follow \Cref{rmk:embedding} and use one-hot encodings $b_k\in\mathbb{R}^K$ for each label $k\in[K]$.
The overall embedding of $\tau=(c_1,\ldots,c_n)$ is $b_\tau=(b_{c_1},b_{c_2},\ldots,b_{c_n})\in\mathbb{R}^{(K-1)n}$. We choose the reference distribution $\mu$ defined on $\mathcal X=\mathcal X_1\times\mathcal X_2=\mathbb{R}^r\times\mathbb{R}^p$ to be $\mu=\mathcal{N}\!\left(0,I_{r+p}\right)$, where $r=(K-1)n$ and $p=Kd$.
In this model, the discrete variable $\tau$ has cardinality $K^n$, which grows exponentially with $n$. According to Theorem~\ref{thm:discrete_case}, this would require estimating $K^n$ convex functions $\{\phi_{c_1,\ldots,c_n}\}_{(c_1,\ldots,c_n)\in[K]^n}$. As a result, computing the optimal transport map $T^*$ in \eqref{eq:dis_trans_map} becomes infeasible even for moderate $n$.
To address this computational challenge, we propose a ``mean-field" approximation to reduce the complexity of the transport map. Instead of estimating an exponentially large number of functions, we decompose the collection of convex functions $\big\{\phi_{c_1,c_2,\ldots,c_n}:\, c_i \in [K],\ i \in [n]\big\}$ as a sum of $n$ individual convex functions:
\vspace{-0.5em}
\begin{align}\label{eqn:mf-approx}
    \phi_{c_1,c_2,\ldots,c_n}\big(x^{(2)}\big) = \sum_{i=1}^n \phi_{c_i}^{(i)}\big(x^{(2)}\big), 
\end{align}
which reduces the number of convex functions from $K^n$ to $Kn$. The resulting transport map from \eqref{eq:dis_trans_map} can be correspondingly factorized as $T^*= \big(T_{1}^{*(1)},\,T_{2}^{*(1)},\ldots,\,T_{n}^{*(1)},\, T^{*(2)}\big)$, 
\begin{equation}\label{Eqn:GMM_map}
\begin{cases}
\ c_i = T_{i}^{*(1)}\big( x^{(1)}_i,\,x^{(2)}\big) = 
\underset{k\in[K]}{\arg \max} \left\{\langle x^{(1)}_i, b_{k}\rangle + \phi_{k}^{(i)}(x^{(2)})\right\},~~\text{for}~i\in[n],\\
\ \bm \zeta = T^{*(2)}\big(\big\{x^{(1)}_i\big\}_{i=1}^n,\,x^{(2)}\big) = \kappa \sum_{i=1}^n \nabla \phi_{c_i}^{(i)}(x^{(2)}),
\end{cases}
\end{equation}
for $x^{(1)}_i \in \mathbb{R}^{K-1}$ and $x^{(2)} \in \R^p$. We refer to the collection $\mfT_{\rm MF}$ of all transport maps of this form as the mean-field transport map class.
Then, our transport map estimator $\wh T$ is obtained by optimizing the objective in \eqref{eq:mixture_KL_opt} over all $T \in \mfT_{\rm MF}$.

The approximation in \eqref{Eqn:GMM_map} follows the spirit of mean-field variational inference \citep{blei2017variational} by simplifying the dependence among latent variables. Our approach lies between a fully flexible model and the classical mean-field approximation, which assumes a fully factorized family. The scheme in \eqref{eqn:mf-approx} imposes a partial conditional independence structure among $\{c_i\}_{i=1}^n$ given the model parameter $\bm\zeta$ determined by $x^{(2)}$, substantially relaxing the full independence assumption.
Since \cite{yang2020alpha} already establishes the estimation consistency for the standard mean-field approximation in Bayesian latent variable models, we expect similar consistency to hold for this more flexible scheme. Numerical results in \Cref{sec:gmm_simu} confirm that the simplified structure in \eqref{eqn:mf-approx}  effectively captures the target posterior.

\vspace{-0.3em}
\subsection{Posterior summaries via OT-derived quantities}\label{sec:center_outward}
Constructing Bayesian credible intervals or regions requires a meaningful notion of quantiles for multivariate posterior distributions. Various approaches have been proposed in the literature, including marginal quantiles, ellipsoidal quantiles for approximately Gaussian posteriors, copula-based quantiles, Tukey depth contours \citep{tukey1975mathematics}, and data depth methods \citep{liu1990notion}. 
More recently, \cite{hallin2021distribution} introduced OT-based center-outward quantiles, where quantile contours of a reference distribution, such as a standard Gaussian, are transformed through the OT map $T^*$. These quantiles retain key properties of univariate quantiles and provide a useful tool for multivariate exploratory analysis and visualization.

\vspace{-0.3em}
\subsubsection{Center-outward quantiles}\label{sec:center_outward_quantiles}

We demonstrate how our OT-based method facilitates the construction of center-outward quantiles for multivariate distributions. This approach preserves the geometric structure of the target distribution, provides meaningful directional information, and can be used for Bayesian exploratory analysis and visualization.
We begin with examples of constructing multivariate quantile contours. %
Our approach is closely related to \cite{hallin2021distribution}, which uses the fact that $T^*$ is the gradient of a convex function to preserve the ordering of outward quantiles. Here, we learn $T^*$ from the unnormalized density of $\pi_n$, whereas \cite{hallin2021distribution} estimate $T^*$ from the empirical data distribution, yielding quantiles defined only at observed data points. In the remainder of this section, we take the reference distribution to be a standard Gaussian and construct center-outward quantile contours by mapping its quantile circles through the estimated OT map $\widehat{T}$.

To illustrate the effectiveness of our OT-based method in preserving non-convex quantile shapes,
we evaluate it on a multimodal distribution, namely the mixture of two bivariate Gaussian distributions considered in \cite{hallin2021distribution}. The target distribution is
$\frac{1}{2}\mathcal{N}(-4m_h, I_2) + \frac{1}{2}\mathcal{N}(4m_h, I_2)$, where $m_h=(1,0)^\top$. 
To generate the $q$-th center-outward quantile, we draw samples from the circle $\{x\in\mathbb{R}^2:\|x\|_2^2=-2\ln(1-q)\}$, corresponding to the $q$-th quantile of $\mathcal{N}(0,I_2)$, and map them through the estimated transport map $\widehat T$. %
{In \Cref{fig:two_balls}, we map the quantile contours (20\%, 50\%, and 90\%) and the axes of the reference measure $\mu$, represented by sign curves from \cite{hallin2021distribution}, to the target distribution. 
The figure also shows the approximated posterior samples and the corresponding center-outward quantiles produced by our method, normalizing flow methods using Planar transformers (Planar, \citep{rezende2015variational}), ICNN \citep{amos2017input}, and the triangular map \citep{marzouk2016sampling}. We also compute the $W_2$ distance between samples from the transport maps and the true distribution in \Cref{tab:W2_two_ball}.

From \Cref{fig:two_balls}, we observe that under our map, all probability mass of $\mu$ with $x_1<0$ is transported to the left mode of $\pi_n$, while the remaining mass is transported to the right mode. The quantile contours remain disconnected due to the well-separated components of $\pi_n$, demonstrating that our quantiles adapt naturally to the geometry of the target distribution. In contrast, although Planar generates reasonable samples, its directional structure is distorted. ICNN captures the separation between the two components through the ReLU activation but simply partitions $\mu=\m N(0,I_2)$ into two half-Gaussians. The triangular map fails to separate the two components because its continuity and smoothness make it unsuitable for capturing multimodal target distributions.
Additional comparisons with other transport map methods and an example on a non-convex banana-shaped distribution are provided in \Cref{sec:additional_quantile}, further illustrating the advantage of our method in preserving complex quantile geometry and directional information. \Cref{tab:W2_two_ball} also shows that our method and normalizing flow methods incur much smaller $W_2$ error in approximating the bimodal target distribution than the other competitors. For this reason, in later numerical studies involving more complex and potentially multimodal distributions, we restrict comparisons primarily to normalizing flow methods.
}
\begin{figure}[!ht]
    \centering
    \subfloat[OT \label{fig:two_balls_OT}]{%
        \includegraphics[trim={4em 1em 4em 6em},clip,width=0.24\textwidth]{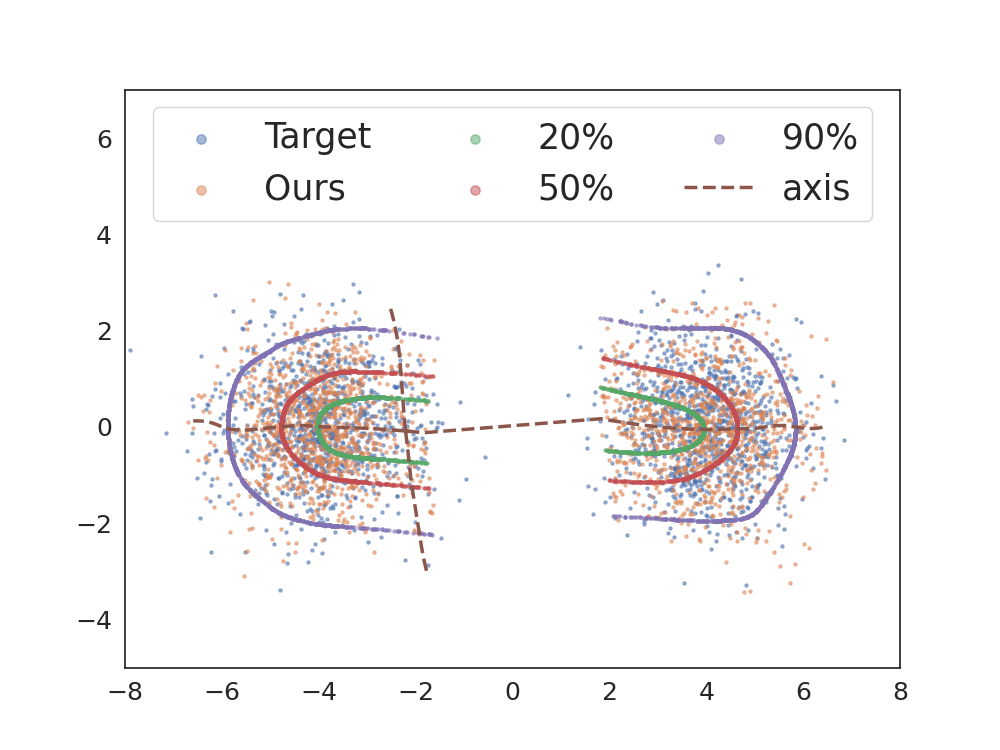}
    }%
    \subfloat[Planar \label{fig:two_balls_NF}]{%
        \includegraphics[trim={4em 1em 4em 6em},clip,width=0.24\textwidth]{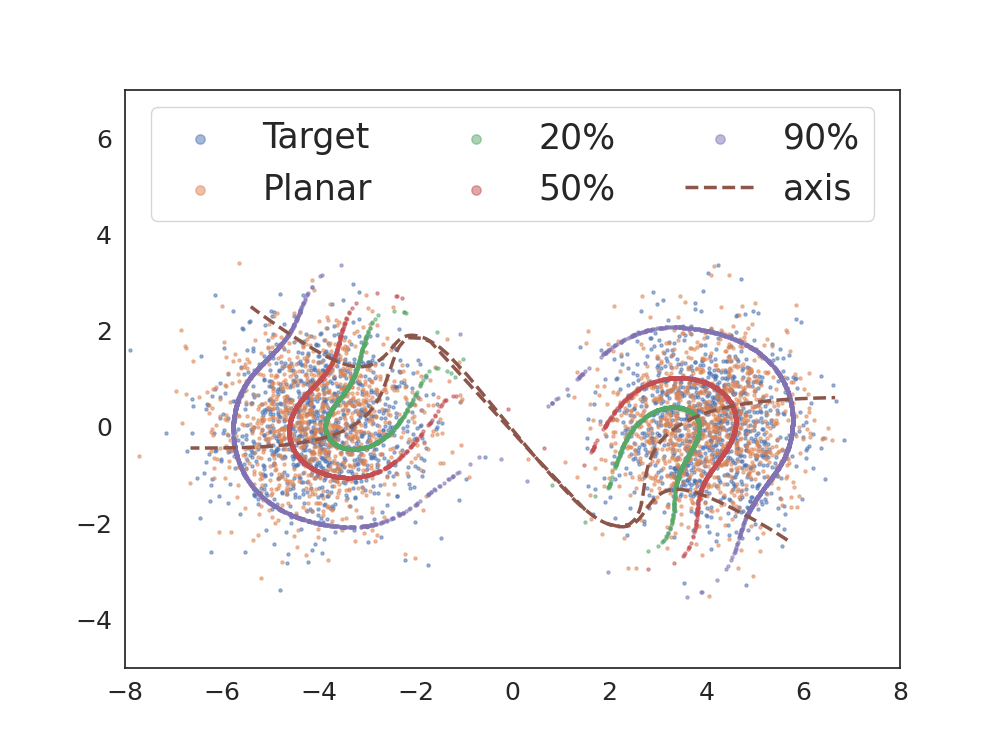}
    }
    \subfloat[ICNN]{
        \includegraphics[trim={4em 1em 4em 6em},clip,width=0.24\textwidth]{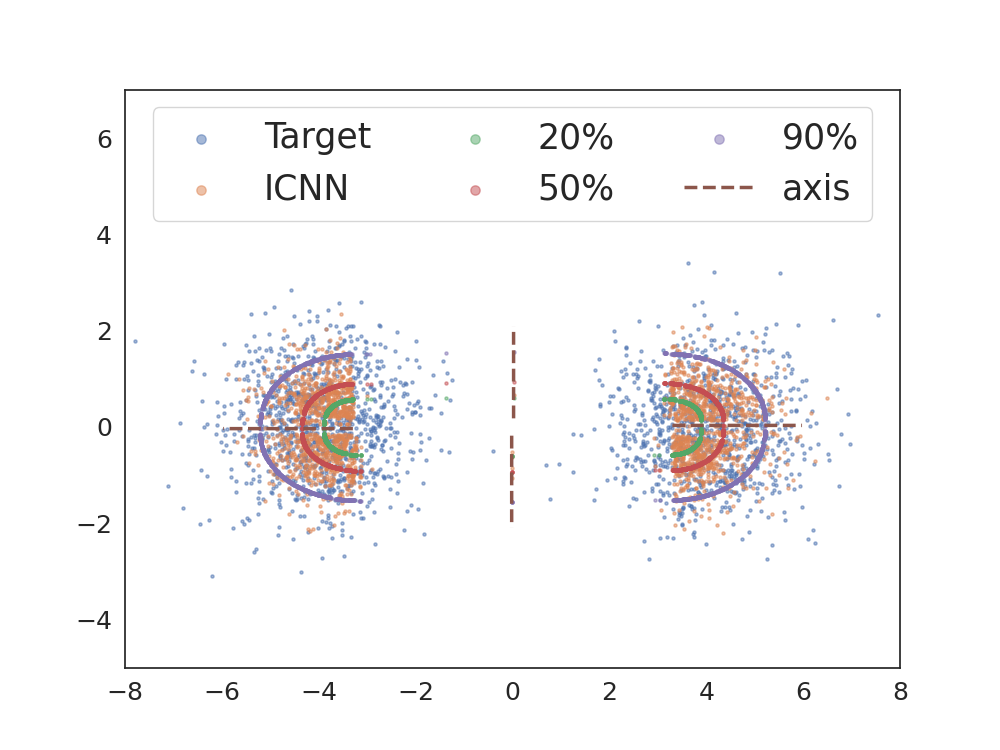}
    }
    \subfloat[Triangular Map]{
        \includegraphics[trim={4em 1em 4em 6em},clip,width=0.24\textwidth]{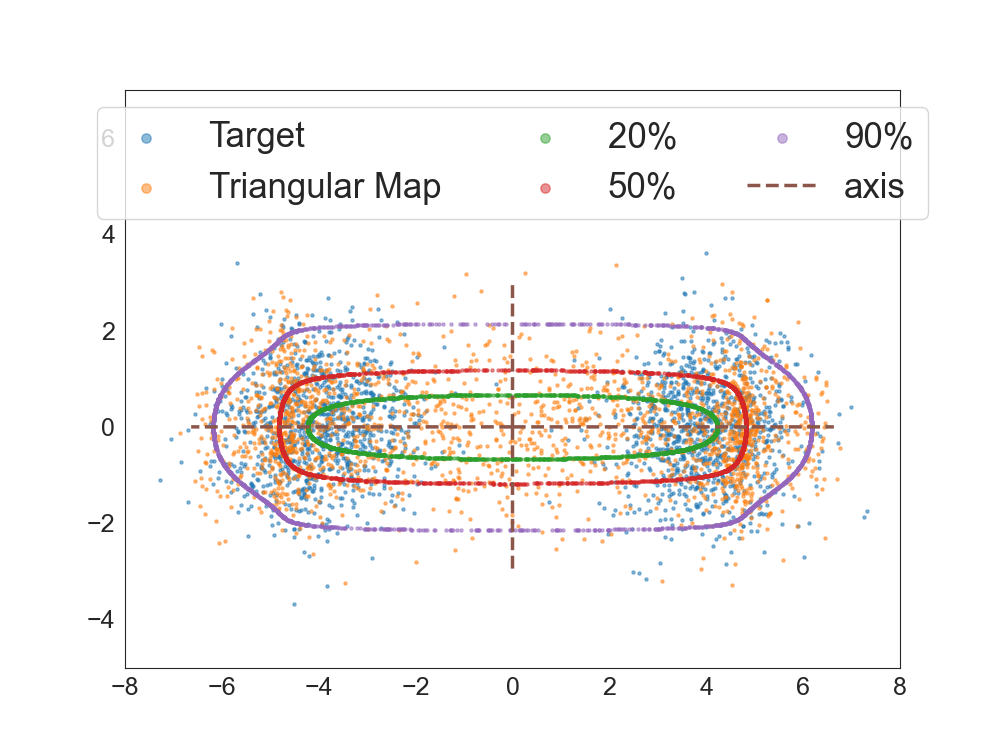}
    }
    \caption{Center-outward quantile contours for the mixture of two bivariate Gaussian example. Here we include the quantile contours by passing the reference contours through the transport map estimated by our method (with $L=2$),  Planar, ICNN and the triangular map.}\label{fig:two_balls}
\end{figure}

\begin{table}[!ht]
\centering
\begin{tabular}{c|c|c|c|c}
\toprule
 Ours & Planar & ICNN & Triangular Map & NSF \\
 \midrule
0.010 &0.010 & 5.057 &  0.183 & 0.085 \\
\bottomrule
\end{tabular}
\caption{$W_2$ distance between samples from the approximated distribution and true distribution in the mixture of two Gaussian example.  NSF: neural spline flow \citep{durkan2019neural}.}\label{tab:W2_two_ball}
\end{table}

\subsubsection{Applications in Bayesian exploratory analysis} 

\vspace{-0.2em}
Multivariate center-outward credible regions, defined by the center-outward quantiles of the posterior $\pi_n$, provide a convenient tool for Bayesian exploratory analysis. These regions account for dependence among parameters and offer a low-cost way to reveal structural patterns before more refined downstream inference. We illustrate two applications of center-outward quantiles.

The first application is the construction of $95\%$ simultaneous credible intervals $\{[a_i,b_i]\}_{i=1}^p$ for a parameter vector $\theta=(\theta_1,\ldots,\theta_p)^\top$ such that $\pi_n\Big(\bigcap_{i=1}^p \{\theta_i\in[a_i,b_i]\}\Big)\ge 95\%$.
These intervals provide a simple way to visualize joint uncertainty across multiple parameters and are particularly useful in multiple comparison settings. Classical approaches based on highest posterior density regions are computationally demanding in high dimensions, while Bonferroni-type marginal adjustments are often overly conservative. In contrast, we construct simultaneous credible intervals from the multivariate center-outward credible region by finding the smallest box $[a_1,b_1]\times\cdots\times[a_p,b_p]$ that contains the $95\%$ region. This approach is straightforward, scalable, and naturally accommodates dependence among parameters.
Concretely, we draw $N$ i.i.d.\ samples uniformly from the $95\%$ quantile ball of the reference distribution $\mathcal N(0,I_p)$, namely $X^{(j)}\sim \mathrm{Unif}\!\big(B_p(0,\sqrt{q_{\chi^2_p}(0.95)})\big)$, where $B_p(0,r)$ is the $p$-dimensional Euclidean ball of radius $r$ and $q_{\chi^2_p}(0.95)$ is the $0.95$ quantile of the chi-squared distribution with $p$ degrees of freedom. These samples are mapped to the parameter space by $\widehat\theta^{(j)}=\widehat T(X^{(j)})$. The simultaneous credible intervals are then given by the coordinate-wise ranges $\big\{[\min_j \widehat\theta_i^{(j)},\,\max_j \widehat\theta_i^{(j)}]\big\}_{i=1}^p$.
An example of this approach for assessing variable importance in Bayesian logistic regression is provided in \Cref{sec:logistic}.

A second application is the definition of a Bayesian analogue of the classical $p$-value for assessing the plausibility of a multivariate parameter value. In the frequentist framework, the $p$-value for $\theta_0$ is the probability of observing data, generated under $p(\cdot\mid\theta_0)$, that are at least as extreme as the observed data. Analogously, a Bayesian $p$-value for $\theta_0$ can be defined as the posterior probability of parameters that are more extreme than $\theta_0$. While credible intervals quantify uncertainty relative to the posterior of $\theta$ and confidence intervals quantify uncertainty relative to the sampling distribution of data, this Bayesian $p$-value measures extremity directly in the parameter space. Center-outward quantile contours provide a natural notion of extremity and form the basis for this definition.

Because the optimal transport map preserves the ordering of center-outward quantiles \citep{chernozhukov2017monge}, the Bayesian $p$-value of $\theta_0$ can be computed as one minus the quantile level of its preimage $(T^*)^{-1}(\theta_0)$ under the reference distribution $\mathcal{N}(0,I_p)$, namely
$1 - q_{\chi^2_p}^{-1}\big(\|(T^*)^{-1}(\theta_0)\|_2^2\big)$. This construction relies on a geometric property of $T^*$: as the gradient of a convex potential, it satisfies a multivariate monotonicity that preserves quantile ordering and allows efficient numerical computation of the inverse map $(T^*)^{-1}$.
These Bayesian $p$-values provide a convenient way to compare the plausibility of different parameter values. They also admit an interpretation related to testing $H_0:\theta=\theta_0$ versus $H_1:\theta\neq\theta_0$, since rejecting $H_0$ when the Bayesian $p$-value falls below $\alpha\in(0,1)$ is equivalent to $\theta_0$ lying outside the $100(1-\alpha)\%$ center-outward credible region. This is not a formal Bayesian test, which would require Bayes factors and point prior mass at $\theta_0$, but rather a fast and interpretable diagnostic for exploratory analysis. Details on computing the inverse map and an additional application to ranking posterior draws are provided in \Cref{sec:inverse_map}.

\section{Simulation Studies}\label{sec:simulation}

In this section, we compare our approach with two alternative transport-based methods: transport Monte Carlo (TMC) \citep{duan2021transport} and three normalizing flow methods, namely Planar transformations \citep{rezende2015variational}, Masked Autoregressive Flow (MAF) \citep{papamakarios2017masked}, and Neural Spline Flow (NSF) \citep{durkan2019neural}. Detailed descriptions of these methods are given in \Cref{sec:review_other_methods}, and the experimental configurations are provided in \Cref{sec:simulation_config}. We have made the code and data publicly available at
\url{https://github.com/yuexiwang/Bayesian-optimal-transport}.

\vspace{-0.5cm}
\subsection{Mixture of Gaussian distributions}\label{sec:mixture_normal_simu}
In this experiment, we evaluate different methods on multimodal continuous distributions given by mixtures of $K$ multivariate normals in $\mathbb{R}^d$. The target distribution is
$\pi_n=\frac{1}{K}\sum_{k=1}^K \mathcal{N}(m_k,\Sigma_k)$,
where $\{m_k,\Sigma_k\}_{k=1}^K$ are assumed known. We consider six combinations of $(d,K)$ with $d\in\{5,10,20\}$ and $K\in\{3,10\}$. The cluster means $\{m_k\}_{k=1}^K$ are sampled from $\mathrm{Unif}[-10,10]$, and the covariance matrices are defined by $(\Sigma_k)_{ij}=\rho_k^{|i-j|}$ with $\rho_k=0.5\cdot(-1)^k$. The reference distribution is the standard Gaussian.

{Both the NF methods and TMC struggle to learn multimodal distributions. NF methods are primarily designed for continuous targets and perform poorly on well-separated mixtures, especially in higher dimensions. NSF performs particularly poorly in this example and is therefore excluded from the comparison. MAF can outperform Planar in simpler settings with small $d$ and $K$ under careful tuning, but it fails to capture distinct modes when either $K$ or $d$ increases. In this example, Planar performs better than MAF because the map from the standard Gaussian reference to each mixture component is approximately linear, which Planar can represent effectively.
For TMC, although it incorporates multinomial mixture weights, its transport plan is not tailored to multimodal targets. As a result, TMC becomes computationally expensive when approximating distributions with multiple modes.}

In contrast, our transport map in \eqref{eq:max_of_convex} is explicitly designed for multimodal targets by exploiting the structure in \Cref{lem:regularity_OT_2}. We set $L=3$ when $K=3$ and $L=32$ when $K=10$. When $K=10$, the modes are closer and harder to distinguish, so a larger $L$ improves stability. We choose $M_\ell=16,\,32,\,64$ for $\ell\in[L]$ when $d=5,\,10,\,20$, respectively.

\begin{table}[!ht]
\centering
{
\begin{tabular}{ |c|c|c|c|c|c| } 
\hline
   (d,\, K)  & Benchmark  & Ours & Planar & TMC & MAF  \\ 
\hline
(5,\, 3) & 0.501 & 1.838 & 3.547 & 2.420 & 0.482 \\
\hline 
(5,\, 10) & 1.104 & 2.671 & 4.393 & 6.560 & 5.217 \\
\hline 
(10,\, 3) & 2.663 & 3.923 & 6.463 & 14.696 & 6.214 \\
\hline 
(10,\, 10) & 3.272 & 5.562 & 10.310 & 12.683 & 68.328 \\
\hline 
(20,\, 3) & 7.437 & 10.287 & 9.979 & 21.694 & 11.349 \\
\hline 
(20,\, 10) & 10.001 & 11.334 & 17.546 & 24.954& 311.193 \\
\hline 
\end{tabular}
\caption{2-Wasserstein distance between $10,000$ samples from target distribution and each method at different combinations of $d$ and $K$.}
\label{table:w2-maps}}
\end{table}

\Cref{table:w2-maps} reports the $2$-Wasserstein distances between $10{,}000$ samples from the target distribution and those generated by each method. We also report the benchmark distance between two independent sets of $10{,}000$ samples from the true target distribution to reflect sampling variability. 
Our transport map achieves the smallest distance across most settings of $d$ and $K$, with the exception of $(d,K)=(20,3)$. The performance gap between our method and the alternatives widens as $K$ or $d$ increases, since NF and TMC have difficulty identifying all modes and capturing the distributional variability, whereas our method handles both effectively.

\begin{figure}[!ht]
    \centering
    \subfloat[Ours \label{fig:normal_OT}]{%
        \includegraphics[width=0.32\textwidth, height=5cm]{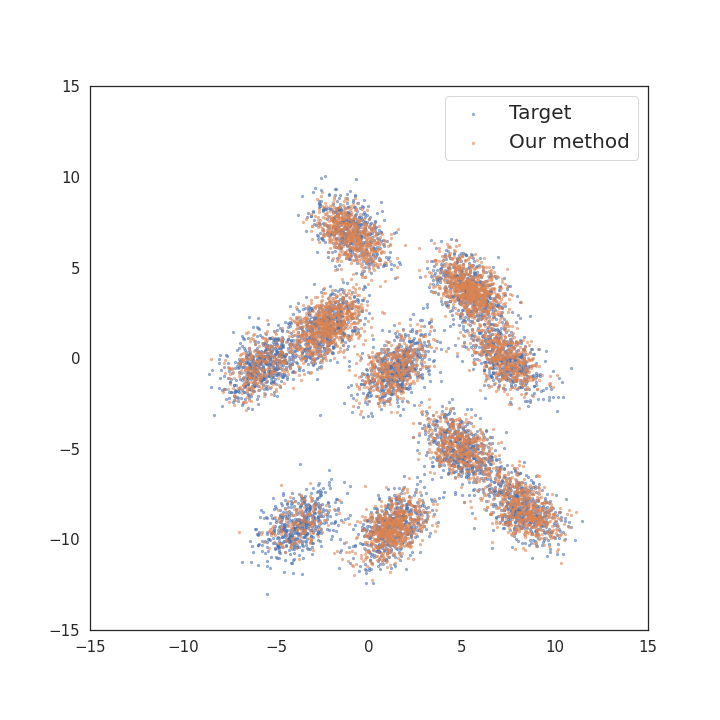}
    }
    \subfloat[Planar]{%
        \includegraphics[width=0.32\textwidth, height=5cm]{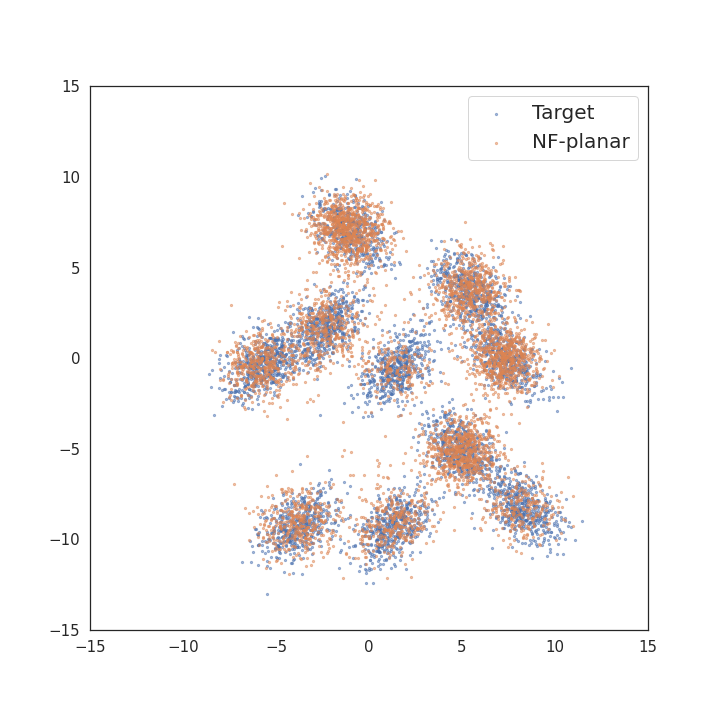}
    }
    \subfloat[TMC]{%
        \includegraphics[width=0.32\textwidth, height=5cm]{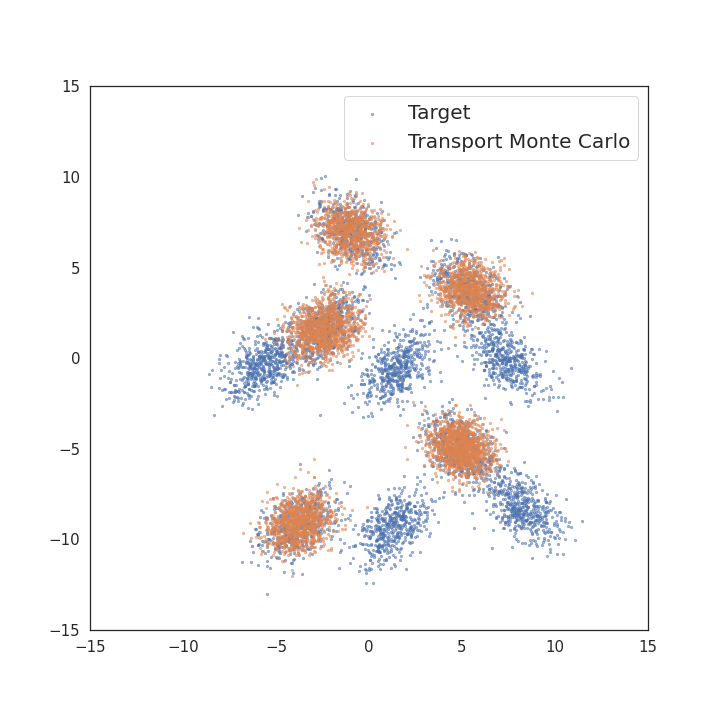}
    }
    \caption{Comparison of draws generated from the target distribution and different approximate methods when \( d=20 \) and \( K=10 \).}
    \label{fig:mixture_normal}
\end{figure}

\Cref{fig:mixture_normal} visually compares samples generated by each method for $d=20$ and $K=10$. Our transport map, with $L=32$ and $M_\ell=64$ for $\ell\in[L]$, performs best by accurately capturing all modes and preserving the spread of each cluster. In contrast, TMC identifies only four modes, and Planar fails to capture the shape of all clusters correctly, indicating a loss of correlation structure.

\vspace{-0.4cm}
\subsection{Bayesian logistic regression}\label{sec:logistic}
In this simulation study, we consider Bayesian logistic regression, a problem known to be challenging due to the nonlinear form of the binomial likelihood. A common strategy to enable conjugate updating is Polya--Gamma data augmentation \citep{polson2013bayesian}, but its implementation can be involved. Our method provides a simpler and more direct alternative.

Given $n$ binary responses $\{y_1,\ldots,y_n\}\subset\{0,1\}$ with covariates $\{x_1,\ldots,x_n\}\subset\mathbb{R}^p$, the posterior of the regression coefficient $\beta\in\mathbb{R}^p$ is
\begin{align*}
\pi_n(\beta\mid \bx, \by)  \ \propto \ \, \underbrace{\prod_{i=1}^n\frac{\big(e^{x_i^T\beta}\big)^{y_i}}{1 + e^{x_i^T\beta}}}_\text{likelihood} \ \times \ \underbrace{\pi(\beta)}_\text{prior},
\end{align*}
where $\by=(y_1,\ldots,y_n)^\top$, $\bx=(x_1,\ldots,x_n)^\top$, and we use a weakly informative prior $\beta\sim\mathcal{N}(0,\sigma^2 I_p)$ with $\sigma=10$.
For the simulated data, we set $n=1000$ and $p=10$. The entries of the true $\beta$ are drawn from $\mathrm{Unif}[-1,1]$. The covariates $x_i$ are sampled from $\mathcal{N}(0,\Sigma)$, where $\Sigma$ is a Toeplitz matrix with $\Sigma_{ij}=\rho^{|i-j|}$ and $\rho\in\{0.1,0.3,0.5,0.7,0.9\}$.
Since the posterior is unimodal in this example, we set $L=1$ in the approximation family \eqref{eq:max_of_convex}. Moreover, because $n\gg p$, we use the linear transport map in \eqref{eq:linear_map} to approximate the posterior, referred to as \texttt{OT\_linear}.
 
We use a Gibbs sampler as a benchmark, since direct sampling from $\pi_n$ is infeasible. The performance of different samplers is evaluated using two metrics: (1) the difference ratio of $95\%$ credible intervals (CIs), and (2) the $2$-Wasserstein distance between $500$ draws from the Gibbs sampler and those from each transport-based method. All metrics are averaged over $100$ repeated experiments.
For the $95\%$ CIs, we obtain $500$ draws from each sampler and construct marginal credible intervals for each parameter. The difference ratio is defined as $|I_1 \oplus I_2|/|I_1|$, where $I_1$ is the CI from the Gibbs sampler, $I_2$ is the CI from the transport-based method, and $I_1 \oplus I_2=(I_1\cup I_2)\setminus(I_1\cap I_2)$.
For the $2$-Wasserstein distance, we report a standardized version to account for variability across different $\rho$. Each parameter is standardized by its standard deviation computed from the Gibbs samples, and the $2$-Wasserstein distance is computed on the standardized parameters and averaged over the $10$ dimensions. Results for different correlation strengths $\rho$ are summarized in \Cref{fig:diff_ratio_w2_dist_new}.

From \Cref{fig:diff_ratio_w2_dist_new}, our methods consistently perform best across all values of $\rho$, while TMC performs substantially worse in terms of both credible intervals and predictive accuracy.
We also observe that the difference between the map learned using the convex unit construction with $L=1$ (OT) and the simpler linear map in \eqref{eq:linear_map} (\texttt{OT\_linear}) is minimal. This provides empirical evidence for the effectiveness of the linear approximation in \eqref{eq:linear_map}, especially when the sample size $n$ is large.

\begin{figure}[!ht]
    \centering
    \includegraphics[width=\textwidth]{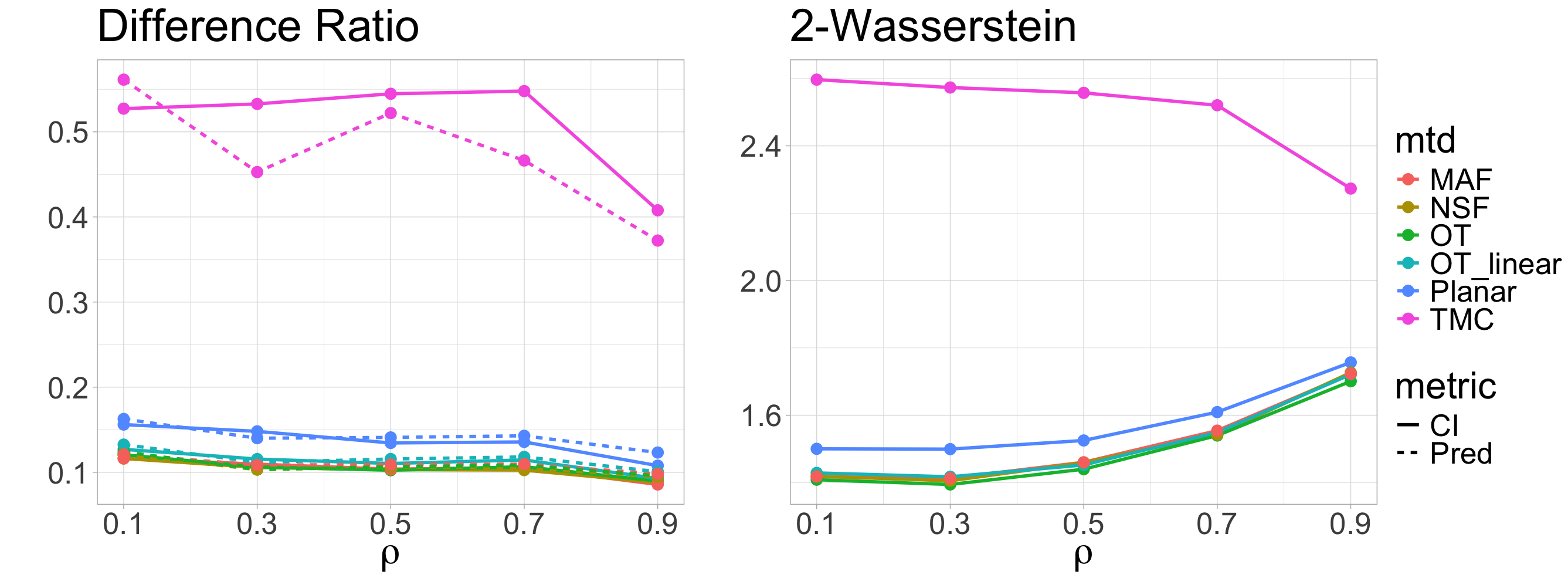}
    \caption{Difference ratio of \(95\%\) credible intervals for parameters and \(95\%\) prediction intervals of new data points, and the standardized 2-Wasserstein distance for posterior distributions under transport map methods compared with the Gibbs posterior.}
    \label{fig:diff_ratio_w2_dist_new}
\end{figure}

{\noindent {\bf  Exploratory analysis using OT-derived quantiles.} 
We present another example to illustrate the inferential advantage of the approach in \Cref{sec:center_outward} for exploratory analysis involving multiple parameters.
We consider a sparse regression setting with $\beta=(2,0,4,0,3,0,-1,0,1,0)^\top$ and introduce moderate correlation among predictors by setting $\rho=0.5$. This induces dependence among the posterior marginals of the coefficients $\widehat\beta_i$, which can make interpretation based solely on marginal summaries misleading. In such cases, simultaneous credible intervals derived from the center-outward quantile regions in \Cref{sec:center_outward} provide a way to visualize joint uncertainty while accounting for parameter dependence.

Recall that the method constructs approximate $95\%$ simultaneous credible intervals by finding the smallest axis-aligned box that contains the $95\%$ center-outward credible region. For exploratory purposes, one can assess the relevance of each coefficient by checking whether zero lies within its interval. This visualization offers an interpretable summary of parameter importance that accounts for joint dependence and controls familywise error. A detailed view of the posterior distribution is shown in \Cref{fig:logistic_variable_selection}.

\begin{figure}[!ht]
\centering
\subfloat[Marginal 95\% CI\label{fig:logistic_marginal}]{
\includegraphics[clip,width=0.3\textwidth]{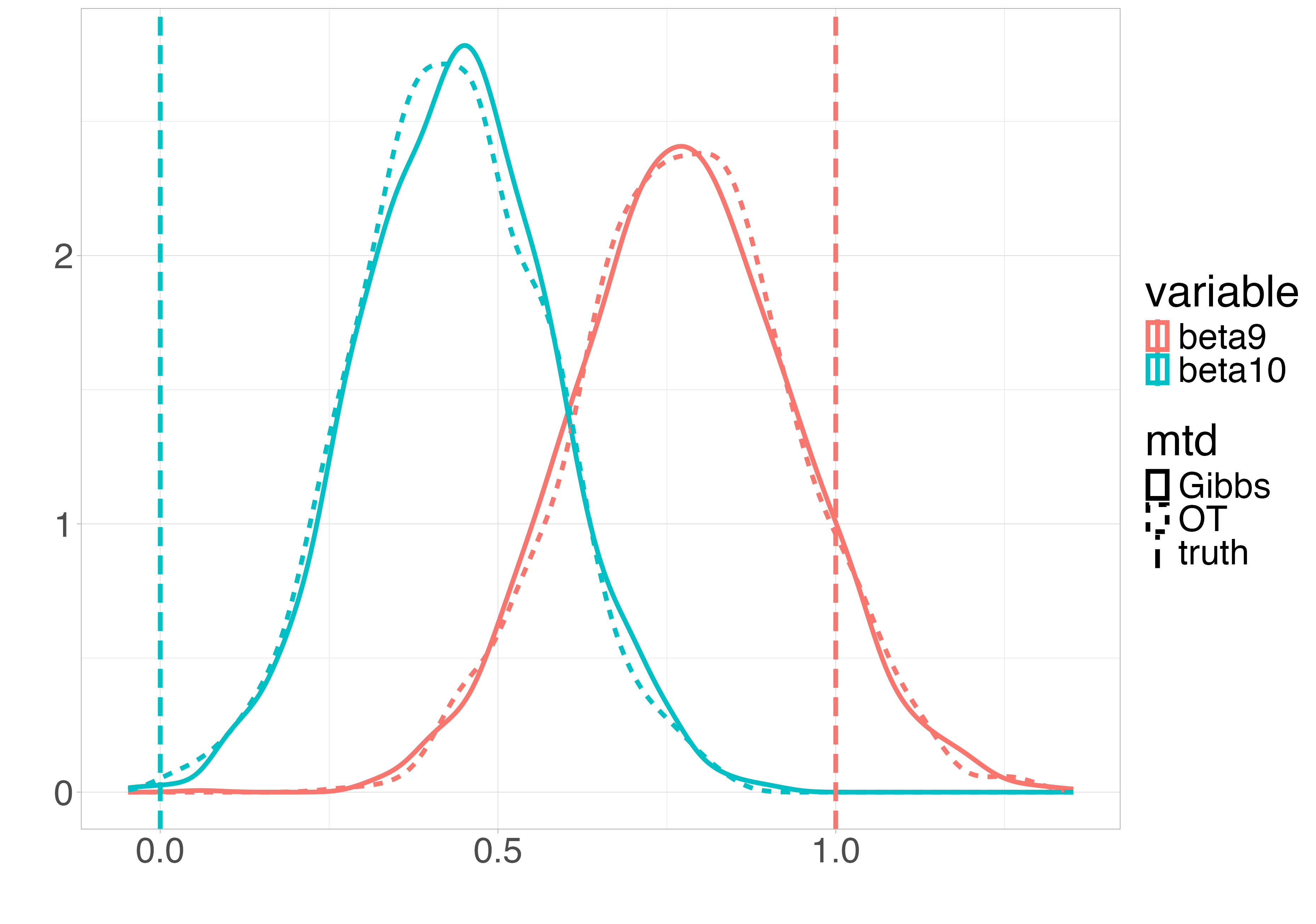}
}
\subfloat[95\% Credible Region\label{fig:logistic_credible_region}]{
\includegraphics[clip,width=0.35\textwidth]{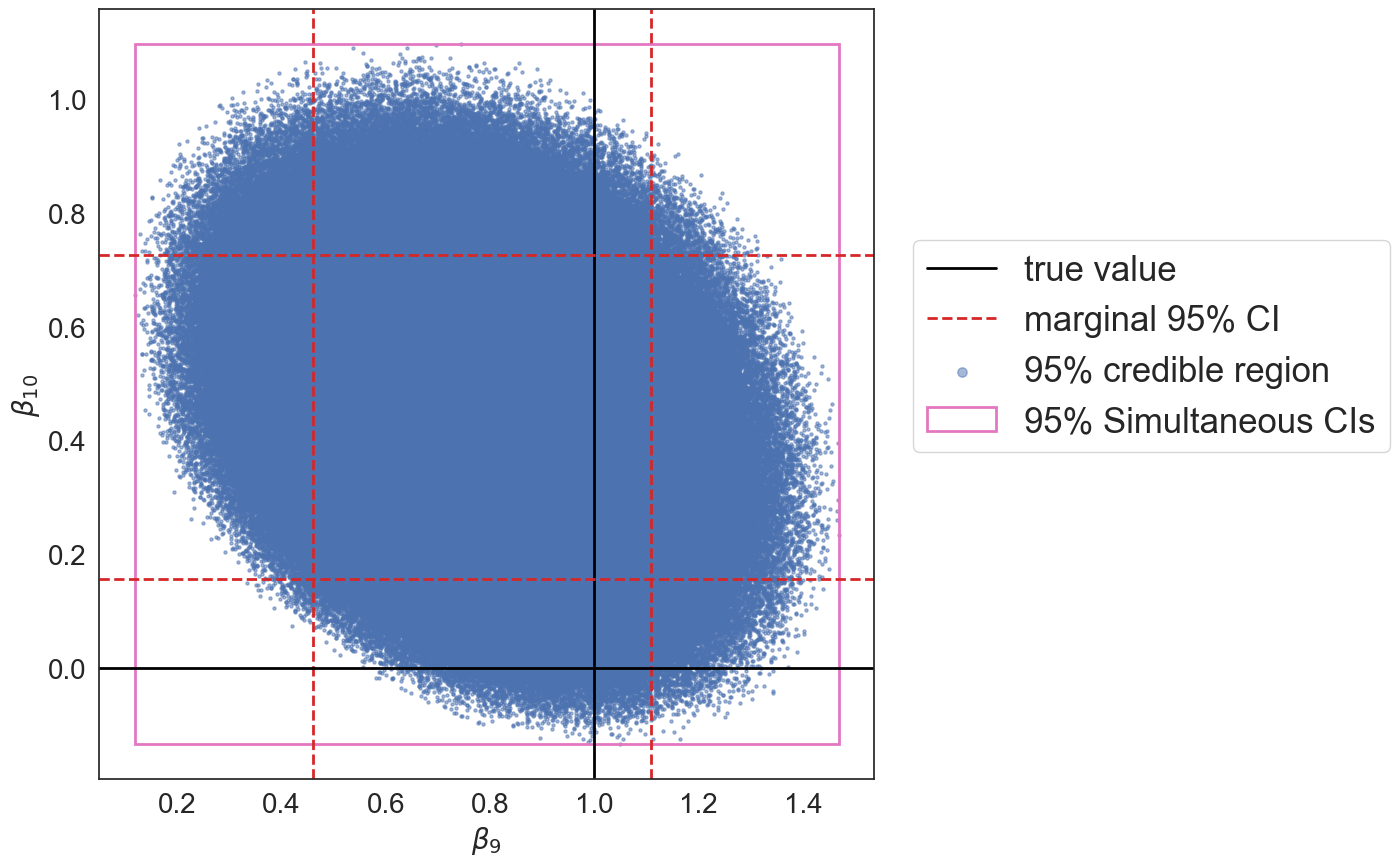}
}
\subfloat[95\% Simultaneous CIs\label{fig:logistic_ci}]{
\includegraphics[clip,width=0.3\textwidth]{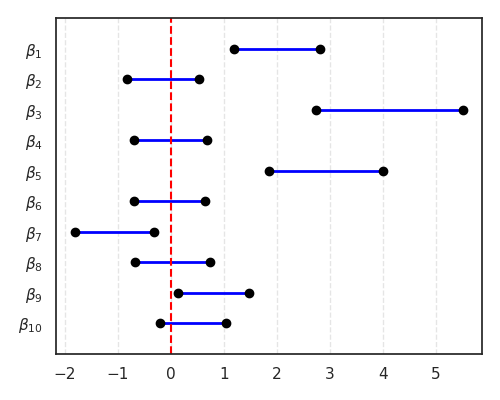}
}
\caption{Posterior summary for $(\beta_9,\beta_{10})$ and simultaneous credible intervals for all $\beta_j$'s. Here we include: (a) marginal density of $\beta_9$ and $\beta_{10}$ produced by Gibbs posterior and our method (OT); (b) the 95\% credible region by OT (projected onto $(\beta_9, \beta_{10})$); and (c) all 95\% simultaneous credible intervals (CIs)}\label{fig:logistic_variable_selection}
\end{figure}

As shown in \Cref{fig:logistic_marginal}, the marginal posterior distributions from our method closely match those from the Gibbs sampler. Both approaches produce marginal $95\%$ credible intervals for $\beta_{10}$ that exclude zero. However, the joint distribution of $(\beta_9,\beta_{10})$ in \Cref{fig:logistic_credible_region} reveals a negative correlation that is captured by the OT-based credible region. This suggests that the marginal view may overstate the evidence against $\beta_{10}=0$, while the joint credible region provides a more nuanced interpretation. Such information can guide more careful follow-up analysis in the presence of parameter dependence.
In \Cref{fig:logistic_ci}, we visualize all $95\%$ simultaneous credible intervals, where $5$ out of $10$ intervals contain zero. This exploratory analysis indicates that incorporating variable selection, for example through sparsity-inducing priors on the regression coefficients, may improve both predictive and inferential performance.

For multivariate hypotheses such as the null model $\beta=(0,\ldots,0)^\top$, it is informative to examine its preimage under the OT map. In this example, the preimage is
$(-27.0,\,4.1,\,-50.6,\,-2.3,\,-37.9,\,10.9,\,16.3,\,16.2,\,-12.1,\,-3.3)^\top$
with $\ell_2$ norm $74.5$, far exceeding the $95\%$ chi-squared radius $q_{\chi^2_{10}}(0.95)=18.3$. The corresponding Bayesian $p$-value,
$1-q^{-1}_{\chi^2_{10}}(74.5^2)<2.22\times 10^{-16}$, indicates that $\beta=0$ is extremely implausible under the posterior. Although not a formal test, this provides a fast way to assess the compatibility of a multivariate parameter value with the joint posterior structure.

}

\vspace{-0.3cm}
 
\subsection{Gaussian mixture models}\label{sec:gmm_simu}
In this subsection, we evaluate the performance of our ``mean-field" transport map on Bayesian Gaussian mixture model with $K=3$ mixture components, as discussed in \Cref{sec:gmm}.  The goal is to assess how well our method approximates the posterior distribution of both latent variables $\{c_i\}$ and cluster means $\{\bm m_k\}$.

We generate $n=300$ samples $Z^{(n)}=\{ Z_1,  Z_2, \ldots,  Z_n\}$ in $p=2$ dimensions and their latent variables $\{c_1,c_2,\ldots,c_n\}$ are uniformly generated from the $K$ components. For the cluster means, we consider $\{{m}_1=(-\Delta/2,0)^T,\, {m}_2=(0,\Delta)^T,\,{m}_3=(\Delta/2,0)^T\}$ where $\Delta \in \{2, 4, 6\}$  controls the separation between clusters.  When $\Delta$ is small, the clusters are closer to each other and the problem becomes harder. Then given the latent variable $c_i$, the observations are drawn from $ Z_i\mid c_i \sim  N( m_{c_i}, I_2)$. We set the prior for cluster means as $ m_k \sim \mN(0, 10^2 I_2)$ for every $k=1, \ldots, K$.

\begin{table}[!ht]
\centering
{
\resizebox{0.9\textwidth}{!}{
\begin{tabular}{ c|cc c|ccc|ccc } 
\toprule
     & \multicolumn{3}{c|}{$\Delta=2$} & \multicolumn{3}{c|}{$\Delta=4$} & \multicolumn{3}{c}{$\Delta=6$} \\ 
 & Ours & Planar & MAF & Ours & Planar & MAF & Ours & Planar & MAF \\
 \midrule
\mbox{Latent variables} & 0.024 & -  & -& 0.005 & - & - & 0.001  & -& - \\
$m_1$ & 0.078 & 0.053 & 0.052 & 0.031 & 0.030 & 0.030 & 0.030 & 0.028 & 0.027 \\

$m_2$ & 0.092  & 0.072 & 0.072 & 0.034 & 0.029  & 0.028 & 0.032 & 0.027 & 0.027 \\

$m_3$ & 0.082  & 0.055 & 0.055 & 0.031 & 0.029 & 0.029 &  0.029 & 0.028 & 0.026\\
\bottomrule
\end{tabular}}}
\caption{Average total variation distance for latent variables and 2-Wasserstein distance for cluster means over $100$ replications for our method, Planar and MAF. The probability mass function of the latent variable is computed in closed form instead of using soft-max approximation.}\label{tab:gmm_distance_closed_form}
\end{table}

Similar to \Cref{sec:logistic}, we use a Gibbs sampler as the benchmark. Due to the large latent space, TMC is computationally infeasible, so we compare only with Planar and MAF. To evaluate performance, we generate $1000$ posterior samples of the latent variables $\{c_i\}$ and cluster means $\{m_k\}$ from both our method and the Gibbs sampler. For the latent variables, we compute the total variation distance for each $c_i$ and average over all $300$ variables. For the cluster means $\{m_k\}$, we compute the $2$-Wasserstein ($W_2$) distance between the empirical distributions. Since Planar and MAF do not produce posteriors for latent variables, we evaluate them only on the mean parameters $\{m_k\}$. All metrics are averaged over $100$ repeated experiments, with results reported in \Cref{tab:gmm_distance_closed_form}.
When $\Delta$ is large and the clusters are well separated, our method performs comparably to the NF methods in estimating the cluster means. Unlike NF methods, however, our approach simultaneously estimates both the latent variables $c_i$ and the means $m_k$, yielding a more interpretable posterior. When $\Delta$ is small and the clusters overlap, the problem approaches a singular regime in which the MLE no longer achieves the parametric $\sqrt{n}$ rate, making estimation more challenging and degrading the performance of all methods. Although our method exhibits a larger $W_2$ distance than the NF methods in this regime, it still provides estimates for the latent variables, which NF methods do not.

\section{Empirical Analysis on the \textit{Yeast} Dataset}\label{sec:real_data}
In this section, we evaluate the performance of our method on the yeast dataset \citep{elisseeff2001classification}. The dataset contains $2,417$ genes, each with $103$ covariates and $14$ gene functional classes. Each functional class defines an independent binary classification task. In our analysis, we focus on the first gene functional class and model the problem using Bayesian logistic regression.

We begin by pre-screening covariates, retaining only those with an absolute correlation of at least 0.1 with the class label. This filtering step results in $24$ covariates. We then apply a generalized linear model (GLM) for standard logistic regression and use an MCMC implementation for Bayesian logistic regression. Our method is compared with four alternatives: Planar, TMC, MAF and NSF.
We assess performance from two perspectives: uncertainty quantification and variable selection. For variable selection, a covariate is considered significant if its marginal 95\% credible interval (or confidence interval) does not include zero. Since the true model is unknown, we use the MCMC posterior as the benchmark.

When trained on the full dataset, both GLM and MCMC identify the same 10 significant covariates (including the intercept) at a 5\% significance level. Among the transport map methods, our approach  and MAF select exactly the same covariates as MCMC. Note that we also check the significance of covariates using the 95\% simultaneous credible intervals we show in \Cref{sec:logistic} and we find only 3 covariates significant, suggesting strong correlation among covariates. In contrast, Planar misses one covariate, NSF selects on extra covariate and TMC selects two additional covariates not selected by MCMC. To evaluate predictive performance, we split the dataset into training and testing subsets, with 70\% of the data used for training. Both of our method and NSF achieve a test accuracy of 76.8\%, slightly outperforming MAF (76.6\%), Planar (76.4\%) and TMC (76.7\%), and comparable to the MCMC-based method (76.7\%).

We further analyze uncertainty quantification by comparing the 95\% credible interval length difference ratios between the transport-based methods and MCMC, as described in \Cref{sec:logistic}. The results for the 10 significant variables are shown in \Cref{table:diff_credible_intervals_yeast}. For most variables, the credible intervals produced by our method are the closest in length to those from MCMC, further supporting its reliability.

\begin{table}[!ht]
\centering
\resizebox{0.9\textwidth}{!}{
\begin{tabular}{|c|c|c|c|c|c|c|c|c|c|c|}
\hline 
& Intercept & Att3 & Att34 & Att58 & Att66 
& Att79 & Att88 & Att89 & Att96 & Att102 \\
\hline
Ours 
& 0.026 & 0.022 & 0.020 & 0.014 & 0.002
& 0.004 & 0.019 & 0.025 & 0.007 & 0.019 \\
\hline
Planar 
& 0.018 & 0.052 & 0.090 & 0.035 & 0.050
& 0.046 & 0.037 & 0.033 & 0.020 & 0.005 \\
\hline
TMC 
& 0.011 & 0.005 & 0.11 & 0.014 &
       0.01 & 0.047 & 0.058  & 0.034 &
       0.036 & 0.012 \\
       \hline
NSF 
& 0.005 & 0.012 & 0.024 & 0.019 &
       0.023 & 0.008 & 0.037 & 0.009 &
        0.013 & 0.017  \\
\hline
MAF
& 0.012 & 0.015 & 0.03  & 0.012 &
       0.012 & 0.012 & 0.022 & 0.008 & 0.01 & 0.012 \\
\hline
\end{tabular}
}
\caption{Difference ratio of 95\% credible intervals between transport map methods and the MCMC posterior for the 10 significant variables selected by MCMC.}
\label{table:diff_credible_intervals_yeast}
\end{table}

To assess robustness, we repeat the variable selection process 100 times using different train-test splits. In \Cref{tab:sig_var_set}, we report the average number of over-selected and under-selected variables compared to MCMC, which selects an average of 10.2 significant variables. Over-selection refers to variables selected by the transport method but not by MCMC, while under-selection refers to variables selected by MCMC but missed by the transport method.
From \Cref{tab:sig_var_set}, we observe that our method consistently yields good variable selection results to MCMC. This is consistent with the fact that its credible intervals are also the most similar to those from MCMC. These results highlight that our method is a reliable and computationally efficient alternative for Bayesian logistic regression, offering both accurate uncertainty quantification and robust variable selection.

\begin{table}[!h]
\centering{
\begin{tabular}{ |c|c|c|c|c|c| } 
\hline
       & Ours & Planar & TMC & NSF & MAF \\ 
\hline 
\#under-select & 0.04 & 0.75 & 0.02 & 0.14 & 0.14\\
\hline 
\#over-select & 0.18 & 0.05 &1.87 & 0.59 & 0.15\\
\hline
\end{tabular}}
\caption{Average number of under-select and over-select significant variables from  transport map methods compared with MCMC over $100$ repetitions.}
\label{tab:sig_var_set}
\end{table}

\section{Discussion}\label{sec:discussion}
{We propose an efficient sampler for Bayesian posteriors with unknown normalizing constants by learning a transport map from a simple reference distribution to the target posterior. Once learned, the map enables independent sampling at negligible cost. Following prior transport-based approaches \citep{kim2013efficient,el2012bayesian}, the map is learned by minimizing the KL divergence between the transported and target distributions, thereby avoiding evaluation of the normalizing constant.
Our main contribution is a structurally constrained transport map class $\mathcal T$, motivated by optimal transport theory, which ensures unique recovery of the Monge map and reduces computational complexity. This structure also facilitates downstream inferential tools, including multivariate quantiles and ranks for Bayesian exploratory analysis.
We further extend the framework to mixed discrete–continuous settings and introduce simplified transport families tailored to specific regimes, including affine maps for approximately Gaussian posteriors and mean-field transport classes for latent variable models with large discrete spaces.}

One interesting future direction is to relax the invertibility assumption in \Cref{prop:equiv} and develop transport map estimators that better align with the geometric structure of the Wasserstein space. For example, when the target distribution concentrates near a low-dimensional manifold \citep{arjovsky2017wasserstein}, using a reference distribution whose dimension matches the intrinsic rather than the ambient dimension may reduce computational cost and improve interpretability. Replacing the KL divergence in the learning objective with alternative discrepancy measures, such as Stein discrepancies~\citep{gorham2015measuring,liu2016kernelized}, is another promising direction, as these avoid Jacobian determinant evaluations and may be more robust in high dimensions.

Another important direction concerns the combination of multiple transport maps via the maximum of potential functions, as in Lemma~\ref{lem:regularity_OT_2}. Unlike classical mixtures of couplings or mixture-of-experts constructions, where weights are produced by a separate gating mechanism and may collapse in moderate to high dimensions, the max-of-potentials formulation determines both the local map and its activation through the same potential function. This allows the method to adaptively select a small number of effective transport maps based on the geometry of the target distribution. Understanding the theoretical properties of this construction remains an interesting open problem.

Finally, we highlight a potential application of our framework to graphical models \citep{wainwright2008graphical}. While existing work has largely focused on discrete settings \citep{chen2020optimal,akagi2020probabilistic,haasler2021multi}, continuous and mixed-variable graphical models remain underexplored. Our approach offers a new perspective for addressing these more complex scenarios.

	{\fontsize{11.15}{11.15}\selectfont %
	\setlength{\bibsep}{0pt plus 0.24ex}
	
\bibliographystyle{abbrv}
\bibliography{bayes_ot}
}

\newpage
\appendix
\begin{center}
{\bf\Large Supplementary Materials to ``An Optimal Transport-Based Generative Models for Bayesian Posterior Sampling"}
\end{center}

The supplementary material is organized as follows. 
In \Cref{app:com}, we provide a detailed comparison of our method with existing transport-based approaches for Bayesian inference, including normalizing flows, transport Monte Carlo, triangular maps, and ICNNs, and revisit the simulation in \Cref{sec:center_outward_quantiles} to highlight differences in their inferential capabilities. In \Cref{sec:computation}, we present additional computational details related to the algorithms and numerical studies discussed in the main paper. In \Cref{sec:algorithm}, we provide further details on key algorithms. In \Cref{app:A}, we collect the proofs of all lemmas and theorems presented in the main paper. Finally, in \Cref{app:sen}, we conduct sensitivity analyses of our method, including the choice of embeddings in Bayesian latent variable models and the choice of the number of local potential functions $L$ when modeling the overall potential as the maximum of $L$ local potentials.

{
\section{Comparison with Other Transport-based Methods}\label{app:com}

In this section, we provide a more detailed comparison between our proposed method and other transport-based approaches for Bayesian inference. This includes a review of normalizing flows (NFs), transport Monte Carlo (TMC), and alternative methods for constructing transport maps, such as triangular maps \citep{el2012bayesian,marzouk2016sampling} and input convex neural networks (ICNN) \citep{geuken2025input}. Later, we revisit the simulation example in \Cref{sec:center_outward_quantiles} to provide a closer look at the differences in inferential ability among these methods.

\subsection{Review}\label{sec:review_other_methods}
First, we briefly review the TMC approach \citep{duan2021transport}, which is the more distinct from other transport-based methods. TMC constructs a random transport plan between a simple reference distribution $\mu$, which is default to be standard uniform, and the target distribution $\pi_n$ using an infinite mixture of simple invertible maps. It approximates the conditional distribution of $\pi(x\mid z)$, where $z\sim \pi_n$ and $x\sim \mu$, using a infinite mixture of inverse maps
\[
\tilde\pi(x\mid z) = \sum_{k=1}^\infty w_k(z) \delta\big\{x-T_k^{-1}(z)\big\},
\] 
where $\tilde\pi(x\mid z)$ is the approximate conditional distribution, $w_k(z)$ are weights that depend on $z$, $\delta\{\cdot\}$ is the Dirac delta and $T_k$'s are simple invertible maps (chosen as location-scale transforms). With this inverse maps, the reverse conditional becomes
\[
\tilde \pi_n(z\mid x) = \sum_{k=1}^\infty w_k\big(T_k(x)\big) \pi\big(T_k(x)\big)\abs{\det\big(\nabla T_k(x)\big)}.
\]
During the sampling stage, TMC first samples $Z_i\iid  \mu$ and then draws component $k$ with probability $w_k(T_k(Z_i))\abs{\det\big(\nabla T_k(x)\big)}$. Finally, it outputs $X_i = T_k(Z_i)$ as a posterior sample. The weights $w_k(z)$ and the parameters of $T_k$'s are learned by minimizing the KL divergence between $\tilde \pi_n(z\mid x)$ and the target posterior $\pi_n(z)$ using stochastic gradient descent.

Second, we review the normalizing flow (NF) approach for Bayesian inference. NFs construct a transport map by composing a series of simple, invertible transformations $\{f_k\}_{k=1}^K$, each with a tractable Jacobian determinant. The overall transformation is given by $T=f_K \circ f_{K-1} \circ \cdots \circ f_1$. Similarly to our method, NFs aim to minimize the KL divergence between the pushforward measure $T_\# \mu$ and the target posterior distribution $\pi_n$ with the objective function same as in \eqref{eqn:new_obj}. However, the focus of NF approaches is primarily on designing flexible and expressive transformation classes, for exampls planar flows \citep{rezende2015variational}, masked autoregressive flows (MAF) \citep{papamakarios2017masked}, and neural spline flows (NSF) \citep{durkan2019neural}, to name a few. These methods often rely on deep neural networks to parameterize the transformations, and more expressive classes usually lead to better approximation of the target distribution. However, this comes at the cost of increased computational complexity and potential challenges in training stability. They are not specifically designed to recover the optimal transport map, which may limit their interpretability and inferential capabilities compared to our method. Additionally, due to the invertibility requirement, NFs fail to capture distributions with discrete compoenents, such as mixture models with latent varaiables, and multi-modal distributions with disconnected supports.

Lastly, we discuss alternative methods for constructing transport map classes. One of the pioneering works is the lower triangular map proposed by \cite{el2012bayesian}, which takes the form 
\begin{align*}
    T(x) = \begin{pmatrix}T^1(x_1) \\ T^2(x_1, x_2) \\ \vdots \\ T^p(x_1, x_2, \ldots, x_p)\end{pmatrix},  \quad \text{for any }  x= (x_1, \ldots, x_p)^T \in \R^p,
\end{align*}
where $T^i$ represents output $i$ of the map. Under the additional constraint that each component of the transport map is strictly monotone along its own coordinate, this triangular map coincides with the Knothe--Rosenblatt (KR) rearrangement \citep{rosenblatt1952remarks,bonnotte2013knothe}. The KR rearrangement yields a triangular transport map constructed via a sequential matching of conditional distributions and is not, in general, optimal for the quadratic transport cost. Its triangular structure guarantees invertibility provided that each component $T^i$ is strictly increasing in $x_i$, but the resulting map depends explicitly on the ordering of the coordinates. We refer the reader to \cite{marzouk2016sampling} for further details on the construction and learning of triangular transport maps.

While both triangular maps and our approach (cf.\ \Cref{lem:splineConvexity}) involve integrals of univariate functions, their roles are fundamentally different. In triangular maps, the integral representation is used directly in each component of the transport map to enforce one-dimensional monotonicity along coordinate axes. By contrast, the class defined in \Cref{lem:splineConvexity} concerns the potential function of the transport map. We require the underlying univariate integrand functions $\varphi$ to be increasing and bounded in order to ensure convexity of the resulting potential. In our method, the integral form is not used to directly parameterize the transport map components, but rather the potential function, which is only needed when evaluating the maximum of local potential functions.

This distinction leads to several advantages of our transport class over triangular parameterizations. First, our representation permits efficient computation of both the potential and its derivatives, including gradients and Hessians. This is particularly important for multimodal distributions with well separated components, where taking the maximum over local potentials allows the map to adapt naturally to complex geometry. Second, our construction is fully symmetric and does not impose axis aligned monotonicity or a predefined ordering of the coordinates. Third, our transport map admits a natural variational interpretation as minimizing the transport cost of pushing forward the reference measure to the target distribution. Finally, our framework flexibly accommodates continuous, discrete, and mixed parameter spaces. When the target distribution contains both discrete and continuous components, the optimal transport map may fail to be invertible, in which case monotone triangular maps are no longer directly applicable, whereas our approach remains valid.

Another related approach is the input convex neural network (ICNN) proposed by \cite{amos2017input}, which proposes a convex neural network architecture with non-negative weights and convex activation functions. One can alternatively parameterize the convex potential function $u(x)$ using ICNNs. The transport map is then obtained as the gradient of this convex potential, i.e., $T(x) = \nabla u(x)$. For ICNNs, the convexity requirement of activation functions, which leaves basically choices of Softplus and ReLU, and non-negativity constraints on weights may limit the expressiveness of the potential function, potentially affecting the quality of the transport map. In contrast, our method employs a more flexible basis expansion approach to approximate convex functions, allowing for a richer representation of the potential function and potentially leading to better approximation of the target distribution. The simulation examples in the next subsection further illustrate this point.

\subsection{Center-outward Quantiles Comparison}\label{sec:additional_quantile}

In this subsection, we provide an additional example to illustrate the effectiveness of our method in obtaining center-outward quantiles for complex distributions. In addition to the Planar, we also include ICNN \citep{amos2017input}, the triangular map \citep{el2012bayesian,marzouk2016sampling} and Neural Spline Flows (NSF) \citep{durkan2019neural} as alternative transport map methods for comparison. TMC is excluded from this comparison since it relies on a standard multivariate uniform reference distribution, which is not spherical uniform \citep[as defined in][]{chernozhukov2017monge}.

We consider the banana-shaped distribution from \cite{hallin2021distribution}, defined as:
$\frac{3}{8}\mN (-8 m_h, \Sigma_1) + \frac{3}{8} \mN (8 m_h, \Sigma_2 )+\frac{1}{4}\mN(-5 m_\nu,\Sigma_3)$,
where $m_h=(1,0)^T$, $m_\nu= (0,1)^T$, and $\Sigma_1=(5,-4;-4,5)$, $\Sigma_2=(5,4;4,5)$, $\Sigma_3 = (4,0;0,1)$.
This distribution is particularly challenging due to its non-convex shape, and meaningful quantile contours must accurately preserve this structure. In this example, we use our proposed method with $L=1$, and ICNN with softplus actiation and a single layer of 32 hidden units. For the triangular map, we use integrated exponential basis functions of order 5, and for NSF, we use 5 coupling layers with 128 hidden units in each layer.

\begin{figure}[!ht]
    \centering
    \subfloat[Reference $\mathcal{N}(0, I_2)$]{%
        \includegraphics[width=0.32\textwidth]{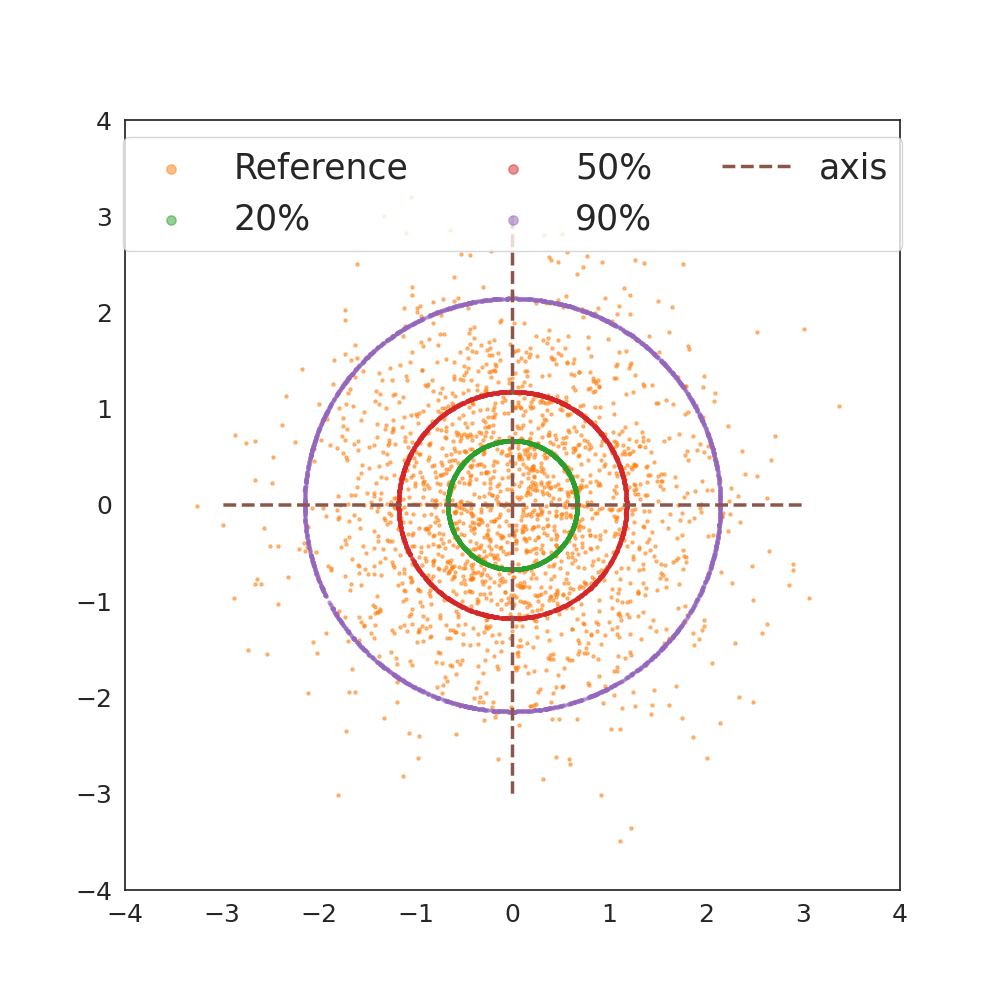}
    }
    \subfloat[OT]{%
        \includegraphics[width=0.32\textwidth]{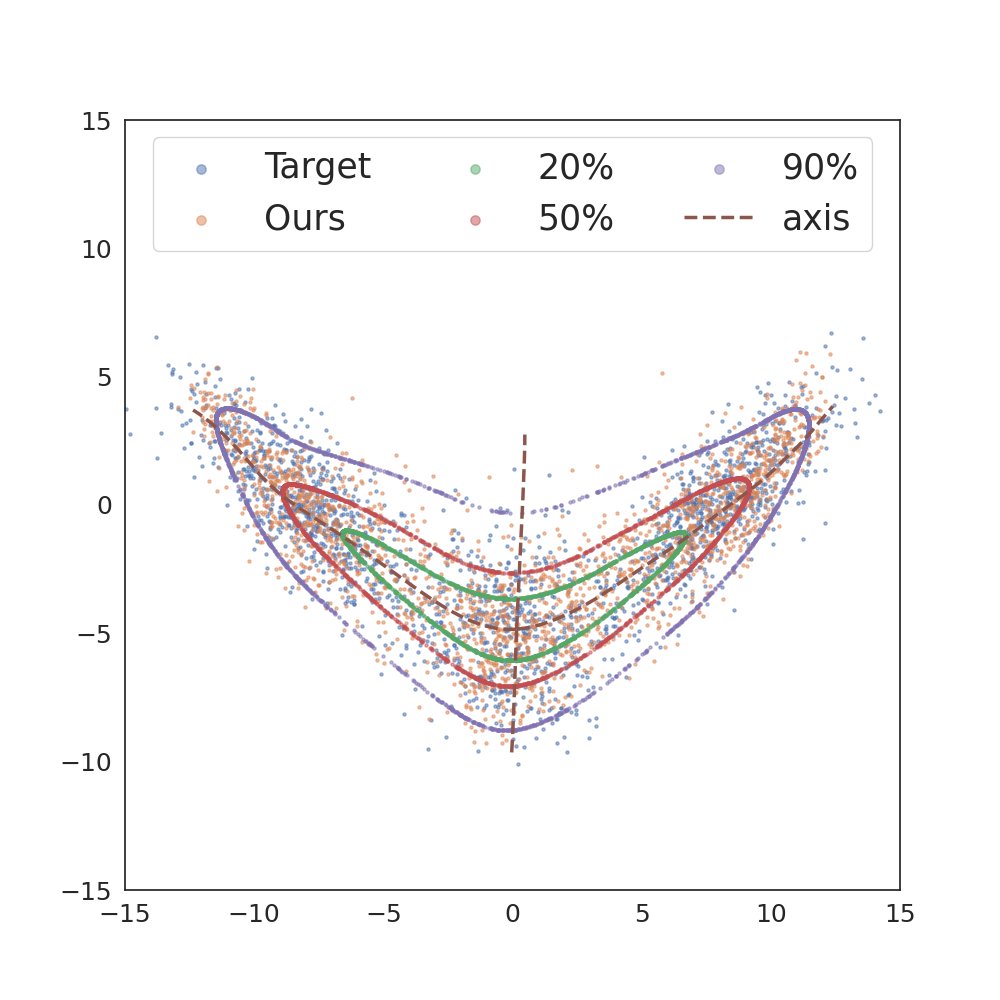}
    }
    \subfloat[Planar]{%
        \includegraphics[width=0.32\textwidth]{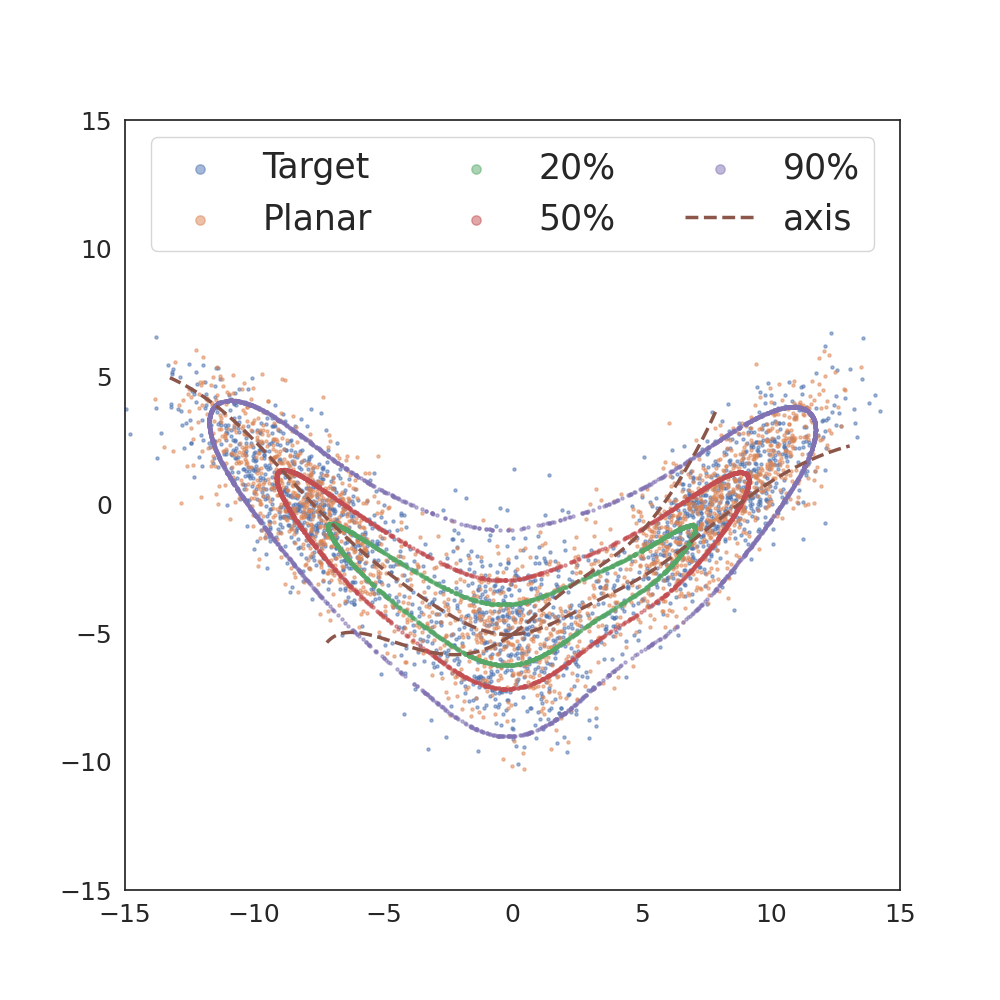}
    }
    \\
    \subfloat[ICNN]{
        \includegraphics[width=0.32\textwidth]{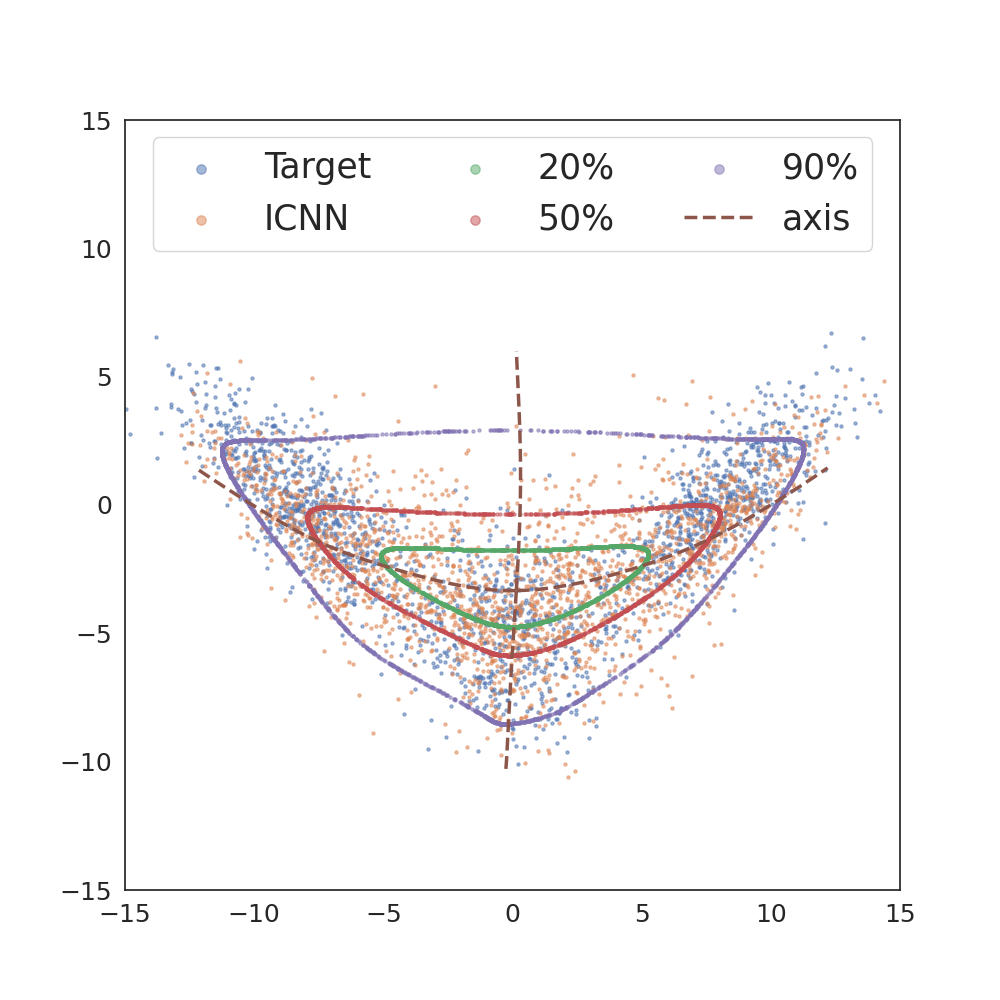}
    }
    \subfloat[Triangular Map]{
        \includegraphics[width=0.32\textwidth]{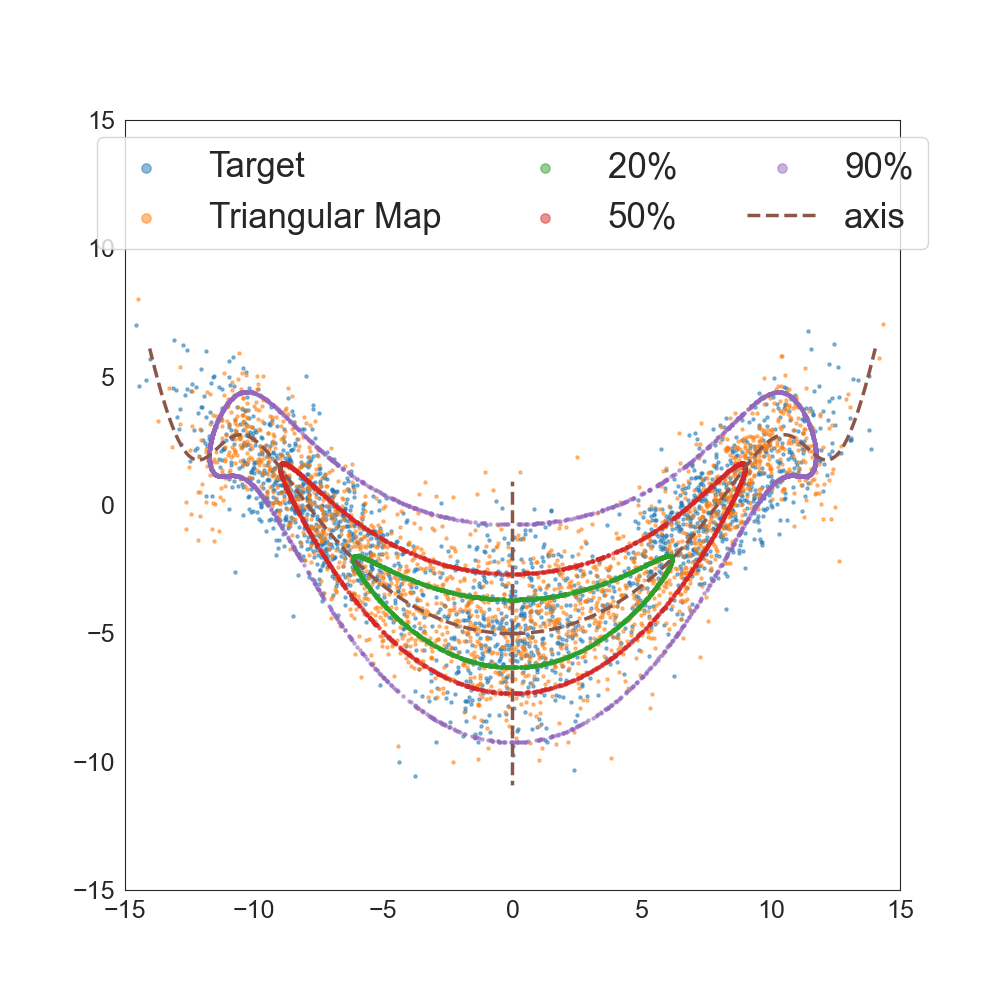}
    }
    \subfloat[NSF]{
        \includegraphics[width=0.32\textwidth]{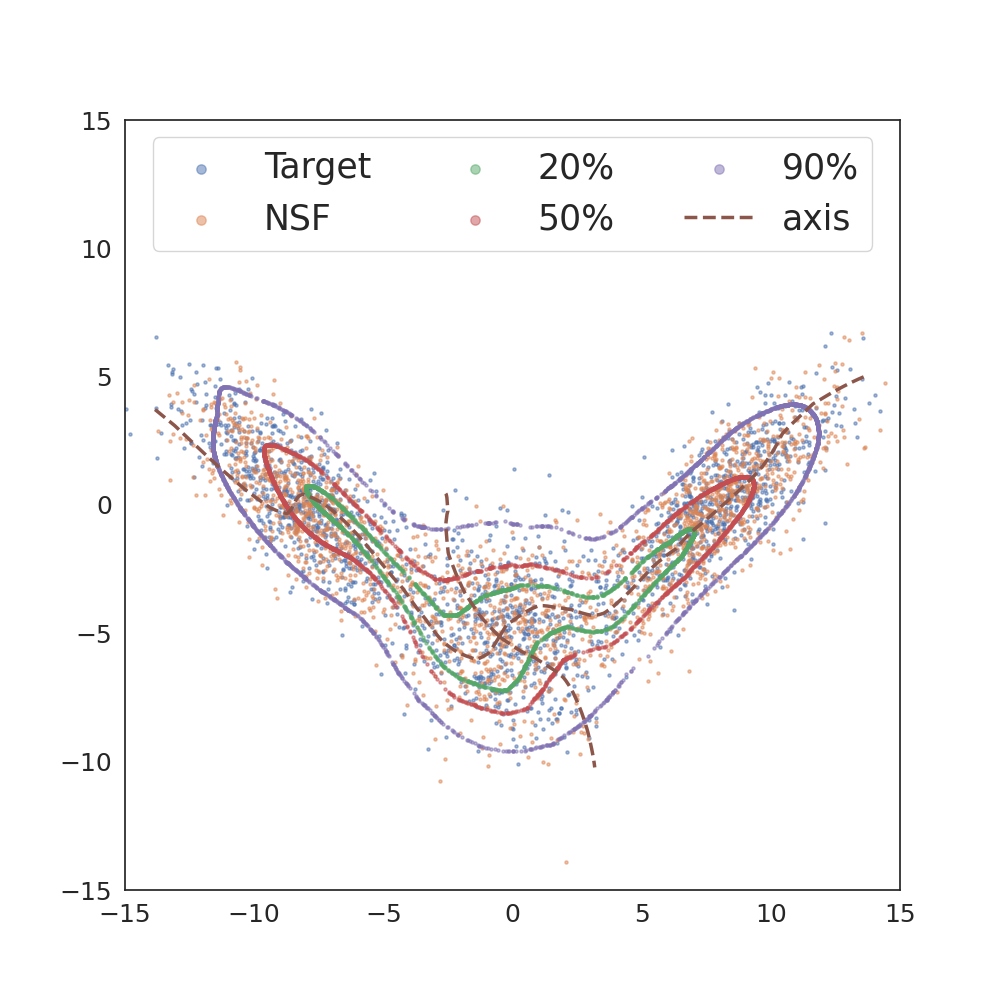}
    }
    \caption{Center-outward quantile contours for the Banana distribution. Here we include the plot from the reference distribution $\mathcal{N}(0, I_2)$ and the corresponding contours by passing the reference contours through the transport map estimated by our method (with $L=1$), Planar, ICNN with softplus activation, Triangular Map and Neural Spline Flows (NSF). We also include axes of the reference distribution to show directional information.}
    \label{fig:banana}   
\end{figure}

As shown in \Cref{fig:banana}, all methods produce samples that resemble the shape of the true distribution. However, our method, which employs a non-crossing map, preserves the non-convex shape of the quantile contours much more effectively. In addition, our method provides more accurate directional information: the probability mass between each pair of axes is expected to be equal, and only our approach successfully partitions the target distribution into four regions of equal probability. 
For ICNN and the Triangular Map, their sign curves are roughly the correct directions but the quantile contours cannot fully capture the geometry of the banana distribution. In contrast, although the Planar and NSF  capture the general shape of the quantile regions, the distorted sign curves reveal a loss of directional information, which may lead to incorrect inferences about conditional dependencies among parameters.

\begin{table}[!ht]
\centering
\begin{tabular}{c|c|c|c|c}
\toprule 
Ours & Planar & ICNN & Triangular Map & NSF\\
\midrule
 1.288$\times 10^{-6}$& 0.123 &1.502 & 0.233& 0.110\\
\bottomrule
\end{tabular}
\caption{W2 distance between samples from the approximated distribution and the true distribution in the banana example.}
\end{table}

For the mixture of two Gaussians distribution in \Cref{fig:two_balls}, we provide more details on implementation details here. We set up ICNN with ReLU activation and a single layer of 32 hidden units, since ReLU is more suitable for handling the separated components in this distribution. For the triangular map, we use integrated exponential basis functions of order 10. Additionally, we also consider the NSF with  6 coupling layers and 128 hidden units per layer and results are plotted in \Cref{fig:two_balls_additional}. NSF is performing significantly better than Planar in this example, but its sign curves are still partially distorted.

\begin{figure}[!ht]
    \centering
    \includegraphics[width=0.3\textwidth]{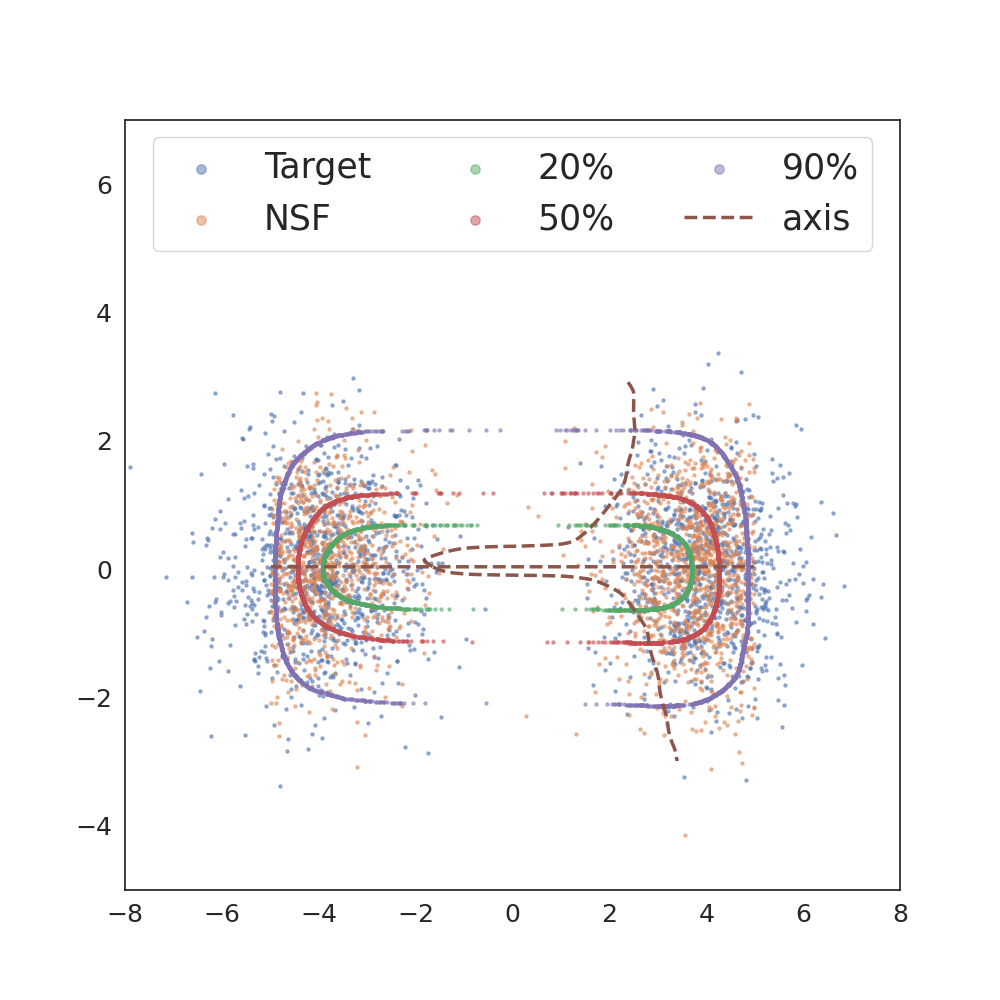}
    \caption{Center-outward quantile contours from NSF for the mixture of two Gaussians distribution.}\label{fig:two_balls_additional}
\end{figure}

\subsection{Connection to Variational Inference}\label{sec:VI_connection}

Our method shares similarities with variational inference (VI) approaches \citep{blei2017variational} in that both aim to approximate the target posterior by minimizing the Kullback--Leibler divergence between a parametric approximation and the posterior. We also employ a mean-field approximation in \Cref{sec:gmm} to simplify the optimization, which is common in VI.

There are, however, several key differences. First, although the pushforward family $\{T_\# \mu : T \in \mathcal{T}\}$ can be viewed as an implicit variational family parameterized by the transport map class $\mathcal{T}$, our method does not directly parameterize the density of the approximate distribution. Consequently, the approximation error is governed by how well the transport map class approximates the optimal transport map, rather than by the expressiveness of an explicit density family. Second, our construction leverages optimal transport theory to design structurally constrained map classes that admit a unique solution. In multimodal settings with well-separated components, this structure yields a principled mechanism for determining the contribution of local transport maps based on the geometry of the target distribution, whereas VI typically relies on mixture models whose weights are learned through optimization.

Finally, while we introduce a mean-field approximation for Bayesian latent variable models, this approximation is imposed on the \emph{potential function} rather than on the density itself. As a result, the induced joint distribution can still retain certain dependencies among variables, in contrast to the fully factorized distributions arising from classical mean-field variational families; see the discussion in the last paragraph of Section~\ref{sec:gmm}. Moreover, our approach yields a generative model that supports additional inferential tasks, such as center-outward quantiles and ranks, beyond posterior approximation alone.

\subsection{Identifiability and Training Stability Comparison}\label{sec:training_stability}

In addition to the inferential capabilities, we also compare  the identifiability and training stability of our method with alternative transport map approaches. 

For triangular maps, under additional structural constraints such as requiring each component of the transport map to be strictly monotone along its own coordinate, the resulting map coincides with the Knothe--Rosenblatt (KR) rearrangement \citep{rosenblatt1952remarks,bonnotte2013knothe}, which is uniquely defined for a fixed ordering of the coordinates. Under this setting, triangular maps are identifiable.

However, this notion of identifiability is intrinsically tied to the choice of coordinate system and variable ordering, since monotonicity is enforced only along coordinate axes. In contrast, the shape constraint in our approach is imposed through convexity of the transport potential, equivalently monotonicity of the gradient of a convex function, which is defined independently of any particular coordinate system. Consequently, the proposed map class is fully symmetric and does not rely on a predefined ordering of variables.

We provide a simple illustration on the banana distribution. We experiment with 3 different polynomial orders and two coordinate systems. The first coordinate system is the original one, and the second coordinate system is obtained by flipping the two axes in the distribution, i.e., $(x_1, x_2)$ becomes $(x_2, x_1)$. This is equivalent to changing the order of the coordinates in the triangular map. 

\begin{figure}[!ht]
\centering
\subfloat[order=2, orginal]{
    \includegraphics[width=0.3\textwidth]{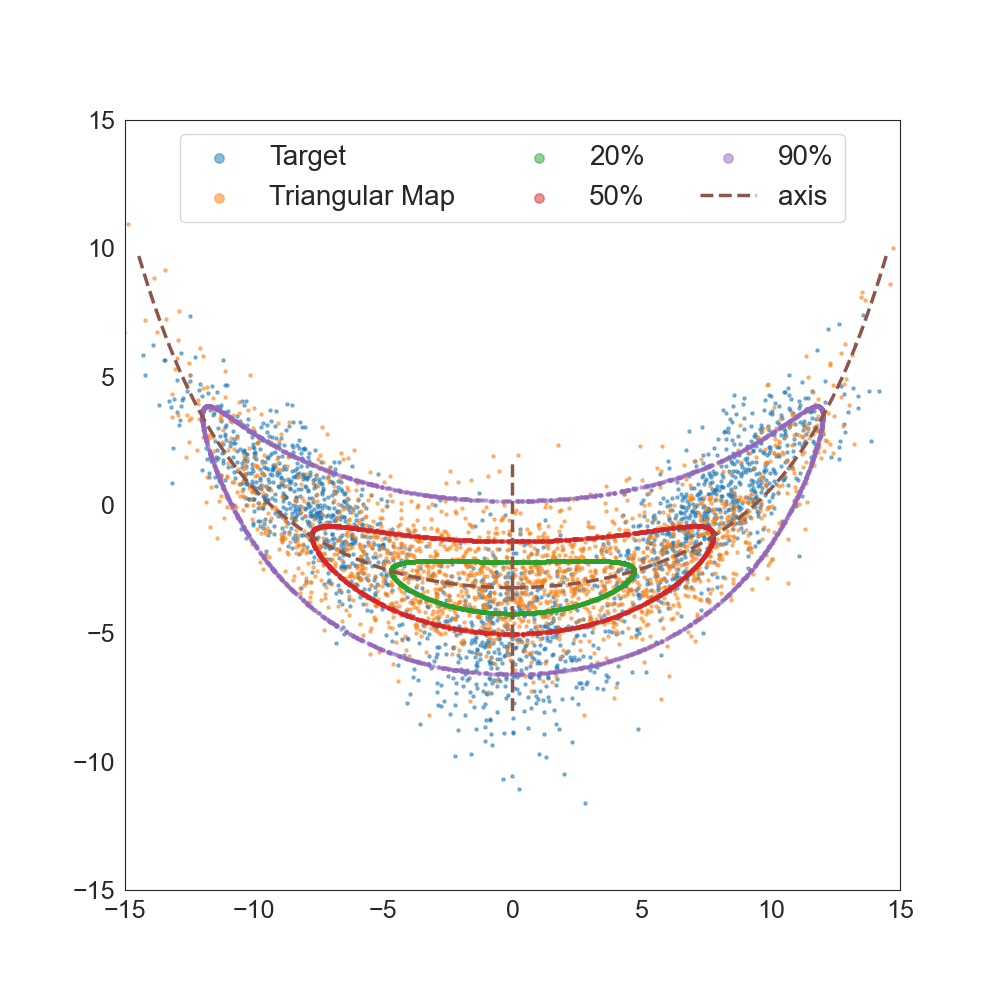}
}
\subfloat[order=5, orginal]{
    \includegraphics[width=0.3\textwidth]{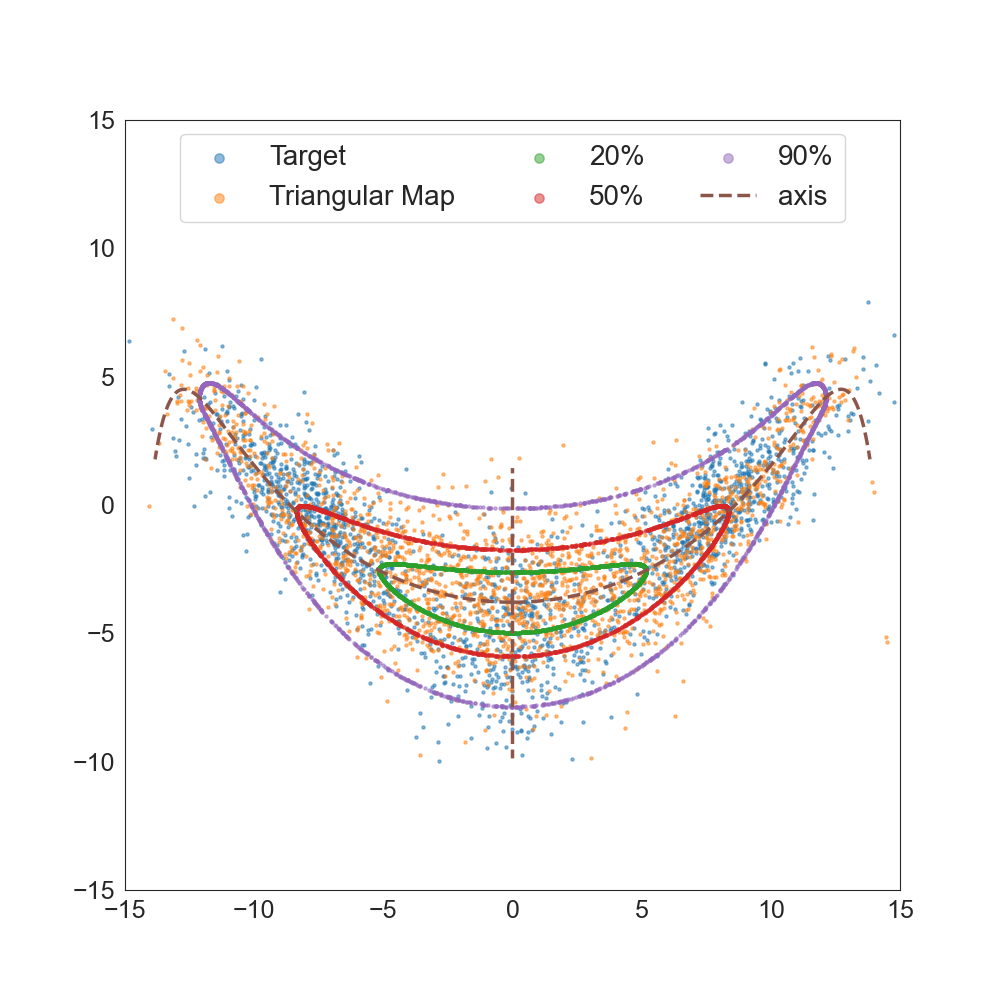}
}
\subfloat[order=10, orginal]{
    \includegraphics[width=0.3\textwidth]{fig/banana_trimap_order10.png}
}

\subfloat[order=2, flipped]{
    \includegraphics[width=0.3\textwidth]{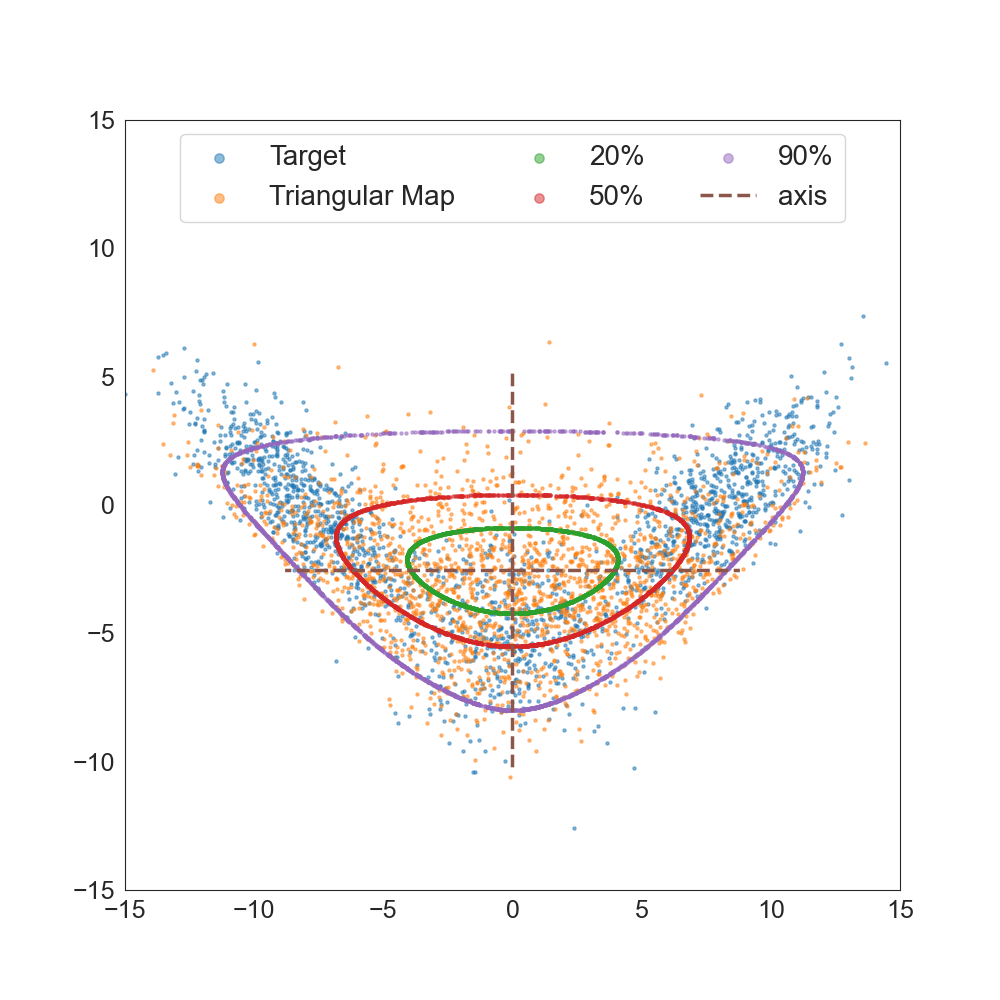}
}
\subfloat[order=5, flipped]{
    \includegraphics[width=0.3\textwidth]{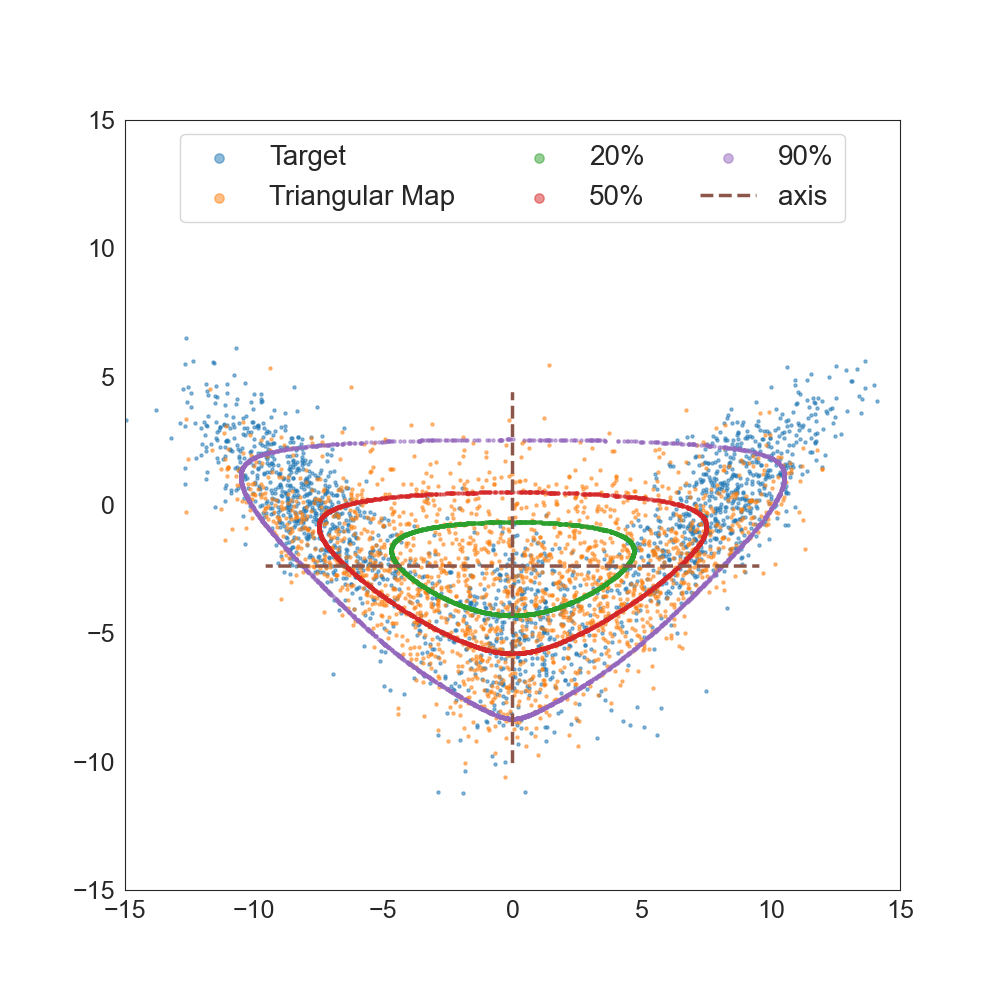}
}
\subfloat[order=10, flipped]{
    \includegraphics[width=0.3\textwidth]{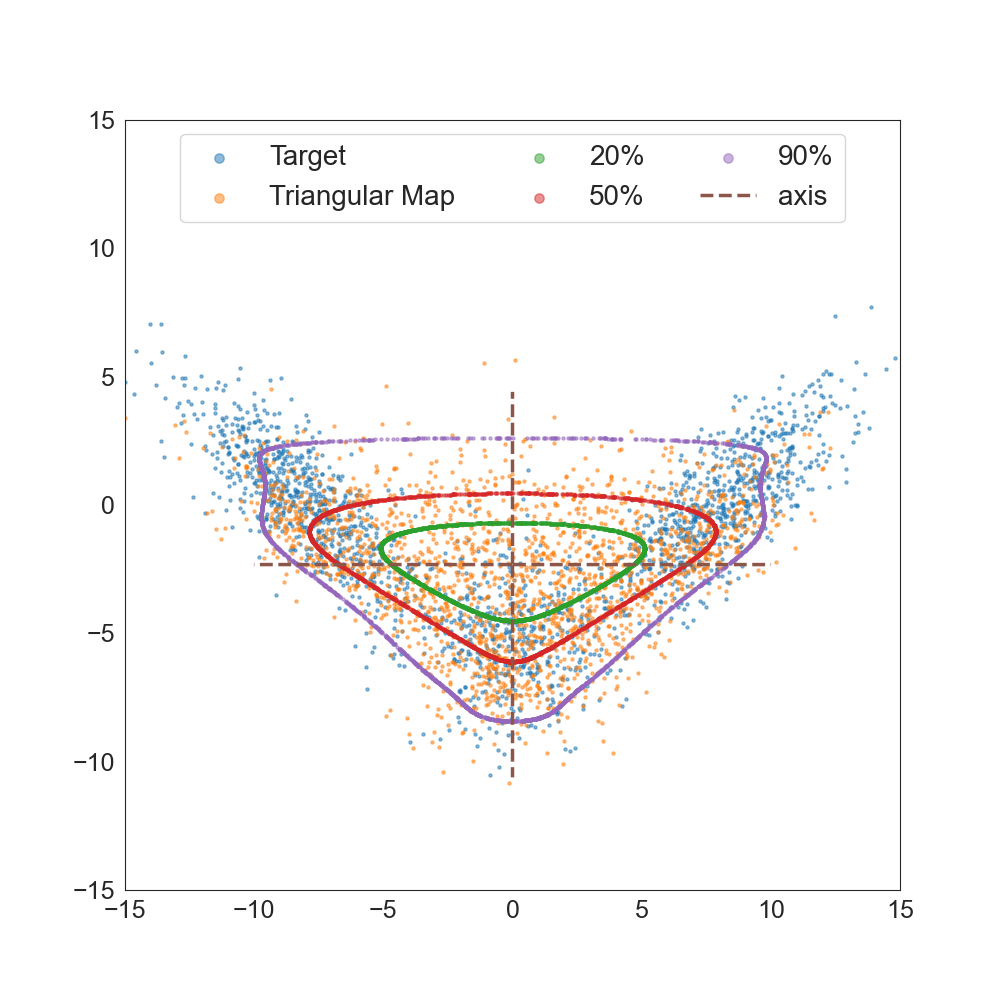}
}
\caption{Results from triangular map under original and flipped coordinate systems on the Banana distribution.}\label{fig:banana_trimap}
\end{figure}

We find that the resulting triangular maps are drastically different, as shown in \Cref{fig:banana_trimap}. First, the triangular map with different orders leads to different approximations of the target distribution, and the approximation quality is sensitive to the choice of polynomial order. Lower-order maps fail to fully adapt
to the geometry, while higher-order constructions tend to overfit in the tail regions and corners. Second, the triangular map under the flipped coordinate system fails to capture the geometry of the target distribution, even with a high polynomial order. This happens partially because the conditional distribution of axis 1 given axis 2 has disconnected support in certain regions (e.g. when axis 2 is above 0), while the conditional distribution of axis 2 given axis 1 is always continuous. As a result, the triangular map struggles with the disconnected support and fails to capture the geometry of the target distribution under the flipped coordinate system.

For alternative approaches that do not impose explicit shape constraints on the transport
map, such as normalizing flows, they are in general not identifiable, which may lead to increased
instability in training. For simplicity, we illustrate this point again with Planar flow on the mixture of two bivariate Gaussians example in \Cref{fig:planar_nf_diff_init}. In this example, we use the same Planar flow with same architectures and same optimization schedule. The only difference is the initialization of the parameters. While the two maps both push forward the reference distribution to a distribution close to the target, the maps themselves are drastically different, as indicated by the center-outward quantiles and sign curves shown.

\begin{figure}[!ht]
\centering
\subfloat[Initialization 1]{
    \includegraphics[width=0.3\textwidth]{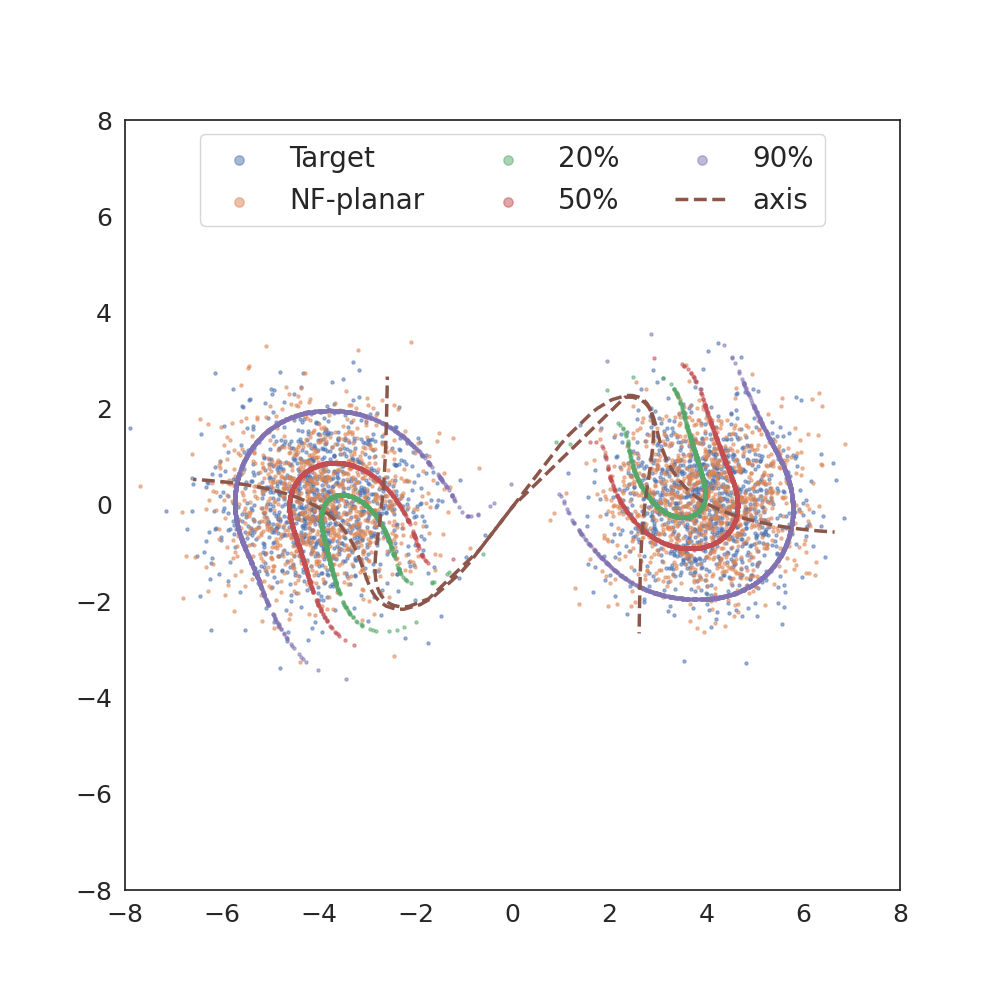}
}
\subfloat[Initialization 2]{
    \includegraphics[width=0.3\textwidth]{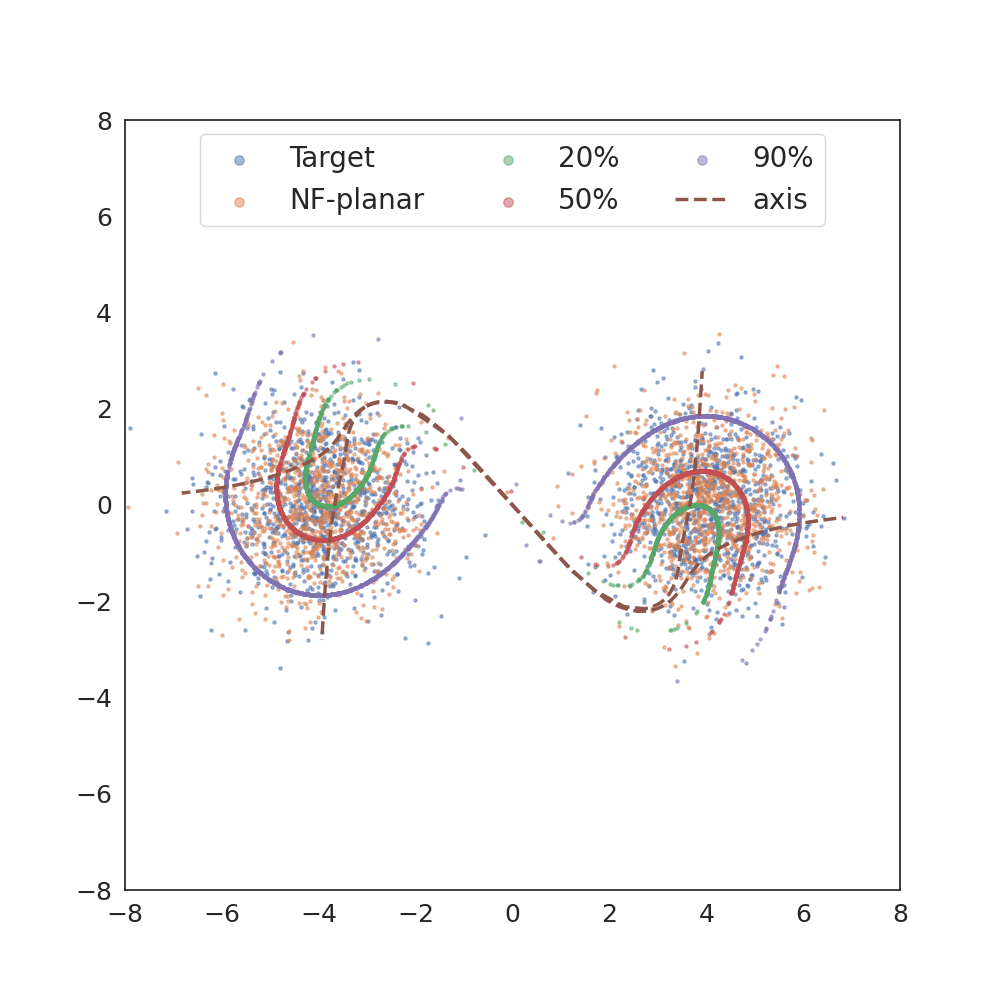}
}
\caption{Results from Planar flow with different initialization on the mixture of two Gaussian distribution}\label{fig:planar_nf_diff_init}
\end{figure}

Finally, we note that for methods that impose explicit shape constraints, such as our method, ICNN and triangular maps, the optimization landscape is more well-behaved. In contrast, for methods without explicit shape constraints, the optimization landscape can be more complex and may contain multiple local minima, which can lead to instability in training. For example, in the mixture of Gaussian example in \Cref{sec:mixture_normal_simu}, we apply the same sinkhorn initialization to all methods to ensure good initialization (details in \Cref{sec:simulation_config,sec:algorithm}). However, due to the complexity of the optimization landscape for NF, different runs with the same initialization can still lead to different results, and we observe that for MAF, while it is more expressive than Planar flow, can often perform worse than the initialization stage after the KL training. We provide such an example below in  \Cref{fig:maf_sinkhorn_kl}, where the KL training leads to a worse approximation of the target distribution compared to the initialization stage, which is likely due to the optimization getting stuck in a local minimum.

\begin{figure}[!ht]
\centering 
\subfloat[Sinkhorn initialization]{
    \includegraphics[width=0.3\textwidth]{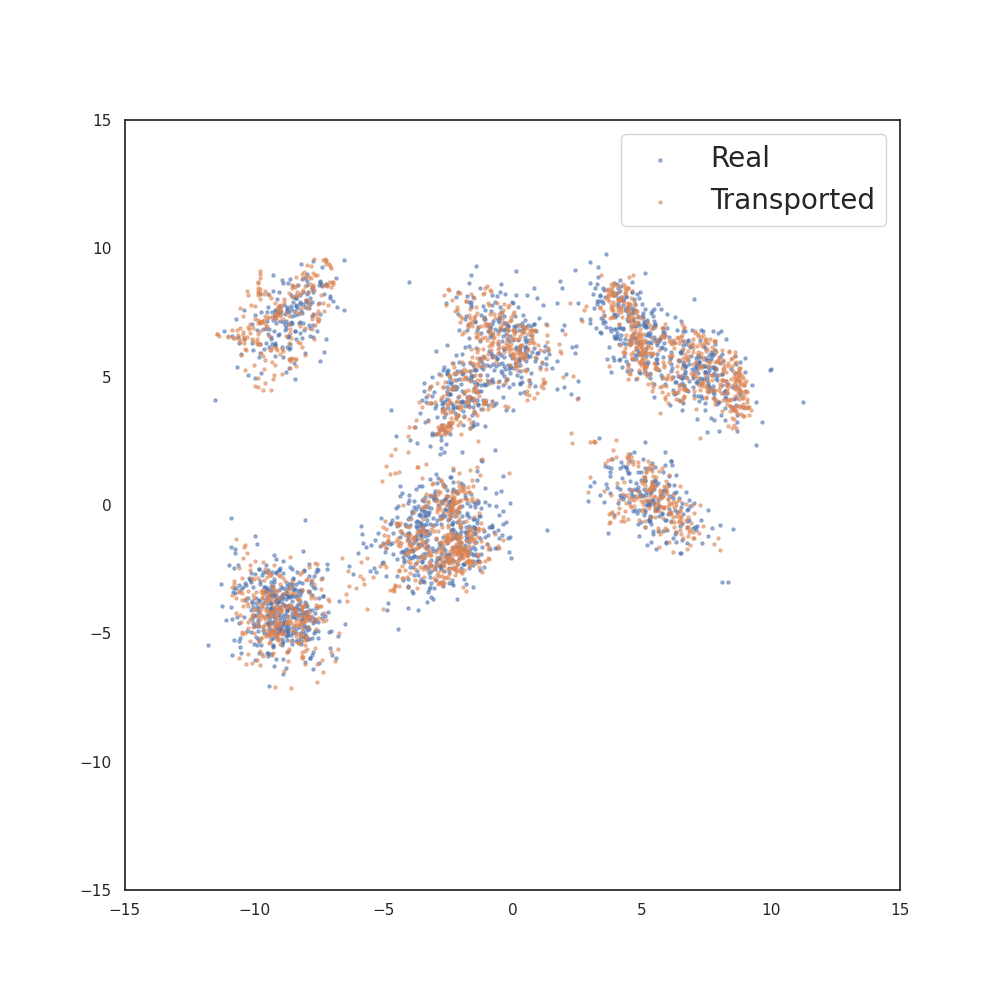}
}
\subfloat[KL after the initialization]{
    \includegraphics[width=0.3\textwidth]{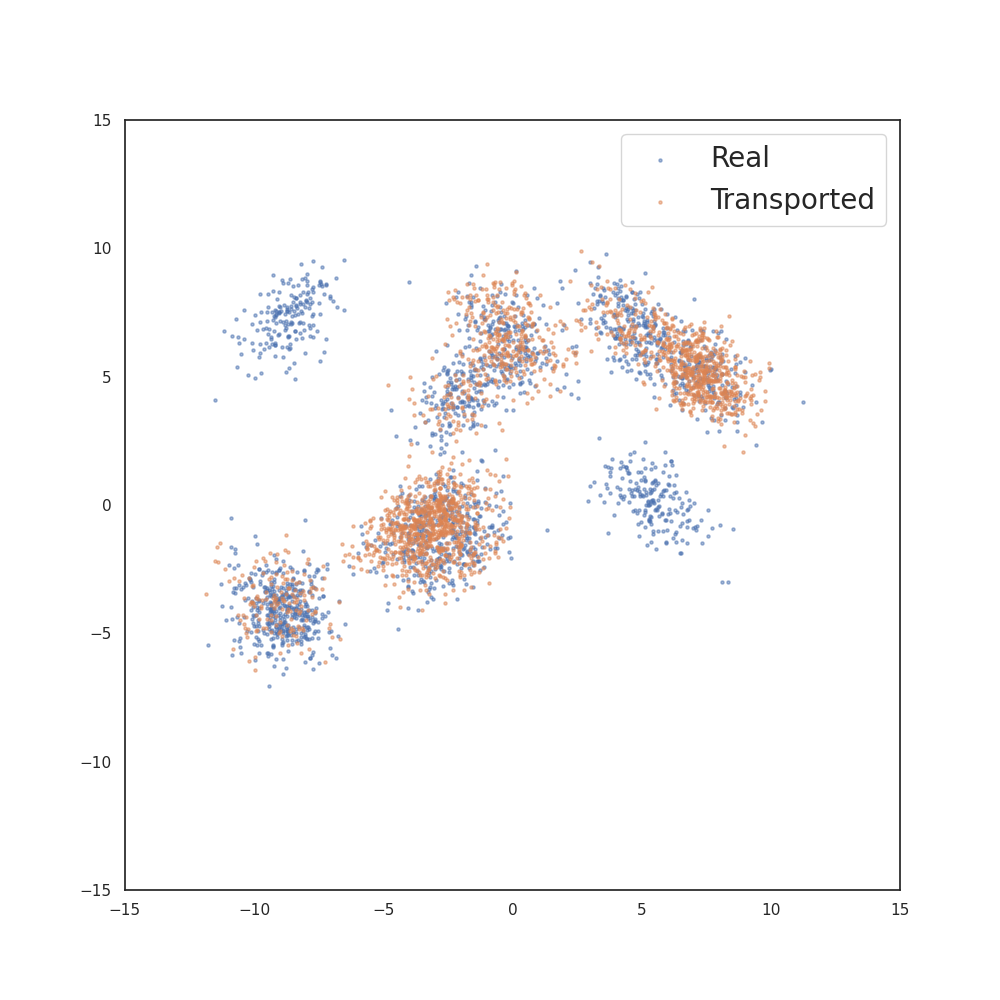}
}
\caption{Results from MAF in the same run after initialization and after KL training. We consider the example of mixture of 10 five-dimensional Gaussians here.}\label{fig:maf_sinkhorn_kl}
\end{figure}

}

\section{Computational Details}\label{sec:computation}

\subsection{Difficulties of the pushforward constraint and infimum attainability}\label{app:constraint}

{Enforcing the constraint $\pi_n = T_\# \mu$ is difficult both analytically and computationally. When the posterior $\pi_n$ is known only up to a normalizing constant, writing its unnormalized density as $\tilde{\pi}_n$, the pushforward condition can be written as
\[
\tilde{\pi}_n(T(x)) \lvert \det \nabla T(x) \rvert = Z\, \mu(x),
\]
where $Z$ is an unknown normalizing constant. Without knowledge of $Z$, this constraint cannot be evaluated or enforced directly.

More generally, even when $Z$ is known, the pushforward condition defines a nonlinear change-of-variables equation for the map $T$, involving the determinant of its Jacobian. Enforcing such a constraint amounts to solving a high-dimensional nonlinear equation that couples the values of $T$ and its derivatives everywhere on the support of $\mu$. This becomes computationally challenging in high dimensions or when $T$ has a complex structure, for example in multimodal settings or when the map is implicitly defined through maxima of multiple potential functions, as in our mixed-parameter construction. In these cases, direct enforcement of the pushforward constraint is computationally prohibitive.

The infimum in~\eqref{eqn:KL_Form} is in fact a minimum. The Kullback–Leibler divergence is nonnegative and achieves its minimum value of zero if and only if $T_\# \mu = \pi_n$. Moreover, our chosen function class $\mathcal{T}$ is sufficiently rich to contain the optimal transport map from $\mu$ to $\pi_n$, whose existence is guaranteed by Lemma~\ref{lem:regularity_OT}. As a result, the infimum is attained at this optimal transport map.
}

\subsection{Inverse Map and Ranking Posterior Draws}\label{sec:inverse_map}

As mentioned in \Cref{sec:center_outward}, a key computational advantage of this optimal transport framework for Bayesian inference is its ability to efficiently compute the inverse OT map $(T^*)^{-1}$. This leverages the fact that the  $(T^*)^{-1}$ is related to the conjugate function of its potential $u^*$ via the identity $\nabla u^{\dagger} = (\nabla u^*)^{-1}$. Therefore, for any posterior sample $Z$, its preimage $X = (T^*)^{-1}(Z)$ can be computed as 
\[
X= \argmax_{x\in \R^p}  \big\{\dprod{Z,x} - u^*(x)\big\}.
\]
We solve for $X = (\wh T)^{-1}(Z)$ using gradient descent: 
\begin{equation}\label{eq:gradient_descent}
X^{(t+1)}= X^{(t)} - \alpha_t \big(\nabla \widehat u(X^{(t)})- Z\big) = X^{(t)} - \alpha_t \big(\wh T(X^{(t)})- Z\big),\quad t=0,1,\ldots.
\end{equation}
Since the objective function is convex, as $\langle Z, x \rangle$ is linear in $x$ and $u^*(x)$ is convex, the gradient descent algorithm converges globally (and exponentially if $u^*$ is strongly convex). This allows us to efficiently compute the pointwise inverse transformation. We denote the solution as $X = \wh R(Z)$. Returning to the problem of assessing the plausibility of $\theta = \theta_0$, this is equivalently to checking whether $\|R(\theta_0)\|^2_2\leq q_{\chi^2_d}(0.95)$. We illustrate this application in \Cref{sec:logistic}.

 The inverse map is also useful for simultaneously comparing or ordering the plausibility of multiple parameter values via OT-derived ranks. For example, we may adopt the center-outward rank $R_\pm(\cdot)$ following \cite{hallin2021distribution}, which satisfies: 
\[
R_{\pm} (Z_1)\geq R_{\pm}(Z_2) \quad  \text{if} \quad \norm{\wh R(Z_1)}_2 \geq \norm{\wh R(Z_2)}_2.
\]
In other words, a point $Z_1$ is ranked higher than $Z_2$ if its preimage under the inverse map is closer to the origin in the reference distribution. This is equivalent to saying that $Z_1$ belongs to a lower quantile contour than $Z_2$.
\Cref{fig:banana_rev} visualizes these ranks for 10 random draws from the banana distribution. The reverse-mapped points in the reference Gaussian clearly preserve the intended ordering structure.

\begin{figure}[!ht]
    \centering
    \includegraphics[width=0.6\textwidth]{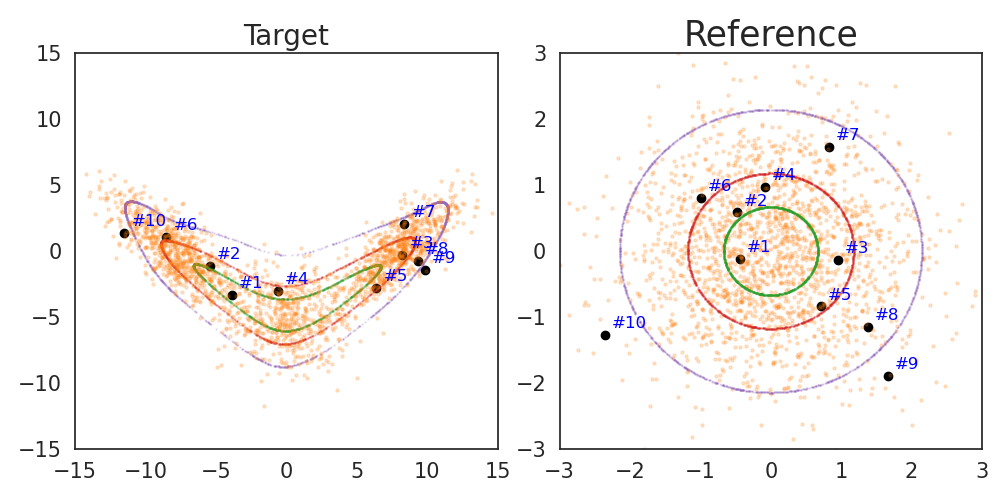}
    \caption{Center-outward ranks for 10 random draws \( Z_i \) from the banana distribution in \Cref{sec:additional_quantile}, and their corresponding locations \( \widehat{h}(Z_i) = (\widehat{T})^{-1}(Z_i) \) in the reference distribution.}
    \label{fig:banana_rev}
\end{figure}

\subsection{Computational Details of Gradients and Jacobians}\label{sec:computation_J}
While automatic differentiation packages help us avoid manually coding backpropagation, certain gradients are still computed by hand. For example, the baseline appears as an input to the exact or approximate objective function in Lemma~\ref{prop:equiv}. First, for a convex basis unit $u_i(x)= f(x; \alpha_i, \beta_i, w_i, v_i)$, we can compute the derivative as
\begin{align}\label{eq:gradU}
\nabla u_i(x)=  \varphi_i\big(\dprod{\alpha_i, x}+w_i\big)\alpha_i +\beta_i.
\end{align}
\noindent The exact derived transform and Jacobian matrix of map in \eqref{eq:grad_max_of_convex} follow from the lemma below.

\begin{lemma}\label{lem:max_jacobian} 
Let $\{u_k\}_{k=1}^L$ be $L$ twice differentiable and convex functions over $\R^p$, and $\iota(x) := \argmax_{k\in[L]} u_k(x)$ for each $x\in\R^p$. Then the gradient $T(x)=\nabla \max_{k\in[L]} u_k(x)$ is $\nabla u_{\iota(\theta)}(x)$, and the Jacobian matrix of $T$ at $x$ is $J_T(x) = \nabla^2 u_{\iota(x)}(x)$.
\end{lemma}

\noindent This lemma states that to compute the gradient of the maximum over a set of convex functions, one first identifies the convex function with the largest evaluation --- determined partly by the intercept $v$ in the convex unit, which also governs the separation of local regions associated with each local potential function $u_k$ --- and then takes the gradient of the selected convex function.
Ideally, the number of convex functions $L$ would match the number of distinct regions or modes of the target distribution $\pi_n$. When the number of modes is unknown, we can begin with a small value of $L$ and increase it as needed. For practical implementation, the $\max$ operation is replaced by a Softmax approximation to ensure differentiability. The sharpness parameter controls a trade-off between smoothness and approximation accuracy: smaller values yield smoother gradients but may slow optimization due to vanishing gradients, while larger values increase the approximation error relative to the hard maximum and can hinder accurate selection among competing potentials.

\subsection{Computation Details of \Cref{lem:mixture_KL}}\label{sec:mixture_KL_computation}
  We first generate $N_x$ samples $\{x_1, x_2,\ldots, x_{N_x}\}$ i.i.d. from $\mu$, and then compute the corresponding transformed variables $\{\theta_1,\theta_2,\ldots,\theta_{N_x}\}$ through transport map $T$, where $\theta_i=(\tau_i,\,\zeta_i)$ with $\zeta_i = \nabla \phi_{\tau_i}(X_i^{(2)})$ and 
\begin{align*}
   \tau_i= \underset{k\in [K]}{\arg \max} \left\{ \dprod{ x^{(1)}_i, b_k} + \phi_{k}(x^{(2)}_i)\right\}, ~~~ \text{for}~i=1,2,\ldots,{N_x}.
\end{align*}
To approximate generally intractable conditional probability $\mathbb{P}(\tau_i\mid x^{(2)}_i)$ in objective \eqref{eq:mixture_KL_opt} for each sample $(\tau_i,\,x_i^{(2)})$, we apply the Monte Carlo method again.  We generate $N_z$ samples $\{z^{(1)}_{i,1}, z^{(1)}_{i,2}, \ldots, z^{(1)}_{i,N_z}\}$ i.i.d.~from $\mu_1$ for each sample $x^{(2)}_i$, and computing the transported variables $\{w^{(1)}_{i,1}, w^{(1)}_{i,2}, \ldots, w^{(1)}_{i,N_z}\}$ as 
\begin{align*}
    w^{(1)}_{i,j} = \underset{k\in[K]}{\arg \max} \left\{ \langle z^{(1)}_{i, j}, b_{k}\rangle + \phi_{k}(x^{(2)}_i)\right\}, ~~~ \text{for}~j=1,2,\ldots,N_z.
\end{align*}
Then probability $\mathbb{P}(\tau_i\mid x^{(2)}_i)$ can be approximated by the empirical distribution,
\begin{align*}
    \widehat{\mathbb{P}}(\tau_i\mid x^{(2)}_i) := \frac{1}{N_z}\sum_{j=1}^{N_z} \bm{1}_{\{w^{(1)}_{i,j} = \tau_i\}}
    \approx \frac{1}{N_z} \sum_{j=1}^{N_z} W_{i,j}(\tau_i),
\end{align*}
where $W_{i,j}$ denotes the softmax approximation of the indicator function, i.e.,
\begin{align*}
    W_{i,j}(\tau_i) = \frac{\exp\left(\gamma(\dprod{z^{(1)}_{i, j}, b_{\tau_i}} + \phi_{\tau_i}(x^{(2)}_i)) \right)}{\sum_{k\in [K]}\exp\left(\gamma (\dprod{ z^{(1)}_{i, j}, b_{k}} + \phi_{k}(x^{(2)}_i))\right)},
\end{align*}
with $\gamma>0$ a tuning parameter. Here we use the softmax approximation to make the objective function differentiable, so that efficient gradient descent  algorithms can be applied to the optimization. In practice, we may reuse the samples $\{x_1^{(1)},x_2^{(1)},\ldots,x_{N_x}^{(1)}\}$ when generating the $\tau_i$'s as $\{z^{(1)}_{i,1}, z^{(1)}_{i,2},\ldots, z^{(1)}_{i,N_z}\}$ to facilitate computational efficiency and reduce space complexity.

\subsection{Computation in Gaussian mixture models}\label{app:comp_GMM}
Here we describe the optimization details of the transport map $T$ in \Cref{sec:gmm}.

To further improve the computational scalability of the convex functions $\phi_{k}^{(i)}$ with respect to the sample size $n$, we may set the relative weight parameter $\kappa$ in (MPm) to $n^{-1}$. With this choice, the second equation in~\eqref{Eqn:GMM_map} for generating $\bm \zeta$ reduces to an average, making the computation less sensitive to large $n$,
\begin{align*}
     \bm \zeta = T_{B}^{(2)}\big(\big\{x^{(1)}_i\big\}_{i=1}^n,\,x^{(2)}\big) =\frac{1}{n}\sum_{i=1}^n \nabla \phi_{c_i}^{(i)}(x^{(2)}),
\end{align*}
where $\nabla \phi_{c_i}^{(i)}(x^{(2)})$ can be interpreted as the message that observation $Z_i$ conveys about the cluster centers. Another justification for choosing $\kappa$ arises from the structure of the concatenated discrete variable $\tau$, which consists of $n$ latent variables. In this case, the corresponding discrete component of the cost in (MPm) becomes
\begin{align*}
    \kappa \sum_{i=1}^n \big\|x^{(1)}_i -  T_{B,i}^{(1)}\big(x^{(1)}_i,\,x^{(2)}\big) \big\|^2.
\end{align*}
To prevent this discrete cost --- which scales as $\mathcal{O}(n\kappa)$ --- from dominating the continuous part of the cost, which is $\mathcal{O}(1)$, it is appropriate to choose $\kappa = \mathcal{O}(n^{-1})$.
To further simplify computation, we may approximate each additive component  $\phi_{k_i}^{(i)}\big(x^{(2)}\big)$ in~\eqref{eqn:mf-approx} using an additive structure of the form $\sum_{k=1}^K \phi_{k_i,k}^{(i)}\big(x_k^{(2)}\big)$, where $x_k^{(2)}\in\mb R^d$ denotes the $k$-th component of $x^{(2)}$ associated with the $k$-th cluster mean vector $m_k$. This corresponds to a block mean-field approximation for the model parameter $\bm \zeta = (m_{1},m_{2},\ldots,m_{K})^T$ 
over the $K$ (blocks of) parameters $\{m_{k}\}_{k=1}^K$, preserving within-block dependencies while ignoring between-block dependencies.
Similarly, a fully factorized mean-field approximation for the model parameter $\bm \zeta$ can be obtained by further decomposing each convex function $\phi_{k_i,k}^{(i)}\big(x_k^{(2)}\big)$ with respect to the individual coordinates $x_k^{(2)}=\big(x_{k,1}^{(2)}, x_{k,2}^{(2)},\ldots,x_{k,d}^{(2)}\big)$ in $\mb R^d$.

More concretely, according to Lemma \ref{lem:mixture_KL}, we first generate $N$ samples $\{x_1, x_2, \ldots, x_N\}\in \mathbb{R}^{n(K-1)+Kd}$ i.i.d. from the reference distribution $\mu$. Then, the transport map $T$ will map the samples to $\{\xi_1, \xi_2, \ldots, \xi_n\}$. To approximate the expectation in Lemma \ref{lem:mixture_KL}, we need to calculate the log-likelihood of the posterior distribution for each sample $\xi_i:=(\xi^{(1)}_i, \xi^{(2)}_i)$, where $\xi^{(1)}_i:=(\xi^{(1)}_{i,1}, \ldots, \xi^{(1)}_{i, n})\in \mathbb{R}^{n(K-1)}$ corresponds to the latent variables and $\xi^{(2)}_i=(\xi^{(2)}_{i,1}, \ldots,\xi^{(2)}_{i,K})\in\mathbb{R}^{Kd}$ corresponds to the concatenated vector $\bm \zeta$ of cluster means. 

In the GMM model, the log-likelihood of the posterior of $(\tau, \bm \zeta)$ can be expressed as
\begin{align*}
    \log p(\tau, \bm \zeta\mid Z^{(n)}) & = \sum_{k=1}^K \log p( m_k) + \sum_{i=1}^N \log p(c_i) + \sum_{i=1}^N \log p( Z_i \mid  c_i, \bm \zeta) \\
    &\propto \sum_{k=1}^K \left[-\frac{1}{2\sigma^2}(m_k -  m_0)^T (m_k - m_0)\right] - n\log K \\
    &~~~~~~~~~+ \sum_{i=1}^N \left[-\frac{1}{2\lambda^2}( Z_i -  m_{c_i})^T (Z_i - m_{c_i})\right].
\end{align*}

Then, for each transported sample $\xi_i\in\mathbb{R}^{n(K-1)+Kd}$, the log-likelihood is defined as 
\begin{align*}
    \log p(\xi_i \mid Z^{(n)}) &\propto \sum_{k=1}^K \left[-\frac{1}{2\sigma^2}(\xi^{(2)}_{i,k} - m_0)^T (\xi^{(2)}_{i,k} - m_0)\right] - n\log K \\
    &~~~~~~~~~+ \sum_{j=1}^N \sum_{k=1}^K \bm{1}_{\{\xi^{(1)}_{i,j} = b_k\}}\left[-\frac{1}{2\lambda^2}( Z_j -  m_k)^T (Z_j -  m_k)\right],
\end{align*}
where the indicator function can be approximated by the Softmax functions in the computation. Specifically, for each sample $\xi_i$, we have 
\begin{align*}
    \bm{1}_{\{\xi^{(1)}_{i,j} = b_k\} } \approx W_{i,j,k} := \frac{\exp\left(\gamma(\dprod{ x^{(1)}_{i, j}, b_k} + \phi_{k}(x^{(2)}_i)) \right)}{\sum_{\ell\in [K]}\exp\left(\gamma (\dprod{ X^{(1)}_{i, j}, b_\ell} + \phi_{\ell}(x^{(2)}_i))\right)},
\end{align*}
where $\gamma$ is a positive tuning parameter.

The calculation of the middle term in the objective \eqref{eq:mixture_KL_opt} is straightforward, and only requires the knowledge of the target posterior $\pi_n$ up to a normalizing constant. For the term $\det\big(J_{\wt{T}}(\tau, x^{(2)})\big)$ in \eqref{eq:mixture_KL_opt}, we note that $\wt{T}$ does not change the first component $\tau\in \mathcal{Y}_1$ of the input, so that
\begin{align*}
    J_{\wt{T}}(\tau, x^{(2)}) = 
\left(\begin{array}{cc}
I & 0 \\
0 & \nabla^2 \phi_{\tau}(x^{(2)})
\end{array}\right).
\end{align*}
Therefore, the determinant of the Jacobian matrix at $(\tau_i,\,x_i^{(2)})$ can be simplified into 
\begin{align*}
    \det(J_{\wt{T}}(\tau_i, \,x_i^{(2)})) = \det(\nabla^2 \phi_{\tau_i}(x_i^{(2)})).
\end{align*}

\subsection{Computation Details for Simulation Studies in \Cref{sec:simulation}}\label{sec:simulation_config}

\paragraph{\bf Configurations for \Cref{sec:mixture_normal_simu}} We follow the optimization procedure described in \Cref{sec:algorithm}. In particular, we first generate $512$ samples from the target distribution and initialize our transport map by minimizing the Sinkhorn distance between the transported draws through our transport map and those from the target distribution. After that, we further optimize our transport map via  \eqref{eqn:new_obj}. We choose flow length $=64$ for Planar in the experiment. For MAF, we choose 6 layers of 128 units per layer.  We apply the same initialization to both  Planar and MAF for fair comparison. Lastly, we choose $500$ location-scale component maps for TMC.

\paragraph{\bf Configurations for \Cref{sec:logistic}} For the nonlinear function $\varphi$, we use $K=32$ convex units with Softsign functions. Since the target distribution is simple, we skip the initialization step and initialize all parameters from standard Gaussian distributions.   
We set flow length$=32$ for Planar. For MAF, we use 4 layers with 32 hidden units per layer. For NSF, we use 5 blocks of 128 hidden units per block. we use TMC with $300$ location-scale component maps. 

\paragraph{\bf Configurations for \Cref{sec:gmm_simu}}To construct our transport map, we first choose unit vectors in $\mb R^3$ for each value of the latent variables, since the components are not ordered. For the convex functions in formula \eqref{eq:dis_trans_map}, we adopt the ``mean-field'' approximation  in \eqref{eqn:mf-approx}. Since the continuous parameters follow unimodal distribution, we choose $M=1$ and $K_i=32$ for all convex functions. In addition, we choose $\kappa=\frac{1}{n}$ for transport map formula~\eqref{eq:dis_trans_map} as suggested in \Cref{sec:gmm}.  For Planar, we use flow length$=64$. For MAF, we use 4 layers with 32 hidden units per layer. For NSF, we use 5 blocks of 128 hidden units per block.

{\paragraph{Computing time comparison} Here we provide the computing time comparison for different methods in the logistic regression example ($\rho$=5). All methods run on a NVIDIA Tesla V100 GPU with 32GB memory. We report the distribution of computing time across 100 experiments in \Cref{fig:time_comparison}. As we can see, our method has comparable computing time with other baselines.

\begin{figure}[!ht]
    \centering
    \includegraphics[width=0.6\textwidth]{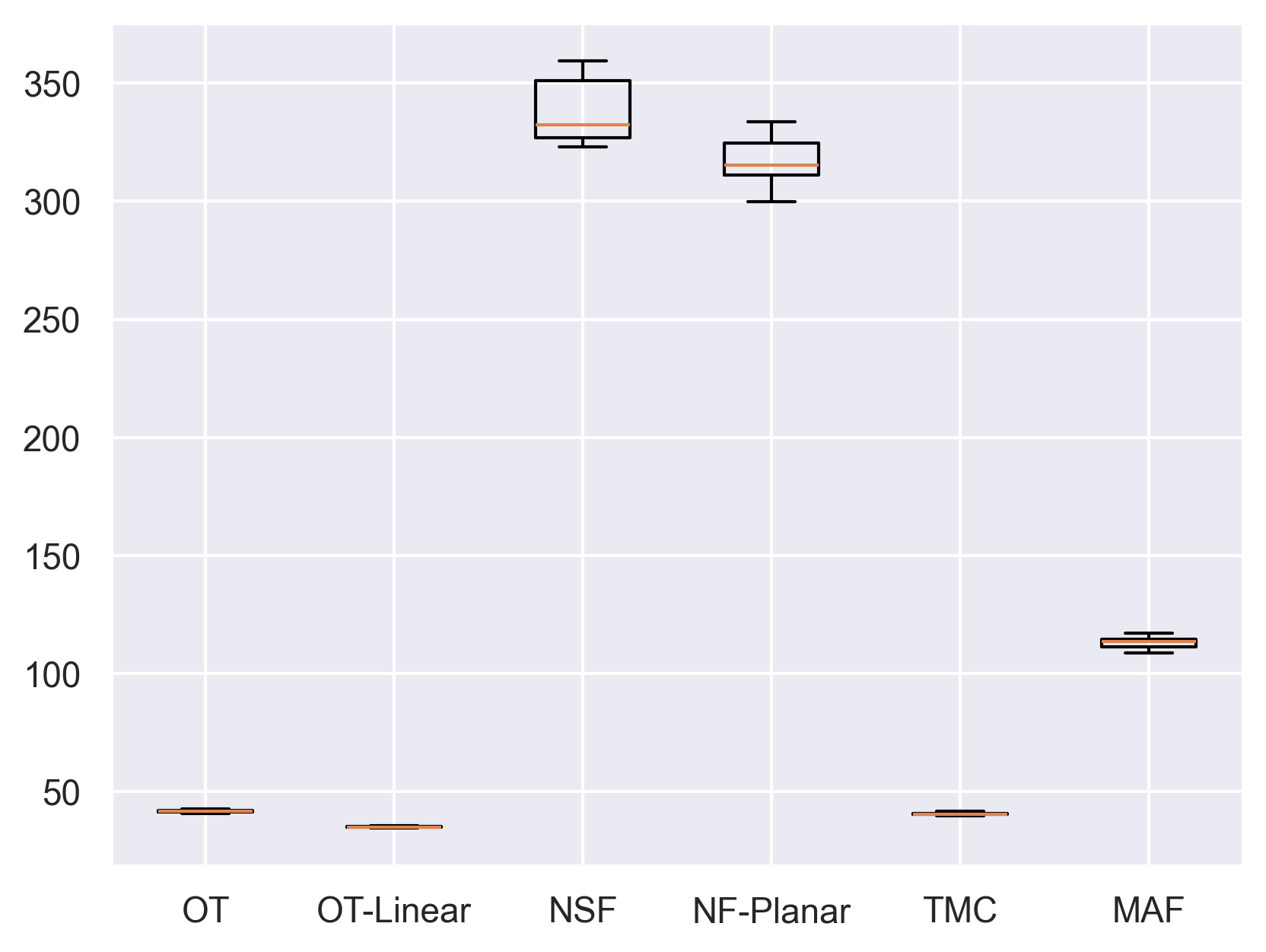}
    \caption{Computing time comparison across 100 experiments in the logistic regression example ($\rho$=5).}
    \label{fig:time_comparison}
\end{figure}}

\section{Algorithm Description}
\label{sec:algorithm}
In this section, we describe the details of the optimization procedure for the transport map. Since the optimization problem \eqref{eqn:new_obj} can be highly non-convex, e.g., when the density of the target distribution is not log-concave, a good initialization of the parameters in the transport map is critical to achieve good optimization performance. Therefore, our optimization of the transport map can be decomposed into two steps: (1) initializing the parameters by minimizing the Sinkhorn divergence; (2) further optimizing the parameters by minimizing the KL-divergence. We employ the stochastic gradient-based optimization methods to optimize the parameters in the transport map.

\paragraph{\bf Initialization.}
We first generate $N$ samples $\{y_1, y_2, \ldots, y_N\}$ i.i.d. from the target distribution. If directly sampling from the target distribution is not possible, we can apply approximation methods such as MCMC or variational inference to obtain approximate draws of the target distribution. Since our goal is only to find good initialization of parameters, it is not required to have the approximation draws with high accuracy. After obtaining $N$ draws from the target distribution, we will minimize the distance between the transported points $\{T(x_1), T(x_2), \ldots, T(x_N)\}$ from the reference distribution through map $T$ and $\{y_1, y_2, \ldots, y_N\}$ to optimize the parameters in the transport map $T$
\[
\widehat T_\text{init} =\arg\min_{T\in \mfT} d(\{T(x_i)\}_{i=1}^n, \{y_i\}_{i=1}^n )
\]
 with  the distance function $d(\cdot, \cdot)$ being  Sinkhorn distance \citep{sinkhorn} or Wasserstein GAN distance \citep{arjovsky2017wasserstein}. In our implementation, we use the Python Optimal Transport \citep{flamary2021pot} framework to compute the Sinkhorn distance and ADAM optimizer \citep{kingma2014adam} for gradient descent.

\paragraph{\bf Optimization.} Since the optimization programs for continuous and mixed variables cases are similar, except for the approximated conditional probability in \eqref{eq:mixture_KL_opt} (see \Cref{sec:mixture_KL_computation}), we focus on the continuous variables case here for conciseness.

 In order to fit into the more general case where the number of the discontinuous regions in the target distribution is not known or the target distribution is multimodal, we construct the approximate family of transport map through the maximum of convex functions as in \eqref{eq:max_of_convex}.
\begin{align}
    \mf T := \Big\{T=\nabla u \Big|& u (x)= \max_{k\in[L]} u_k(x), \, u_k(x)=\sum_{m=1}^Mf\big(x;\,\alpha^{(k)}_{m}, \beta^{(k)}_{m}, w^{(k)}_{m}, v^{(k)}_{m}\big), \label{eq:grad_max_of_convex}\\ 
    &\alpha^{(k)}_{m}\in\R^p, \beta^{(k)}_{m}\in\R^p, w^{(k)}_{m}\in\R, v^{(k)}_{m}\in\R, \ \mbox{for all } k\in[L] \mbox{ and } m\in[M] \Big\}.\nonumber
\end{align}

Starting from the initialized parameter values obtained in the previous step, we apply stochastic gradient-based methods to solve the optimization problem in \eqref{eqn:new_obj}. In each iteration, the estimated objective function and its gradient with respect to the parameters $\bm\gamma := (\alpha_m^{(k)}, \beta_m^{(k)}, w_m^{(k)}, v_m^{(k)})$ are obtained through Monte Carlo method. We first generate $N$ samples $\{x_1, x_2, \ldots, x_N\}$ i.i.d. from the reference distribution $\mu$. The estimated objective function and its gradient are defined as 
\begin{align*}
    \widehat{\mathrm{d}}_{\mathrm{KL}} &:= \frac{1}{N} \sum_{i=1}^N \big[\log (\pi_n \circ T(x_i))  + \log |\det(J_T(x_i))|\big]\\ 
    \widehat\Gamma_t(\bm \gamma) &:= \frac{1}{N} \sum_{i=1}^N \big[\nabla_{\bm \gamma} \log (\pi_n \circ T(x_i))  + \nabla_{\bm\gamma} \log |\det(J_T(x_i))|\big],
\end{align*}
and the parameters are updated through
\begin{align*}
    \bm \gamma_{t+1} \leftarrow \bm \gamma_t + \eta_t \widehat\Gamma_t(\bm \gamma_t),
\end{align*}
where $\eta_t$ is the learning rate and can be updated based on ADAM optimizer~\citep{kingma2014adam}.

 We use the PyTorch auto-differentiation to obtain estimated gradient $\widehat \Gamma_t(\bm \gamma)$, which requires the objective function to be differential with respect to the parameters.  To make $\widehat{\mathrm{d}}_{\mathrm{KL}}$ differentiable, we approximate the max function in \eqref{eq:grad_max_of_convex} with  Softmax function, i.e., 
\begin{align*}
    T(x) & \approx \sum_{k=1}^L \omega_k(x) \nabla u_k(x) \\
    J_T(x) & \approx \sum_{k=1}^L \omega_k(x) \nabla^2 u_k(x)
\end{align*}
with 
\begin{align*}
    \omega_k(x) = \frac{\exp\left(\gamma u_k(x) \right)}{\sum_{l\in [L]}\exp\left(\gamma u_l(x)\right)},
\end{align*}
where $\gamma > 0$ is the concentration parameter. When $\gamma$ goes to infinity, only one $\omega_k$ are close to one while all other $\omega_k$'s will be close to zero. As a result, the Softmax approximation will be close to the maximum function when $\gamma \to \infty$. 

For the stopping criteria, we will stop the optimization when the empirical objective $\widehat{\mathrm{d}}_{\mathrm{KL}}$ does not improve within a number of iterations. Alternatively, we can consider a quantitative method used in \cite{el2012bayesian}, where they use the empirical variance of the objective function \eqref{eqn:new_obj} over reference samples $\{x_1, x_2,\ldots, x_N\}$ to monitor the convergence. When the pushforward distribution $T_{\#}\mu$ is close to $\pi_n$, the empirical variance is close to zero. Therefore, we can stop the optimization when the empirical variance falls below some positive threshold.

\section{Proofs}\label{app:A}
In this section, we collect proofs to the main theoretical results presented in the main paper.

\subsection{Lemma from \cite{panov2015finite}}\label{sec:potential_conjugate}
We summarize the result from \cite{panov2015finite} below which is a key ingredient of our proof. 

\begin{lemma}
\label{lem:grad_inv}
Suppose function $f\in\mb R^p$ is closed, strictly convex and differentiable. Let $f^\dagger\in \mb R^p$ denote the conjugate function of $f$. Then,
\begin{align*}
    (\nabla f^{\dagger})^{-1}(x) = \nabla f(x).
\end{align*}
\end{lemma}
\begin{proof}\!:
Since function $f$ is closed and strictly convex, the conjugate function of $f^\dagger$ is still $f$. Therefore, for any $x\in \mb R^p$, we have
\begin{align*}
    f(x) = \sup_{z\in \mb R^p} \{ \langle x, z\rangle - f^{\dagger}(z) \}.
\end{align*}
Suppose the above optimization problem is maximized at $y\in\mb R^p$, then we have
\begin{align*}
    \nabla_y \left[\langle x, y\rangle - f^\dagger(y)\right] = 0 \quad 
    \Longrightarrow \quad \nabla f^\dagger(y) = x \quad \Longrightarrow \quad y = (\nabla f^\dagger)^{-1}(x),
\end{align*}
and 
\begin{align}
\label{eq:convex_conjugate_equ}
    f(x) = \langle x , y \rangle - f^\dagger(y).
\end{align}
By the definition of the conjugate function $f^\ast$, we have 
\begin{align*}
    f^\dagger(y) = \sup_{z\in\mb R^p} \left\{\langle y, z\rangle - f(z)\right\}.
\end{align*}
Then, according to the equation~\eqref{eq:convex_conjugate_equ}, the above optimization problem is maximized at point $x$. Therefore, we have 
\begin{align*}
    \nabla_x \left[\langle y, x\rangle - f(x)\right] = 0 \quad 
    \Longrightarrow \quad \nabla f(x) = y.
\end{align*}
Combined with the above equation, it can be shown that
\begin{align*}
    y = \nabla f(x) = (\nabla f^\dagger)^{-1}(x).
\end{align*}
\end{proof}

\subsection{Proof of \Cref{prop:equiv}} \label{sec:equiv_pf}
We first consider the case where $\mu$ and $\pi_n$ are continuous, and slightly abuse the same notation for their corresponding density function
\begin{align}
d_\mathrm{KL}(T_\# \mu \| \pi_n)
& \equiv \int \mu(T^{-1}(\theta)) |\det J _{T^{-1}}(\theta)|
\log \frac{\mu(T^{-1}(\theta)) |\det J_{T^{-1}}(\theta)| }{\pi_n(\theta)}
\dd \theta \nonumber \\ 
& = \int \mu(\xi) \log 
\frac{\mu(\xi) |\det J_T(\xi)|^{-1}}{\pi_n(T(\xi))}
\dd \xi \nonumber \\
& = \E_\mu [ - \log \pi_n \circ T - \log \abs{\det (J_T)} + \log \mu]. \label{eq:inverseKlexpansion}
\end{align}
Because the source distribution $\mu$ is fixed throughout the training, the third term $\E_\mu [\log \mu]$ is a constant term. Therefore, the optimization is equivalent to the minimization of the first two terms.

If one of the distribution is not continuous, we replace the underlying Lebesgue measure by the source measure $\mu$, and the ratio of density functions by the corresponding Radon-Nicodym derivative.

\subsection{Proof of Lemma~\ref{lem:splineConvexity}}\label{sec:splineConvexity_pf}
Since $h(x)=\langle \alpha, x\rangle +w$ is a linear function over $x\in \mb R^p$, it is convex. Moreover, as $\varphi$ is increasing and bounded, the function $u \mapsto \int^{u}_0 \varphi(s) \dd s$ is convex over $u\in\R$. As a result, the function $f$ as compositions of convex function and linear function is also convex.

\subsection{Proof of \Cref{thm:bvm_linear}}\label{sec:bvm_linear_pf}

{The proof follows from the fact that the asymptotic Gaussian approximation  distribution of $\mN(\widehat \theta_n, n^{-1} I^{-1}_{\theta_0})$ in \eqref{eq:bvm_kl} can be expressed as the pushforward measure of an affine transformation $\widetilde T(X) = \widehat \theta_n + n^{-1/2} I_{\theta_0}^{-1/2}X$ belonging to $\mfT_{\rm aff}$, where $X\sim \mu = \m N(0, I_p)$ and $I_{\theta_0}^{-1/2}$ denotes the matrix square root of $I_{\theta_0}^{-1}$.

Since $\wh T$ minimizes the KL divergence objective and $\widetilde T$ is a feasible solution, the resulting risk value satisfies
\[
d_\mathrm{KL}\big({\wh T}_\#\mu \,\| \, \pi_n\big) \leq d_\mathrm{KL}\big(\mN(\widehat \theta_n, n^{-1} I^{-1}_{\theta_0})\,\| \, \pi_n\big) = \m O_p(n^{-1/2}).
\]
}

\subsection{Proof of Theorem~\ref{thm:discrete_case}}\label{sec:proof_theorem1}
Our proof has two parts. In the first part, we show that any optimal transport map must take the form of formula~\eqref{eq:dis_trans_map}. In the second part, we show that if any transport map $T$ takes the form of \eqref{eq:dis_trans_map} and pushforwards the reference measure $\mu$ to the target posterior $\pi_n$, i.e. $T_\# \mu = \pi_n$, then it must be an optimal transport map from $\mu$ to $\pi_n$. Finally, these two properties together with the uniqueness and regularity properties in \Cref{lem:regularity_OT} lead to the claimed results. 

\paragraph{\bf Part 1: OT map must take \eqref{eq:dis_trans_map}.}

In this part of the proof, we derive the optimal transport map between the reference distribution $\mu$ and target distribution $\pi_n$ which contains both discrete and continuous components. We define a random variable $Y=(Y^{(1)},\,Y^{(2)}):= (b_\tau,\,\zeta) \in \{b_k\}_{k=1}^K \times \mb R^p$ associated with $\theta=(\tau, \zeta)$. For notation simplicity, we also assume $Y$ to follow $\pi_n$.

The Kantorovich problem between $X\sim \mu$ and $Y\sim \pi_n$ associated with the problem~(MPm) is defined as
\begin{align*}
    \mbox{(KPm)}\qquad \inf_{\gamma \in \Pi(\nu,\, \pi_n)} \int \kappa\norm{x^{(1)} - y^{(1)}}^2 + \norm{x^{(2)} - y^{(2)}}^2 \, \dd \gamma(x, y).
\end{align*}
Then, the dual formulation for (KPm), known as Kantorovich duality, is defined as 
\begin{align}
\label{eq:duality}
    &\sup_{\varphi,\, \psi} \int_{\mb R^r\times \mb R^p} \varphi(x) \,\dd \mu(x) + \int_{\mb R^r\times \mb R^p} \psi(y) \,\dd \pi_n(y),\\
    &\quad  \mbox{s.t.}\quad \varphi \in L^1(\mu),\quad \psi \in L^1(\pi_n) \nonumber \\
    &\qquad\mbox{and}\quad \varphi(x) + \psi(y) \leq \kappa\norm{x^{(1)} - y^{(1)}}^2 + \norm{x^{(2)} - y^{(2)}}^2 \nonumber\\
    &\quad \mbox{for every} \ x\in\mb R^r\times \mb R^p \quad \mbox{and} \quad y \in \{b_k\}_{k=1}^K \times \mb R^p. \nonumber
\end{align}
For each eligible function pair $(\varphi,\,\psi)$ in  \eqref{eq:duality}, we have 
\begin{align*}
    & \varphi(x) + \psi(y) \leq \kappa\norm{x^{(1)} - y^{(1)}}^2 + \norm{x^{(2)} - y^{(2)}}^2 \\
    \iff & \underbrace{\kappa\norm{x^{(1)}}^2 + \norm{x^{(2)}}^2 - \varphi(x)}_{f(x)} + \underbrace{\kappa\norm{y^{(1)}}^2 + \norm{y^{(2)}}^2 - \psi(y)}_{g(y)} \\
    & \geq 2\kappa \dprod{x^{(1)}, y^{(1)}} + 2\dprod{ x^{(2)}, y^{(2)}}.
\end{align*}
Therefore, the dual formulation \eqref{eq:duality} can be equivalently written as 
\begin{align*}
    \mbox{(DPm)} \quad & C_{\mu,\, \pi_n} - \inf_{f, \,g} \left\{\int_{\mb R^r\times \mb R^p} f(x) \,\dd \mu(x) + \int_{\mb R^r\times \mb R^p} g(y) \,\dd \pi_n(Y)\right\}, \\
    &\quad \mbox{s.t.}\quad f \in L^1(\mu),\quad g \in L^1(\pi_n)\\
    &\qquad \mbox{and}\quad f(x) + g(y) \geq 2\kappa \dprod{ x^{(1)}, y^{(1)}} + 2\dprod{x^{(2)}, y^{(2)}} \\
    &\quad \mbox{for every} \ x\in\mb R^r\times \mb R^p \quad \mbox{and} \quad y \in \{b_k\}_{k=1}^K \times \mb R^p,
\end{align*}
where 
\begin{align*}
    C_{\mu,\,\pi_n} := \int_{\mb R^r\times \mb R^p} \kappa\norm{x^{(1)}}^2 + \norm{x^{(2)}}^2\,\dd \mu(x) + \int_{\mb R^r\times \mb R^p} \kappa\norm{y^{(1)}}^2 + \norm{y^{(2)}}^2\,\dd \pi_n(y).
\end{align*}
Therefore, for any eligible pair of functions $(f,\,g)$ in problem~(DPm), we have
\begin{align}
\label{eq:condition_on_func_detailed}
    f(x) & \geq \sup_{z\in \{b_k\}_{k=1}^K \times \mb R^p}\left\{2\kappa \langle x^{(1)}, z^{(1)}\rangle + 2\langle x^{(2)}, z^{(2)}\rangle - g(z) \right\} \nonumber\\
    & = 2\sup_{k\in[K]}\left\{\kappa \langle x^{(1)}, b_k\rangle + \sup_{z^{(2)}\in\mb R^p} \left\{ \langle x^{(2)}, z^{(2)}\rangle - \frac{1}{2}g(z^{(1)}, z^{(2)}) \right\}\right\}.
\end{align}
We define $h(y):=\frac{1}{2}g(y)$ over $\{b_k\}_{k=1}^K \times \mb R^p$ and $h_k(y^{(2)}) := h(b_k, y^{(2)})$ for $k\in[K]$. Then, for any function pair $(f, h)$ eligible for problem (DPm), we have
\begin{align*}
    f(x) + 2h(y) \geq 2\kappa \dprod{x^{(1)}, y^{(1)}} + 2\dprod{ x^{(2)}, y^{(2)}}.
\end{align*}
In the following, we show that it suffices to only consider convex functions for $h_k$ to solve the problem (DPm). For any function $h_k$ in $\mb R^p$, its convex conjugate function is defined as
\begin{align*}
    h_k^{\ast}(x^{(2)}) := \sup_{Z^{(2)}\in\mb R^p} \left\{\dprod{ x^{(2)}, Z^{(2)}} - h_k(Z^{(2)})\right\}.
\end{align*}
Then, we define a function $\tilde{f}$ over $\mb R^r$ as
\begin{align}
\label{eq:f_conjugate}
    \tilde{f}(x) &:= 2\sup_{k\in[K]}~\left\{\kappa\dprod{ x^{(1)}, b_k} + h_k^{\ast}(x^{(2)})\right\} \\
    &= 2 \sup_{k\in[K]}~\left\{\kappa \dprod{ x^{(1)}, b_k } + \sup_{z^{(2)}\in \mb R^p}\left\{ \dprod{ x^{(2)},Z^{(2)}} - h_{k}(Z^{(2)}) \right\}\right\}. \nonumber
\end{align}
For any function pair $(f,\,h)$ eligible for problem (DPm), the pair $(\tilde{f},\,h)$ improves the value of the objective function in (DPm), since $f(x) \geq \tilde{f}(x)$ for all $x\in \mb R^r$ according to the definition of $\tilde{f}$ and the inequality in \eqref{eq:condition_on_func_detailed}. Then, we can define a function $h_k^{\dagger}$ over $\mb R^p$ as the convex conjugate function of $h_k^{\ast}$ and function $h^{\dagger}(b_k, y^{(2)}) :=h_k^{\dagger}(y^{(2)})$ over $\{b_k\}_{k=1}^K \times \mb R^p$. The function pair $(\tilde{f},\, h^{\dagger})$ will further improve $(\tilde{f}, h)$ on the value of the objective function in (DPm), since $h(y) \geq h^{\dagger}(y)$ for all $y=(b_k, y^{(2)})\in \{b_k\}_{k=1}^K \times \mb R^p$ by the definition of the conjugate functions, i.e., 
\begin{align*}
    h^{\dagger}(y) &= h_k^{\dagger}(y^{(2)}) \\
    &= \sup_{x^{(2)}\in \mb R^p} \left\{\dprod{ x^{(2)}, y^{(2)}} - h_{k}^{\ast}(x^{(2)})\right\} \\
    &= \sup_{x^{(2)}\in \mb R^p} \left\{\dprod{ x^{(2)}, y^{(2)}} - \sup_{Z^{(2)}\in \mb R^p} \left\{ \dprod{ x^{(2)}, Z^{(2)} }- h_{k}(Z^{(2)}) \right\}\right\} \\
    &\leq \sup_{x^{(2)}\in \mb R^p} \left\{\dprod{ x^{(2)}, y^{(2)}} - \dprod{ x^{(2)}, y^{(2)}} + h_{k}(y^{(2)})\right\} \\
    &= h_{k}(y^{(2)}) = h(y).
\end{align*}
Based on the above discussion, the problem (DPm) can be formulated as a problem with $K$ convex functions $h_k(y^{(2)})$ for $k\in[K]$, i.e.,
\begin{align*}
    \mbox{(DPm')} \quad & C_{\mu,\, \pi_n} - 2 \cdot \inf_{f,\,h} \left\{\int_{\mb R^r\times \mb R^p} f(x) \,\dd \mu(x) + \int_{\{b_k\}_{k=1}^K \times \mb R^p} h(y) \,\dd \pi_n(y)\right\}, \\
    &\quad  \mbox{s.t.} \quad h \in L^1(\pi_n) \quad \mbox{and} \quad h_k(y^{(2)}):=h(b_k,y^{(2)}) \ \mbox{is convex over} \ \mb R^p \ \mbox{for} \ k\in[K] \\
    &\qquad \mbox{and} \quad f(x) := \max_{k\in[K]}~\left\{\kappa\langle x^{(1)}, b_k\rangle + h_k^{\ast}(x^{(2)})\right\}.
\end{align*}
According to the duality equality between (KPm) and (DPm) and the equivalence between (DPm) and (DPm'), it is satisfied that if point $(x, y)$ belongs to the support of the optimal transport plan $\gamma$ between $\mu$ and $\pi_n$ minimizing (KPm) and $(f,\, h)$ is a pair of functions maximizing (DPm'), then
\begin{align*}
    f(x) + h(y) = \kappa \dprod{ x^{(1)}, y^{(1)}} + \dprod{ x^{(2)}, y^{(2)}}.
\end{align*}
This implies that for $(x,y)\in \gamma$ with $y = (b_\tau, \zeta)$, it  satisfies 
\begin{align}
\label{eq:optimality_equal}
    & \kappa \dprod{ x^{(1)}, b_\tau} + \dprod{ x^{(2)}, \zeta} - h_\tau(\zeta) \nonumber \\
    =& \max_{k\in[K]}~\left\{\kappa\dprod{ x^{(1)}, b_k}+ h_k^{\ast}(x^{(2)})\right\} \nonumber \\
    =& \max_{k\in[K]}~\left\{\kappa \dprod{ x^{(1)}, b_k} + \sup_{Z^{(2)}\in \mb R^p}\left\{ \dprod{ x^{(2)}, Z^{(2)}} - h_{k}(Z^{(2)}) \right\}\right\} \nonumber \\
    \geq& \left\{\kappa \dprod{ x^{(1)}, b_\tau} + \dprod{ x^{(2)}, \zeta} - h_\tau(\zeta)\right\}
\end{align}
where the first equation is by the definition of $f$.
Therefore, the equality is achieved in the last step. According to equation~\eqref{eq:optimality_equal}, the discrete random variable $\tau\in[K]$ associated with $b_\tau$ is then defined as
\begin{align*}
    \tau = \argmax_{k\in[K]} \left\{ \kappa \dprod{ x^{(1)}, b_k} + h_k^{\ast}(x^{(2)})\right\}.
\end{align*}
Then, by plugging the definition of $\tau$ into \eqref{eq:optimality_equal}, we have 
\begin{align*}
    h^{\ast}_\tau(x^{(2)})=\dprod{ x^{(2)}, \zeta} - h_\tau(\zeta).
\end{align*}
According to the definition of the conjugate function $h_\tau^{\ast}$ of $h_\tau$, the following equation is satisfied.
\begin{align*}
    \nabla_{\zeta} \left[\dprod{ x^{(2)}, \zeta} - h_\tau(\zeta)\right] = 0 \quad \Longrightarrow \quad  x^{(2)} - \nabla h_\tau(\zeta) = 0.
\end{align*}
Therefore, we have
\begin{align*}
    \zeta = (\nabla h_\tau)^{-1}(x^{(2)}) = \nabla h_\tau^{\ast}(x^{(2)}),
\end{align*}
where the second equality is according to Lemma \ref{lem:grad_inv}.

Based on the above relationship between the optimal transport plan $\gamma$ solving (KPm) and the function pair $(f, h)$ solving (DPm'), the optimal transport map $T^*$ that pushforwards $\mu$ to $\pi_n$ can be characterized through $K$ convex functions $\{\phi_k\}_{k=1}^K$ over $\mb R^p$ with $\phi_k(x^{(2)}):= \frac{1}{\kappa} h_k^\ast(x^{(2)})$ for $k\in[K]$. Specifically, $T^*$ maps $(x^{(1)},\,x^{(2)})$ to $(\tau,\,\zeta)$ with
\begin{align}
\label{eq:trans_map}
\left\{\begin{array}{ll}
\tau &= T^{*(1)}\big(x^{(1)},\,x^{(2)}\big) = \argmax_{k\in[K]} \left\{ \dprod{x^{(1)}, \, b_k} + \phi_k(x^{(2)})\right\}, \\
\zeta &= T^{*(2)}\big(x^{(1)},\,x^{(2)}\big) = \kappa \nabla \phi_\tau(x^{(2)}).
\end{array}\right.
\end{align}
In addition, the convex function $u^*$ associated with the optimal transport map $T^*$ is in form of function $f$, i.e.,
\begin{align*}
    u^*\big(x^{(1)},x^{(2)}\big) = \kappa \max_{k\in [K]} \big\{\dprod{ x^{(1)},\, b_k} + \phi_k(x^{(2)}) \big\}.
\end{align*}

\paragraph{\bf Part 2: Map with Form \eqref{eq:dis_trans_map} is Optimal.}

In this part of the proof, we show the following property: suppose there exists a transport map $T_0$ in form of \eqref{eq:trans_map} satisfying $(T_0){\#}(\mu)=\pi_n$, the transport map $T_0$ is then the optimal transport map for the problem (MPm). In addition, the induced transport plan $(id, T_0)_{\#}(\mu)$ is the optimal plan for the Kantorovich problem (KPm) between $\mu$ and $\pi_n=(T_0)_{\#}(\mu)$.

To see this, suppose there exist convex functions $\{\phi_k\}_{k=1}^K$ such that $(T_0)_{\#}(\mu)=\pi_n$. We can define a pair of $(f, h)$ that is eligible for problem (DPm'), i.e.,
\begin{align*}
    f(x^{(1)}, x^{(2)}) :=& \kappa \max_{k\in[K]}~\left\{\dprod{ x^{(1)}, b_k}+ \phi_k(x^{(2)})\right\} \\
    h(b_\tau, \zeta) :=& (\kappa\phi_\tau)^{\ast}(\zeta)\quad \mbox{for} \quad  \tau\in[K],
\end{align*}
where $(\kappa\phi_k)^{\ast}$ is the convex conjugate function of $\kappa \phi_k$ for $k\in[K]$.

For any point $(x, y)$ with $y=(b_\tau, \zeta)$ that belongs to the support $\gamma_0=(id, T_0)_{\#}(\mu)$, we have 
\begin{align*}
    f(x) + h(y) &= \max_{k\in[K]} \kappa \left\{\dprod{x^{(1)}, b_k} + \phi_k(x^{(2)})\right\} + (\kappa\phi_\tau)^{\ast}(\zeta) \\
    &= \kappa \dprod{x^{(1)}, b_\tau} + \kappa \phi_\tau(x^{(2)}) + (\kappa\phi_\tau)^{\ast}(\zeta) \\
    &= \kappa \dprod{ x^{(1)}, b_\tau} + \dprod{x^{(2)}, \zeta},
\end{align*}
where the second equation is by the definition of $\tau$ in \eqref{eq:trans_map} and the third equation is by the property of the convex function $\kappa\phi_\tau$.

Therefore, we have the following inequality.
\begin{align*}
    \min(\mbox{KPm}) &\leq \int \kappa\norm{x^{(1)} - y^{(1)}}^2 + \norm{x^{(2)} - y^{(2)}}^2 \, \dd \gamma_0(x, y) \\
    &= C_{\mu, \pi_n} - 2\int \kappa \dprod{ x^{(1)}, y^{(1)}} + \dprod{ x^{(2)}, y^{(2)}} \, \dd\gamma_0(x, y) \\
    &= C_{\mu, \pi_n} - 2\int f(x) + h(y)\,\dd\gamma_0(x, y) \\
    &= C_{\mu, \pi_n} - 2\left[\int f(x) \,\dd\mu(x) + \int h(x) \,\dd\pi_n(y) \right]  \\
    &\leq \max(\mbox{DPm'}).
\end{align*}
In addition, since $\min(\mbox{KPm}) \geq \max(\mbox{DPm'})$ by the inequality between the primal and dual problems, we have $\min(\mbox{KPm}) = \max(\mbox{DPm'})$. Therefore $T_0$ is the optimal transport map that pushforwards $\mu$ to $\pi_n$ and $(id, T_0)_{\#}(\mu)$ is the optimal plan between $\mu$ and $\pi_n=(T_0)_{\#}(\mu)$.

\subsection{Proof of \Cref{lem:mixture_KL}}\label{sec:mixture_KL_pf}

Let us denote the intermediate joint density of $(\tau, X^{(2)})$ as $\wt \nu(\tau, X^{(2)}) := P(\tau\mid X^{(2)})\mu_2(X_2)$. Then we can rewrite the KL divergence in \eqref{eqn:KL_Form} as
\begin{align*}
d_\text{KL} (T_\# \mu || \pi_n) &= \int \wt \nu(\wt T^{-1}(\tau, \zeta))\abs{\det J_{\wt T^{-1}}(\tau, \zeta)}\log \frac{\wt \nu(\wt T^{-1}(\tau, \zeta))\abs{\det J_{\wt T^{-1}}(\tau, \zeta)}}{\pi_n(\tau, \zeta)} \dd \tau \dd \zeta \\
& = \int \wt \nu(\tau, X^{(2)}) \log \frac{\wt \nu(\tau, X^{(2)})\abs{\det J_{\wt T}(\tau, X^{(2)})}^{-1}}{\pi_n (\wt T(\tau, X^{(2)}))} \dd \tau \dd X^{(2)} \\
& = \int \wt \nu(\tau, X^{(2)}) \log \frac{P(\tau\mid X^{(2)})\mu_2(X^{(2)}) \abs{\det J_{\wt T}(\tau, X^{(2)})}^{-1}}{\pi_n (\wt T(\tau, X^{(2)}))} \dd \tau \dd X^{(2)} \\
&=   \int \sum_{k=1}^K P(k\mid X^{(2)})\mu_2(X^{(2)})\log \frac{P(k\mid X^{(2)})\mu_2(X^{(2)})\abs{\det J_{\wt T}(k, X^{(2)})}^{-1}}{\pi_n (\wt T(k, X^{(2)}))}  \dd X^{(2)}.
\end{align*}
Since we know given $X^{(1)}, X^{(2)}$, $\tau$ is deterministic (we abuse the notation as $\tau(X^{(1)}, X^{(2)})$)
\[
P(k \mid X^{(2)}) = \int P(k \mid X^{(1)}, X^{(2)}) \mu(X^{(1)}\mid X^{(2)}) \dd X^{(1)} =  \int 1_{\tau(X^{(1)}, X^{(2)}) =k }\mu(X^{(1)}\mid X^{(2)})\dd X^{(1)},
\]
then we can proceed the KL divergence derivation as
\begin{align*}
d_\text{KL} (T_\# \mu || \pi_n) & =   \int \sum_{k=1}^K 1_{\tau(X^{(1)}, X^{(2)}) =k }\mu(X)\log \frac{P(k\mid X^{(2)})\mu_2(X^{(2)})\abs{\det J_{\wt T}(k, X^{(2)})}^{-1}}{\pi_n (\wt T(k, X^{(2)}))}  \dd X^{(1)} \dd X^{(2)} \\
& =\int \mu(X) \log \frac{P(\tau\mid X^{(2)})\mu_2(X^{(2)})\abs{\det J_{\wt T}(\tau, X^{(2)})}^{-1}}{\pi_n (\wt T(\tau, X^{(2)}))}  \dd X^{(1)} \dd X^{(2)}\\
& =\int \mu(X) \log \frac{P(\tau\mid X^{(2)})\abs{\det J_{\wt T}(\tau, X^{(2)})}^{-1}}{\pi_n (\wt T(\tau, X^{(2)}))}  \dd X^{(1)} \dd X^{(2)} +\text{const}.
\end{align*}

{
\section{Sensitivity Analysis}\label{app:sen}

\subsection{Choice of the embeddings}\label{sec:embedding_choice} In \Cref{sec:mixed_case}, we propose to embed the discrete parameter $\tau$ into an embedding space $\{b_k\}_{k=1}^K \subset \R^r$ to encode the categorical information. Here we provide some discussions on the choice of the embedding space.

As we mention in \Cref{sec:mixed_case}, we choose different embedding strategies based on the nature of the discrete variable $\tau$. When $\tau$ is a nominal variable, we embed the categories as the vertices of a regular polyhedron. By a regular polyhedron, we mean a convex polytope whose vertices form a finite set of points in Euclidean space with identical pairwise Euclidean distances. This construction ensures that all categories are treated symmetrically. One convenient realization is the one-hot encoding in $\R^K$, namely $b_k = e_k$ for $k \in [K]$, where $e_k = (0,\ldots,0,1,0,\ldots,0)^\top$ has a single nonzero entry in the $k$th coordinate. Under this embedding, the pairwise distance between any two distinct embeddings is $\norm{e_k - e_l}_2 = \sqrt{2}$.
Similarly, when $\tau$ is an ordinal variable, we embed the categories as an increasing sequence on the real line. A simple choice is to set $b_k = k$ for $k \in [K]$, which naturally reflects the ordering and quantifies the gap between different levels of the ordinal variable.

The embedding choice is not unique and we propose the above strategies for their simplicity and interpretability. To provide some reassurance on the embedding choice, here we conduct a sensitivity analysis on the scale of the embedding vectors. We focus on the categorical variable case and consider the mixture of Gaussians example in \Cref{sec:mixture_normal_simu} with $d=5$ and $K=10$. We set the embeddings as $b_k = b\cdot e_k$ for $k\in[K]$ where $b$ can be any constant $b>0$. Under this notation, our transport cost in (MPm) can be written as
\[
\E_{\mu} \Big[\kappa \underbrace{\norm{x^{(1)}-b e_{T^{(1)}}(x) }^2}_\text{discrete part cost}+\underbrace{\norm{x^{(2)}-T^{(2)}(x)}^2}_\text{continuous part cost}\Big].
\]

Certainly, $\kappa$ and $b$ are playing different roles here. While $\kappa$ is controlling the relative weight on the discrete part cost, $b$ controls the distance between different categories and thus how sensitive the loss to different categories. When $b$ is small, categories are closer and the gradient $\nabla\phi_\tau$ are small and less discriminative. When $b$ is large, categories are further away and the gradient $\nabla\phi_\tau$ are big. However, it is hard to disentangle its impact on the contribution of the discrete part cost and the magnitude of the gradient $\nabla\phi_\tau$ from the mathematical representations of the solutions.

Empirically, we do not observe significant difference in the posterior distribution when we change $b$. For example, we consider the Gaussian mixture model with $\Delta=2$ case in Section 5.3 and compute the OT posterior with varying $b$ with $\log(b)$ from -10 to 4, as shown in \Cref{fig:GMM_different_b}. We find the posteriors are very similar except only in the case when $b=e^4=54.60$ is really large. Thus, our OT map is not very sensitive to the scaling of the discrete embedding.
\begin{figure}[!ht]
\centering
\includegraphics[width=0.8\textwidth]{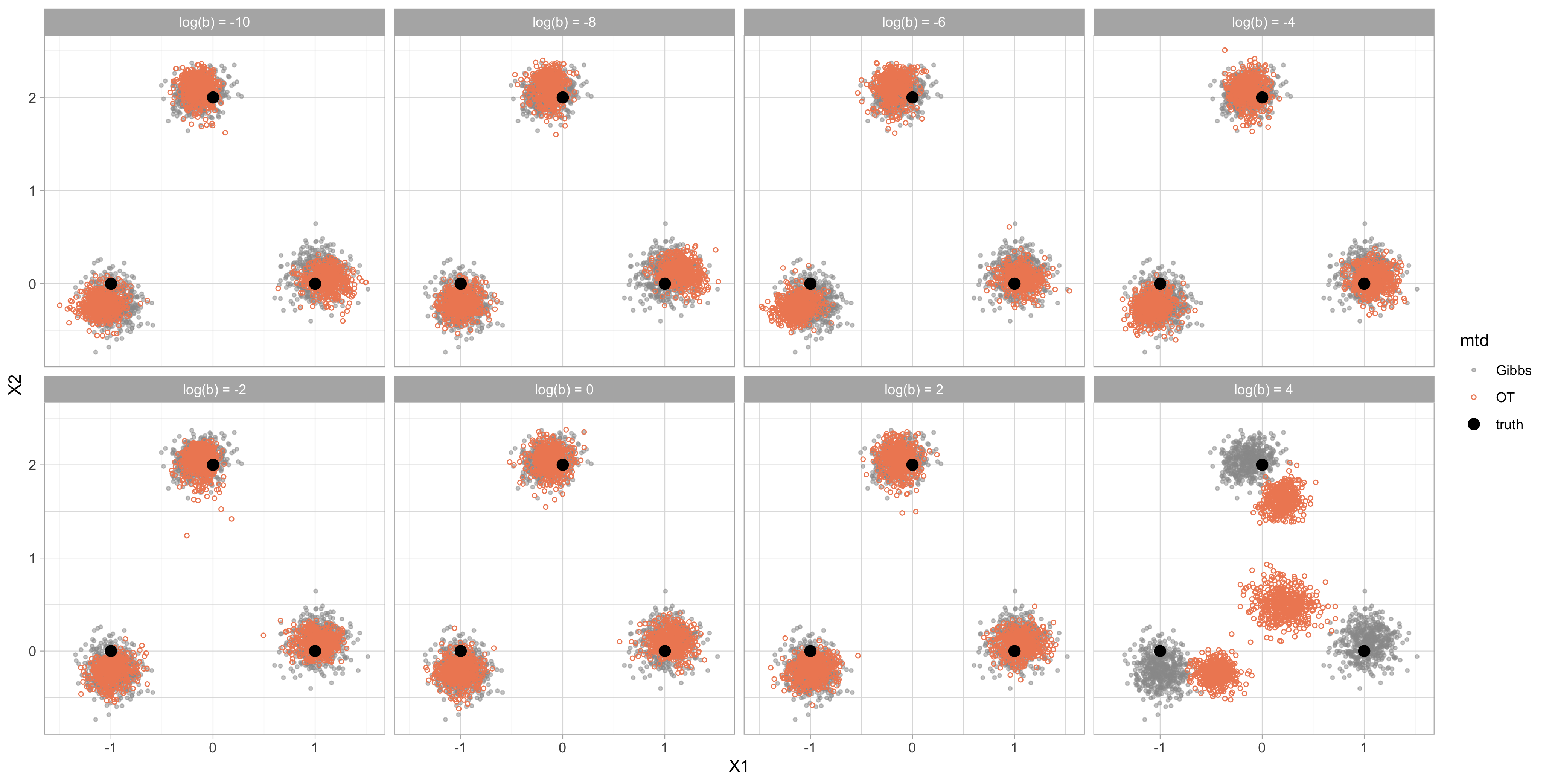}\label{fig:GMM_different_b}
\caption{OT posterior under the GMM model when $b$ is different.}
\end{figure}

\subsection{Choice of the number of local potential functions $L$}\label{sec:choice_L}
When the target distribution is highly multi-modal and the modes are well-separated, \Cref{lem:regularity_OT_2} suggests that the number of local potential functions $L$ should at least match the number of modes $K$ to approximate the optimal transport map well. In practice, the expressiveness of the transport map class is jointly determined by the number of local potential functions $L$ and the complexity of each local potential function (e.g., the number of hidden units). Therefore, one can choose a smaller $L$ if each local potential function is sufficiently complex. When $L>K$, it adds more flexibility to the transport map class, which may help improve the approximation accuracy. In our previous simulation results in \Cref{sec:mixture_normal_simu}, we show that when $L>K$, the OT map can still approximate the target distribution well without overfitting issues. Additionally, the number of well-separated modes $K$ can be smaller than $K^*$, the true number of mixture components, and choosing $L<K^*$ may still yield satisfactory performance. For example, in the Banana distribution in \Cref{sec:additional_quantile}, we set $L=1$ while the true distribution is a mixture of three Gaussians.

To provide a concrete example, we consider the mixture of Gaussians example in \Cref{sec:mixture_normal_simu} with $d=5$ and $K=10$. We examine a list of values for $L$, 
ranging from underestimating to overestimating the number of components. We provide a visualization of the estimated posteriors when $L=1,3,5,7,10,15$ in \Cref{fig:posterior_vs_L}. We observe that even though the true number of components is $K=10$, the OT map with $L=5$ is already able to capture the main structure of the posterior distribution. Increasing $L$ further improves the approximation quality, but the improvement becomes marginal after $L=10$.

\begin{figure}[!ht]
    \centering
    \subfloat[$L=1$, nhid$=16$]{
    \includegraphics[width=0.22\textwidth]{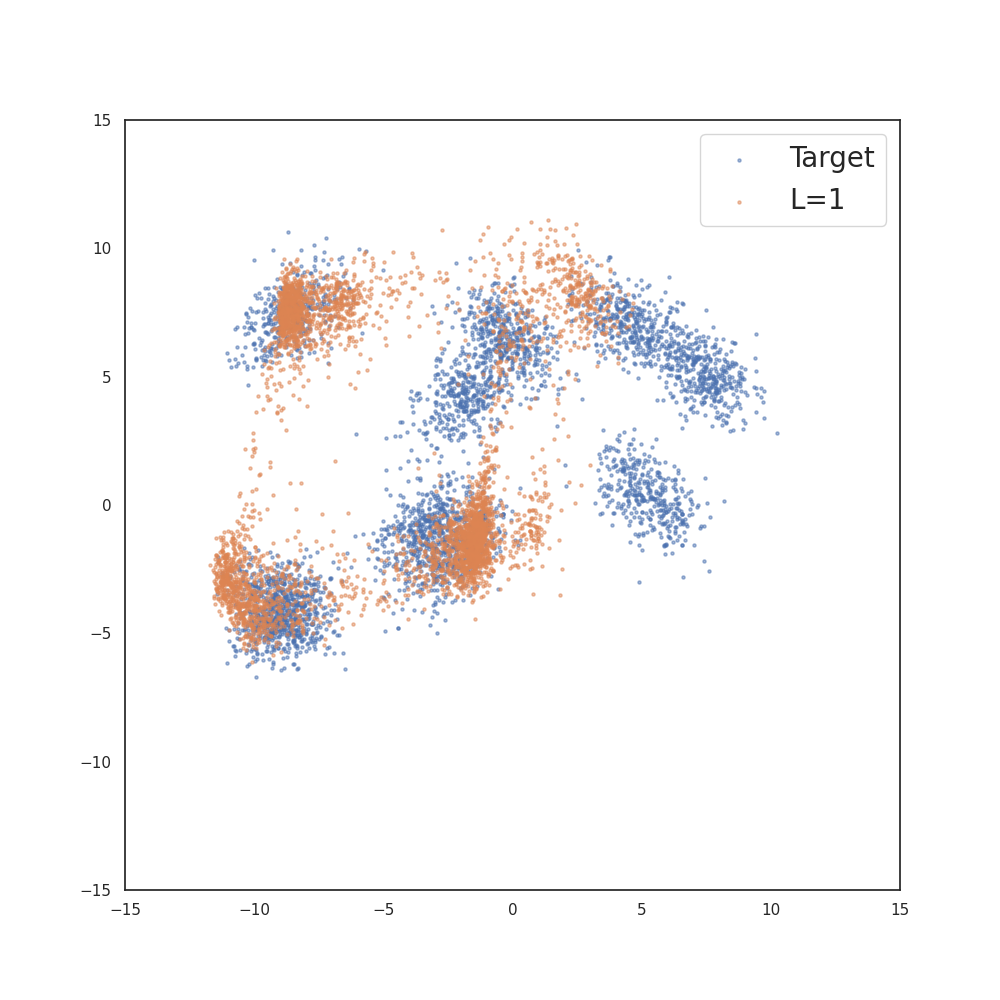}
    }
    \subfloat[$L=3$, nhid$=16$]{
    \includegraphics[width=0.22\textwidth]{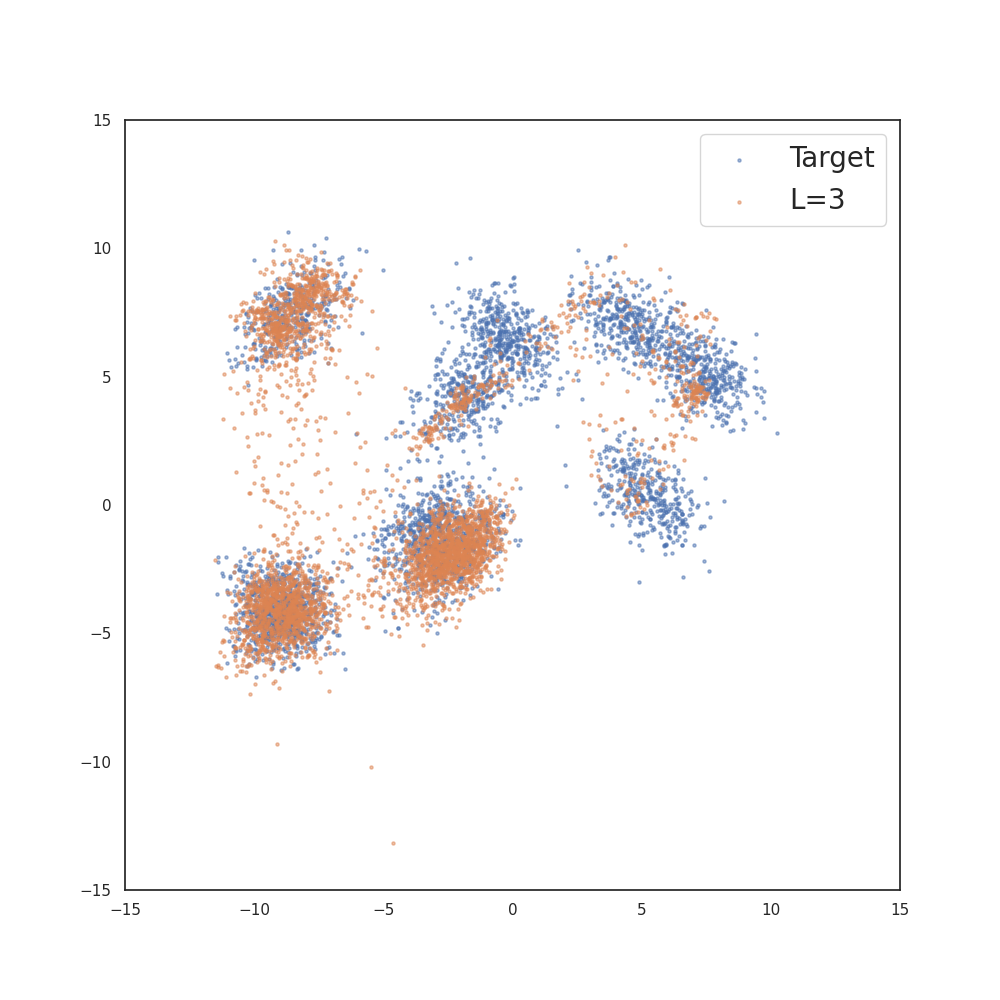}
    }
    \subfloat[$L=5$, nhid$=16$]{
    \includegraphics[width=0.22\textwidth]{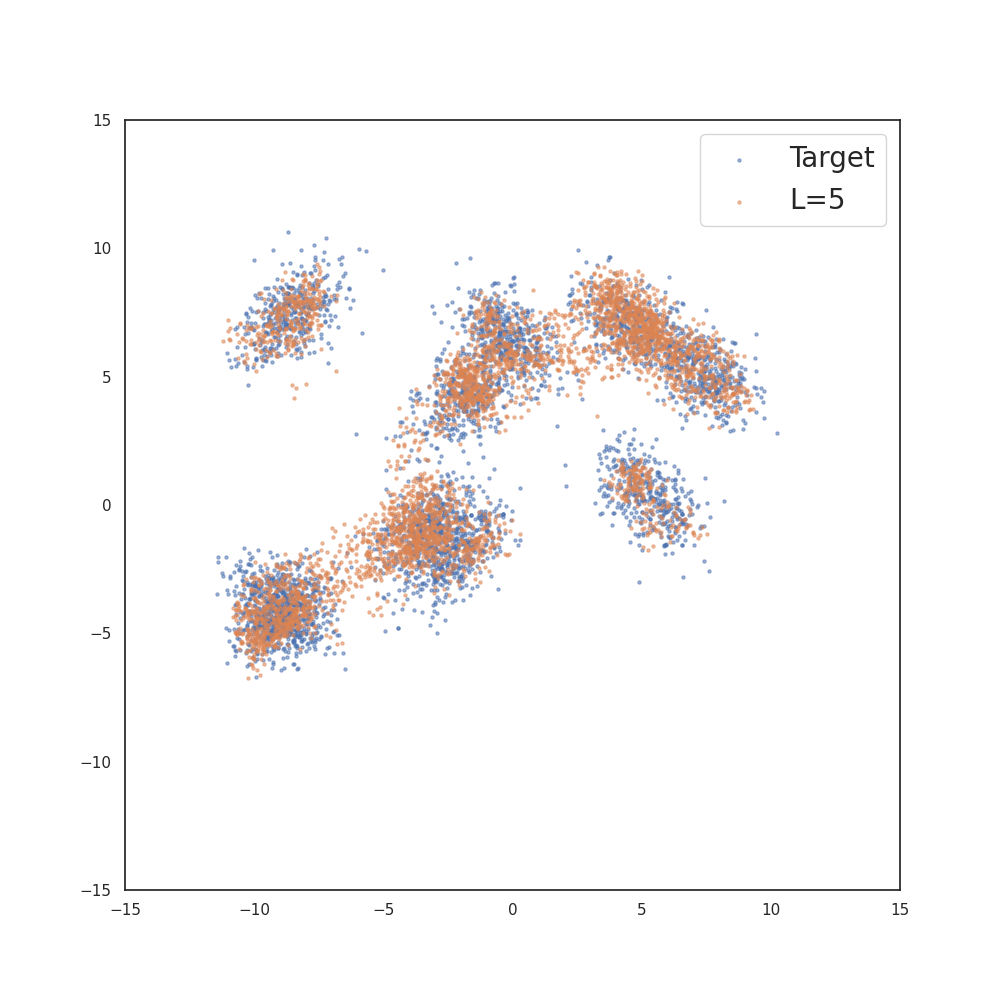}
    }
    \subfloat[$L=7$, nhid$=16$]{
    \includegraphics[width=0.22\textwidth]{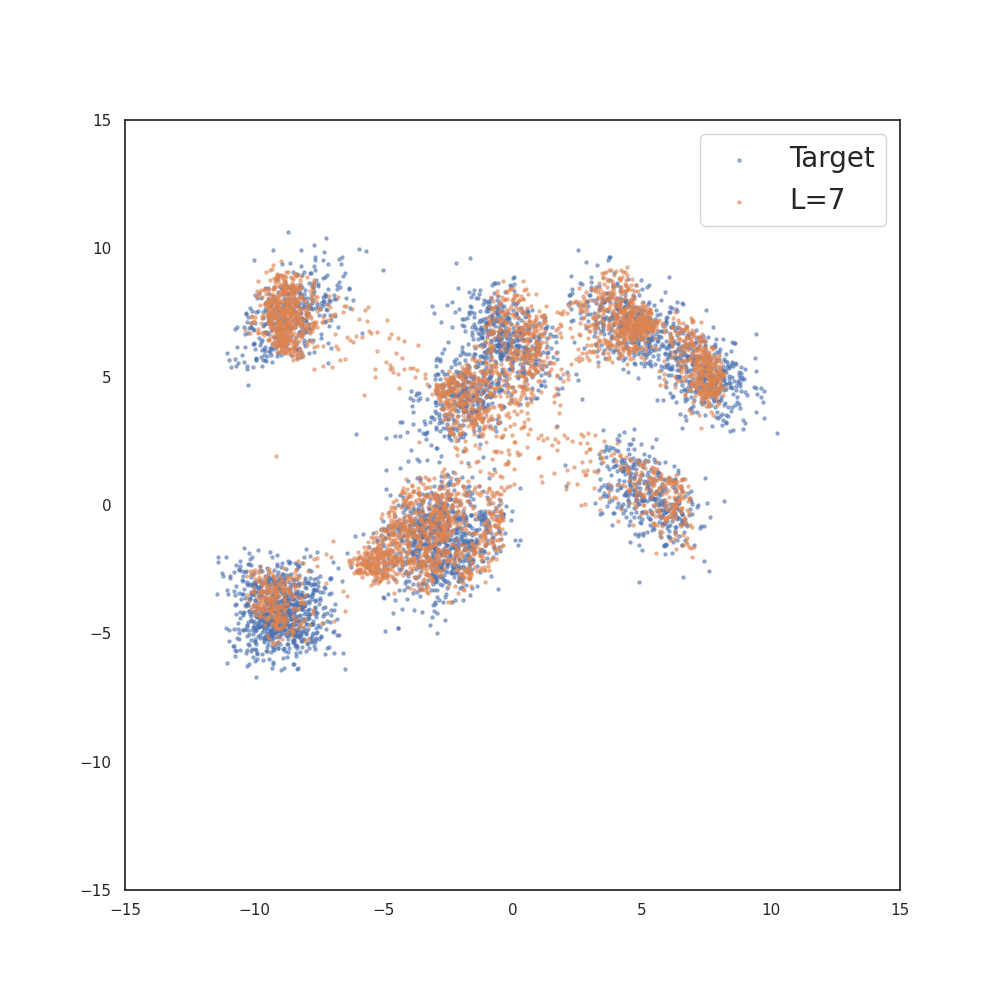}
    }
    \\
    \subfloat[$L=10$, nhid$=16$]{
    \includegraphics[width=0.22\textwidth]{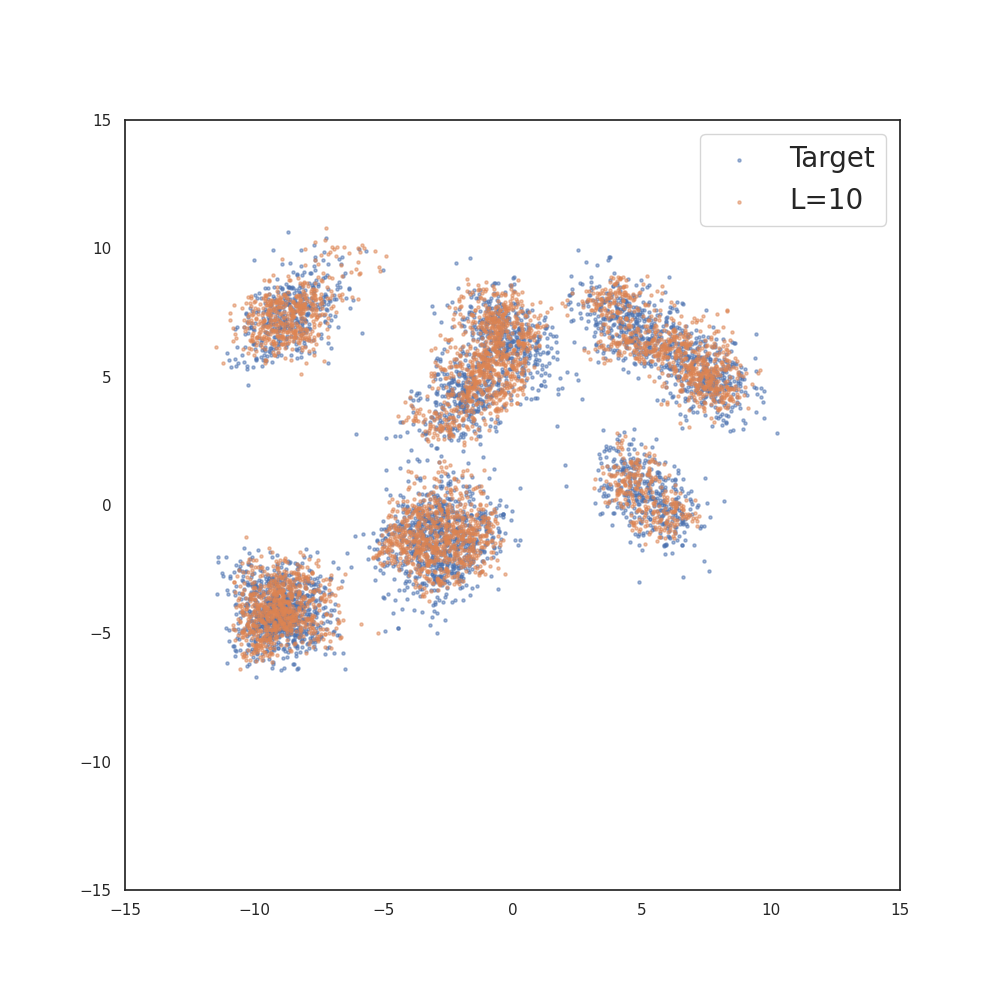}
    }
    \subfloat[$L=15$, nhid$=16$]{
    \includegraphics[width=0.22\textwidth]{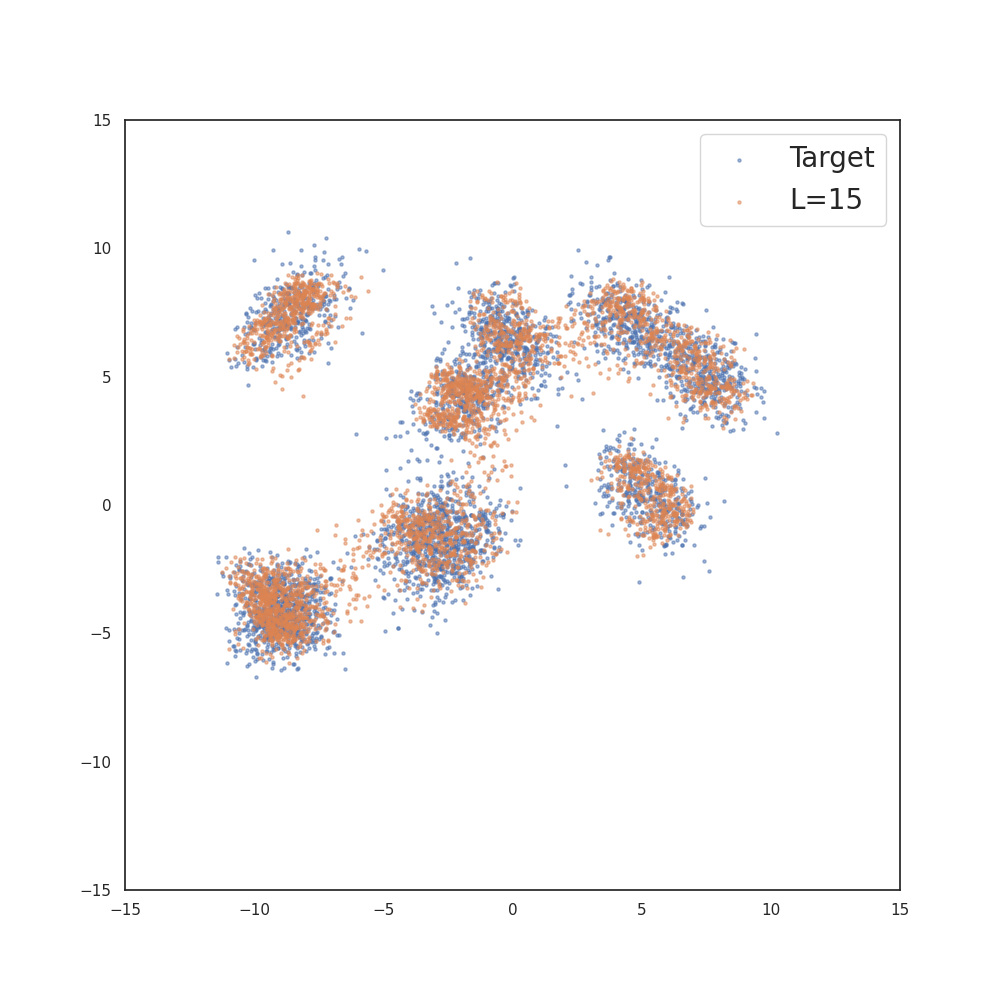}
    }
    \subfloat[$L=3$, nhid$=32$]{
        \includegraphics[width=0.22\textwidth]{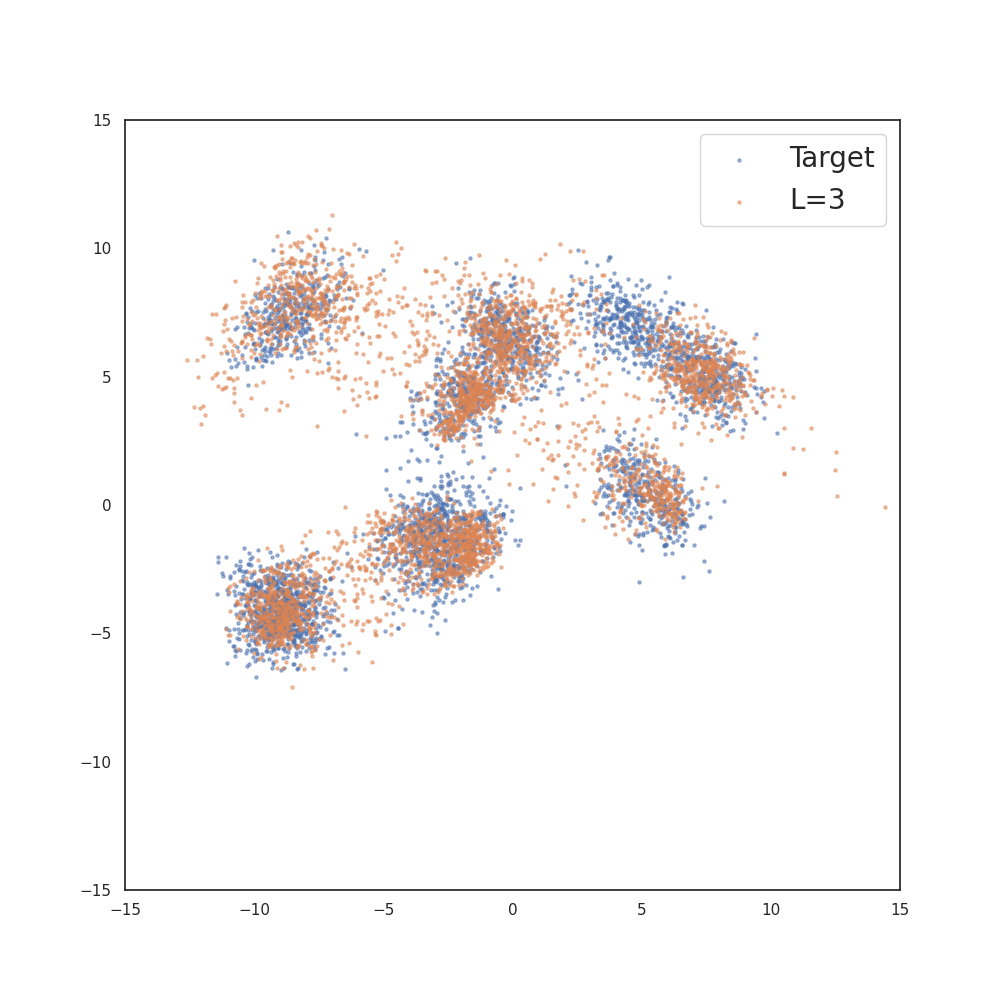}
    }
    \subfloat[$L=5$, nhid$=32$]{
        \includegraphics[width=0.22\textwidth]{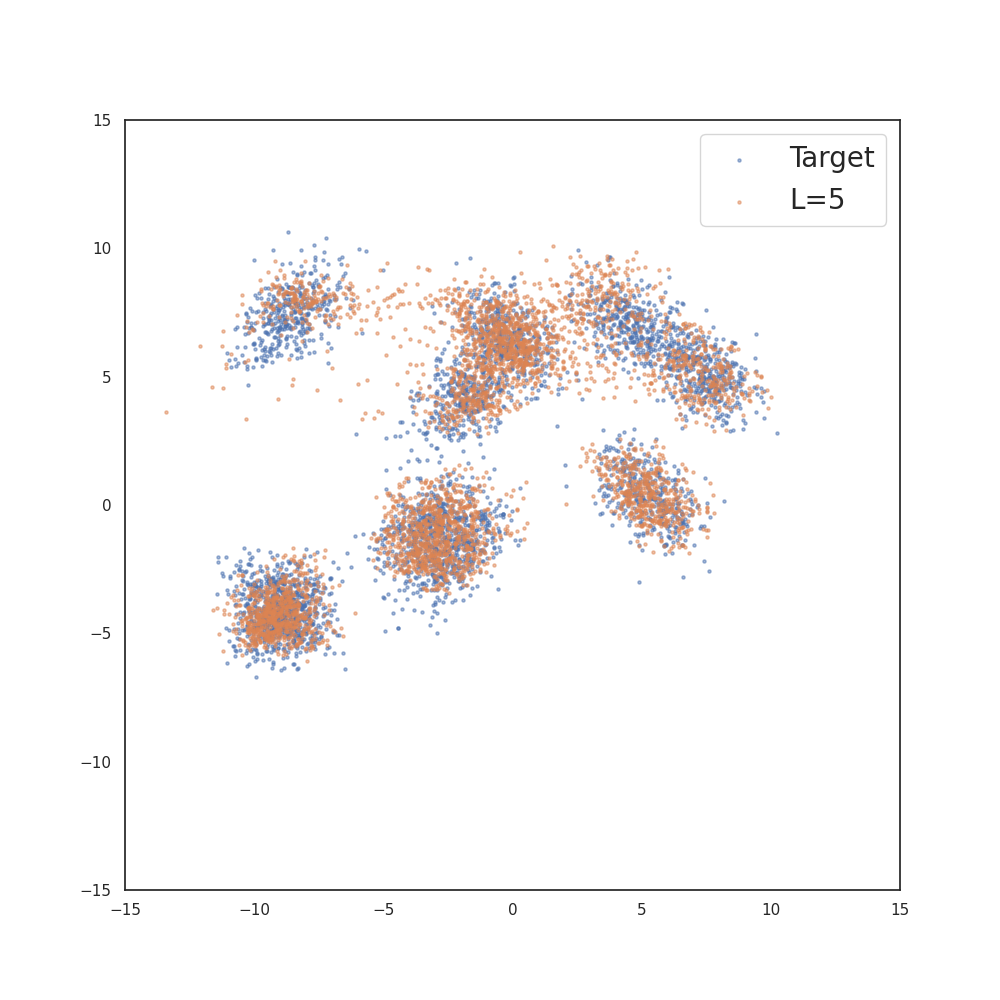}
    }
    \caption{Approximated posterior distributions with varying $L$ in the mixture of Gaussian example with $d=5$ and $K=10$. All models are trained with $15\,000$ iterations with batch size = 512.}\label{fig:posterior_vs_L}
\end{figure}

Additionally, one can select $L$ in a data-driven way in practice. For example, we check the change in training loss with increasing $L$ and we obtain a plot similar to the scree plot in PCA, shown in \Cref{fig:loss_vs_L}. The ``elbow point" in this plot is around $L=7$, which suggests that choosing $L=7$ is sufficient to capture the main structure of the posterior distribution in this example. This is consistent with our finding from the posterior visualization. From both \Cref{fig:posterior_vs_L} and \Cref{fig:loss_vs_L}, we see when we increase the number of hidden units from 16 to 32, the training loss decreases for all choices of $L$ and the posterior draws improve even when $L<K$, which indicates that increasing the complexity of each local potential function also helps improve the expressiveness of the transport map class.

\begin{figure}[!ht]
    \centering
    \includegraphics[width=0.6\textwidth]{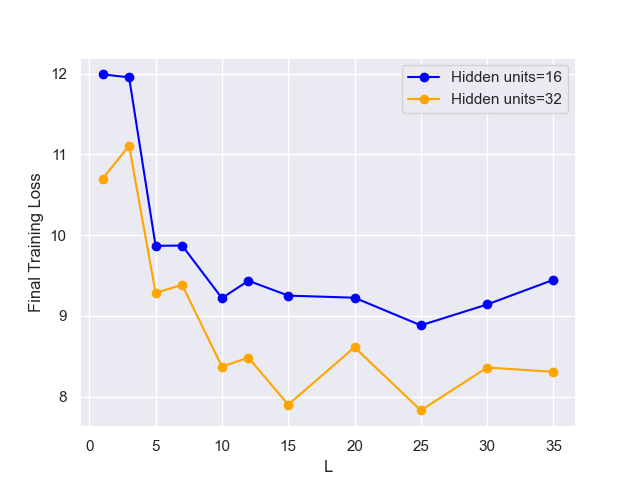}
    \caption{Training loss versus $L$ in the mixture of Gaussian example with $d=5$ and $K=10$.  All models are trained with $15\,000$ iterations with batch size = 512.}
    \label{fig:loss_vs_L}
\end{figure}

}
\end{document}